\documentclass[11pt,letterpaper]{article}
\usepackage{amsmath,amsfonts,amsthm,amssymb}
\usepackage{setspace}
\usepackage{fancyhdr}
\usepackage{lastpage}
\usepackage{extramarks}
\usepackage{xspace}
\usepackage{chngpage}
\usepackage{nicefrac}

\usepackage{tcolorbox}
\usepackage{enumerate}

\usepackage[ruled,noend,noline]{algorithm2e} \usepackage[font=small,labelfont=bf]{caption}
\usepackage[margin=1in]{geometry}
\usepackage{enumitem}
\usepackage{thmtools,thm-restate}
\usepackage{bbm}

\allowdisplaybreaks

\usepackage[dvipsnames]{xcolor}

\usepackage{wrapfig}

\title{\textbf{Online Allocation using Few Samples}}

\date{}
\author{}

\usepackage[usenames,dvipsnames]{xcolor}
\usepackage[backref=page,linktocpage=true,breaklinks,colorlinks,citecolor=PineGreen,linkcolor=Violet]{hyperref}
\usepackage[capitalise]{cleveref}

\newtheorem{Theorem}{Theorem}[section]
\newtheorem{Lemma}[Theorem]{Lemma}

\newtheorem{Definition}[Theorem]{Definition}
\newtheorem{Observation}[Theorem]{Observation}

\newtheorem{Claim}[Theorem]{Claim}
\newtheorem{Corollary}[Theorem]{Corollary}

\crefname{Theorem}{Theorem}{Theorems}
\crefname{Lemma}{Lemma}{Lemmas}
\crefname{Definition}{Definition}{Definitions}
\crefname{Claim}{Claim}{Claims}

\newcommand{\parta}{(\text{\uppercase\expandafter{\romannumeral1}})}
\newcommand{\partb}{(\text{\uppercase\expandafter{\romannumeral2}})}
\newcommand{\partc}{(\text{\uppercase\expandafter{\romannumeral3}})}
\newcommand{\be}{\mathtt{BE}}
\newcommand{\bet}{\mathtt{BE}_{-t}}

\newcommand{\opt}{\mathsf{Opt}}
\newcommand{\hopt}{\widehat{\opt}}

\newcommand{\unif}{\mathsf{Unif}}

\newcommand{\alg}{\mathsf{Alg}}
\newcommand{\algbbat}[2][\samples]{\mathsf{Alg}_{\blackbox}(#2)}
\newcommand{\algbb}{\algbbat[\samples]{p(1-2\eps)\B}}

\newcommand{\packingbb}[1][p(1+\eps)]{\mathsf{Pack}_{\blackbox}(#1 \Bpacking; \samples)}

\newcommand{\packingalg}{\mathsf{Pack}_{\alg}}
\newcommand{\coveringbb}[1][p(1-\eps)]{\mathsf{Cover}_{\blackbox}(#1 \Bcovering; \samples)}

\newcommand{\coveringalg}{\mathsf{Cover}_{\alg}}

\newcommand{\F}{\mathcal{F}}
\newcommand{\E}{\mathbb{E}}

\newcommand{\one}{\mathbf{1}}\newcommand{\1}{\one}\newcommand{\0}{\mathbf{0}}

\newcommand{\pr}{\mathbf{Pr}} 

\newcommand{\ignore}[1]{{}}
\newcommand{\R}{\mathbb{R}}
\newcommand{\calO}{\mathcal{O}}

\newcommand{\calG}{\mathcal{G}}

\newcommand{\D}{\mathcal{D}}
\newcommand{\bD}{\boldsymbol{\mathcal{D}}}

\newcommand{\eps}{\epsilon}

\newcommand{\IGNORE}[1]{}

\newcommand{\val}{v}
\newcommand{\bprice}{\pmb{\lambda}}

\newcommand{\B}{\mathbf{B}}
\newcommand{\calD}{\mathcal{D}}

\newcommand{\alloc}{\boldsymbol{a}}
\newcommand{\load}{\alloc}

\newcommand{\End}{\mathsf{end}}
\newcommand{\hist}{\mathcal{H}^{\mathsf{s}}}
\newcommand{\histend}{\mathcal{H}^{\mathsf{s}}_{\End}}

\newcommand{\blackbox}{\mathtt{RoBB}}

\newcommand{\typespace}[1][t]{\Gamma_{#1}}
\newcommand{\typespaceall}[1][i]{\Gamma}
\newcommand{\sample}{{\tilde \gamma}}
\newcommand{\online}{\gamma}
\newcommand{\packing}{\alloc^{\mathsf{p}}}
\newcommand{\covering}{\alloc^{\mathsf{c}}}
\newcommand{\Bpacking}{\B^{\mathsf{p}}}
\newcommand{\Bpackingj}{B^{\mathsf{p}}_j}
\newcommand{\Bpackingjt}{\tilde{B}^{\mathsf{p}}_j}
\newcommand{\Bcovering}{\B^{\mathsf{c}}}
\newcommand{\Bcoveringj}{B^{\mathsf{c}}_j}
\newcommand{\Bcoveringjt}{\tilde{B}^{\mathsf{c}}_j}

\newcommand{\sampleval}{{\tilde \val}}
\newcommand{\samplealloc}{{\tilde \alloc}}

\newcommand{\event}{\mathtt{Event}}
\newcommand{\bad}{\mathtt{Bad}}
\newcommand{\good}{\mathtt{Good}}

\newcommand{\decisionset}[1][t]{\Theta_{#1}}
\newcommand{\decisionsetsample}[1][t]{\tilde{\Theta}_{#1}}
\newcommand{\decisionsetall}[1][t]{\Theta}

\newcommand{\decision}[2][t]{\blackbox(#2 \mid \hist_{#1})}

\newcommand{\decisionh}[2][t]{\blackbox(#2 \mid #1)}
\newcommand{\decisionvar}{\theta}
\newcommand{\decisionsample}[1][\pi(i)]{\tilde{\decisionvar}_{#1}}
\newcommand{\decisiononline}[1][t]{\decisionvar_{#1}}

\newcommand{\decisiongeneric}[1][\pi(i)]{\hat{\decisionvar}_{#1}}
\newcommand{\tterm}{\tau_{\mathsf{term}}}

\newcommand{\xlpat}[2][t]{x_{#1,#2, \decisionvar}}
\newcommand{\xlp}[1][t]{x_{#1,\online_{#1}, \decisionvar}}

\newcommand{\xstarlpat}[2][t]{\xlpat[#1]{#2}^{*}}
\newcommand{\xip}[1][t]{z_{#1, \decisionvar}}
\newcommand{\xipdot}[1][t]{z_{#1, \cdot}}
\newcommand{\xstarip}[1][t]{\xip[#1]^{*}}
\newcommand{\xprimeip}[1][t]{\xip[#1]'}

\newcommand{\ylpbb}[2][t]{y_{#1,#2, \decisionvar}}
\newcommand{\valat}[1][t]{\val_{#1}(\decisionvar)}
\newcommand{\valsample}[1][\pi(i)]{\sampleval_{#1}(\decisionsample[#1])}
\newcommand{\valonline}[1][t]{\val_{#1}(\decisiononline[#1])}
\newcommand{\allocat}[1][t]{\alloc_{#1}(\decisionvar)}
\newcommand{\allocatj}[1][t]{a_{#1,j}(\decisionvar)}
\newcommand{\alloconline}[1][t]{\alloc_{#1}(\decisiononline[#1])}
\newcommand{\alloconlinej}[1][t]{a_{#1,j}(\decisiononline[#1])}
\newcommand{\allocsample}[1][\pi(i)]{\samplealloc_{#1}(\decisionsample[#1])}
\newcommand{\allocsamplej}[1][\pi(i)]{\tilde{a}_{#1,j}(\decisionsample[#1])}
\newcommand{\allocsampleat}[1][\pi(i)]{\samplealloc_{#1}(\decisionvar)}
\newcommand{\allocsampleatj}[1][\pi(i)]{\tilde{a}_{#1}(\decisionvar)}

\newcommand{\loadat}[1][t]{\allocat[#1]}
\newcommand{\loadatj}[1][t]{\allocatj[#1]}
\newcommand{\loadsampleatj}[1][t]{\allocsampleatj[#1]}

\newcommand{\loadatthisj}[2][t,j]{a_{#1}(#2)}

\newcommand{\loadonlinej}[1][t]{\alloconlinej[#1]}

\newcommand{\packingat}[1][t]{\packing_{#1}(\decisionvar)}
\newcommand{\packingatj}[1][t]{a^{\mathsf{p}}_{#1,j}(\decisionvar)}
\newcommand{\packingsampleat}[1][\pi(i)]{\tilde{\packing}_{#1}(\decisionvar)}
\newcommand{\packingsampleatj}[1][t]{\tilde{a}^{\mathsf{p}}_{#1,j}(\decisionvar)}

\newcommand{\packingsamplej}[1][\pi(i)]{\tilde{\packing}_{#1,j}(\decisionsample[#1])}
\newcommand{\packingonline}[1][t]{\packing_{#1}(\decisiononline[#1])}
\newcommand{\packingonlinej}[1][t]{a^{\mathsf{p}}_{#1,j}(\decisiononline[#1])}

\newcommand{\coveringat}[1][t]{\covering_{#1}(\decisionvar)}
\newcommand{\coveringatj}[1][t]{a^{\mathsf{c}}_{#1,j}(\decisionvar)}
\newcommand{\coveringsampleat}[1][\pi(i)]{\tilde{\covering}_{#1}(\decisionvar)}
\newcommand{\coveringsampleatj}[1][t]{\tilde{a}^{\mathsf{c}}_{#1,j}(\decisionvar)}

\newcommand{\coveringsamplej}[1][\pi(i)]{\tilde{\covering}_{#1,j}(\decisionsample[#1])}
\newcommand{\coveringonline}[1][t]{\covering_{#1}(\decisiononline[#1])}
\newcommand{\coveringonlinej}[1][t]{a^{\mathsf{c}}_{#1,j}(\decisiononline[#1])}

\newcommand{\instance}{\mathcal{I}}

\newcommand{\instanceonline}[1][t]{\instance^{\mathsf{r}}}

\newcommand{\LP}{\mathsf{LP}}
\newcommand{\LPrarandom}{\LP^{\mathsf{ra}}}
\newcommand{\LPrafluid}{\overline{\LP}^{\mathsf{ra}}}

\newcommand{\LPlb}{\mathsf{LP}^{\mathsf{lb}}}
\newcommand{\LPlbrelax}{\overline{\mathsf{LP}}^{\mathsf{lb}}}
\newcommand{\IPpc}{\mathsf{LP}^{\mathsf{pc}}}
\newcommand{\LPpc}{\overline{\mathsf{LP}}^{\mathsf{pc}}}

\newcommand{\sspi}{\textsf{SSPI}}
\newcommand{\ows}{$p$-\textsf{Sample}}
\newcommand{\psample}{\ows}

\newcommand{\previewat}[1][p]{$#1$-\textsf{Preview}}
\newcommand{\previewatmath}[1][p]{#1-\textsf{Preview}}
\newcommand{\preview}{\previewat}
\newcommand{\ro}{\textsf{RO}}
\newcommand{\onlHistIndText}{online-history independent}
\newcommand{\OnlHistIndText}{Online-History Independent}
\newcommand{\OnlHistInd}{\textsf{OnlineHI}}
\newcommand{\offHist}{offline sample-history}

\newcommand{\requests}{[n]}

\newcommand{\samples}{\mathcal{S}}

\newcommand{\samplesnot}{\mathcal{R}}
\newcommand{\samplesnotp}{\hat{\samplesnot}}
\newcommand{\nsamples}{s}

\newcommand{\hpi}{\hat{\pi}}

\newcommand{\optR}[2][\requests]{\opt_{#1}(#2)}
\newcommand{\optRlb}[1][\requests]{\opt_{#1}}
\newcommand{\optRlbfl}[1][\requests]{\overline{\opt}_{#1}}
\newcommand{\optfl}[1][\B]{\overline{\opt}(#1)}
\newcommand{\zlp}[1][t]{z_{#1, \decisionvar}}
\newcommand{\zhlp}[1][t]{\hat{z}_{#1, \decisionvar}}
\newcommand{\zstarlp}[1][t]{z^*_{#1, \decisionvar}}

\newcommand{\goodsampling}{\mathcal{E}_{\mathsf{s}}}
\newcommand{\goodload}[1][b]{\mathcal{E}_{\mathsf{\ell}}(#1)}
\newcommand{\goodloadcomp}[1][b]{\mathcal{E}^c_{\mathsf{\ell}}(#1)}

\begin{document}

\author{Matthew Faw\footnote{School of Computer Science / H. Milton Stewart School of Industrial and Systems Engineering / Algorithms and Randomness Center, Georgia Institute of Technology, Atlanta, USA. Email: mfaw3@gatech.edu. Supported in part by NSF award CCF-2440113.} \and  Sahil Singla\footnote{School of Computer Science, Georgia Institute of Technology, Atlanta, GA, USA. Email: ssingla@gatech.edu. Supported in part by NSF awards CCF-2327010 and CCF-2440113.}  \and  Yifan Wang\footnote{School of Computer Science, Georgia Institute of Technology, Atlanta, GA, USA. Email: ywang3782@gatech.edu. Supported in part by NSF awards CCF-2327010 and CCF-2440113.}}

\maketitle

\begin{abstract}
\smallskip

We study online allocation problems where $n$ requests over $m$ resources arrive in an adversarial order and must be served immediately and irrevocably. This framework captures both Online Resource Allocation, where the goal is to maximize value subject to resource budgets, and Online Load Balancing, where the goal is to minimize the makespan. We seek $(1\pm\epsilon)$-competitive algorithms in the large-budget or large-makespan regime.

\smallskip

 We consider a sampling model that generalizes the following two well-studied sampling models. In the Single-Sample Prophet Inequality ($\mathsf{SSPI}$) model, request $t$ is drawn from an unknown distribution $\mathcal{D}_t$, and the algorithm is given one independent sample from each $\mathcal{D}_t$ before the online phase. In the $p$-$\mathsf{Sample}$ model, the requests are adversarial, but a uniformly random $p$-fraction is revealed upfront as training data.

\smallskip

Although near-optimal algorithms are known in the easier random-order model ($\mathsf{RO}$), where the requests arrive in a uniformly random order, prior algorithms for $\mathsf{SSPI}$ and $p$-$\mathsf{Sample}$ were problem-specific and incurred substantially worse dependencies on $\epsilon$, $m$, and $n$. Our main contribution is a general framework that converts $\mathsf{RO}$ algorithms into algorithms for the $p$-$\mathsf{Preview}$ model, a model that generalizes both  $\mathsf{SSPI}$ and $p$-$\mathsf{Sample}$. As consequences, we obtain near-optimal bounds for Online Resource Allocation, generalized Online Load Balancing, and online mixed packing-covering problems in these adversarial-order sampling models, significantly improving the bounds of [Ghuge, Singla, Wang (STOC'25)] and [Gupta and Molinaro (SODA'26)].
\end{abstract}

\vspace{-1em}
{\small
\setcounter{tocdepth}{1}
\tableofcontents
}

\clearpage

\section{Introduction}
Online allocation is a fundamental problem in computer science and operations research, with applications including online advertising, combinatorial auctions, routing, and online packing/covering LPs. In this problem, a sequence of $n$ requests arrives over time and must be served immediately using $m$ resources; different ways of serving the same request may consume different amounts of the resources.
On the maximization side, the canonical example is \emph{Online Resource Allocation}, where each allocation option for a request has a corresponding \emph{value}, and the objective is to maximize the \emph{welfare}, i.e., the total value of the selected allocations, while never exceeding the allocation \emph{budget} of any of the $m$ resources. On the minimization side, the canonical example is \emph{Online Load Balancing},
where each request must be assigned to one of several feasible options, each inducing load on one or more resources. The objective there is to minimize the \emph{makespan}, i.e., the maximum load on any resource after serving all requests.

We aim to design $(1\pm\eps)$-competitive online allocation algorithms:  for Resource Allocation a $(1-\eps)$-approximation to the offline optimum value, and for Load Balancing, a $(1+\eps)$-approximation to the offline makespan.
Without any assumption on the arriving requests, however, such goals are unachievable. Indeed, even in standard special cases, Online Resource Allocation admits only an $\calO(\nicefrac{1}{n})$ worst-case competitive ratio, while Online Load Balancing has an $\Omega(\log(m))$ worst-case lower bound. Consequently, to obtain near-optimal guarantees, one must assume that the algorithm has some form of knowledge about the input. 

If the online arrivals are in a uniformly random order (secretary model), then an impressive line of work \cite{DevenurHayes-EC09,AWY-OR14,MR-MOR14,devanur2019near,KRTV-SICOMP18,agrawal2014fast,GM-MOR16}  obtains  $(1\pm\eps)$-competitive algorithms under natural \emph{large-budget} conditions; see also book chapters \cite[Chapter 6]{EIV-Book23} and \cite{GS-Book20} for background on large-budget Online Resource Allocation. In the general $m$-resource setting, for Online Resource Allocation each resource should have budget $\Omega\big(\eps^{-2} \log(m)\big)$, and for Online Load Balancing the optimal makespan should be $\Omega\big(\eps^{-2} \log({m}/{\delta})\big)$ when the algorithm fails with probability at most  $\delta$.
However, these algorithms rely on the uniformity of the arrival order. In an adversarial order, the early arrivals may be statistically unrelated to the later arrivals.
Hence, to handle non-uniform arrivals, the following two natural online  models have been studied.

\medskip
\noindent \textbf{(i) Prophet model.}
In this setting, each request $\online_t$ is drawn independently from a distribution $\calD_t$, where different distributions need not be identical.
With known distributions, solving the corresponding ex-ante/configuration LP relaxation and independently rounding its fractional solution yields $(1\pm\eps)$-competitive algorithms for both online resource allocation and online load balancing under the appropriate large-budget condition.
In many real applications, however, assuming complete knowledge of the distributions is unrealistic: a decision maker can only hope to learn about them from previous interactions or limited historical data.

This motivated Ghuge, Singla, and Wang \cite{GSW-STOC25} to consider Online Resource Allocation in the \emph{single-sample prophet inequality} model (\sspi{}). 
In this model, an adversary chooses unknown distributions $\D_1,\ldots,\D_n$; the algorithm receives one independent sample $\sample_t \sim \D_t$ for every $t\in[n]$, then must irrevocably serve a fresh request $\online_t\sim \D_t$ in the online phase.
\cite{GSW-STOC25} show that a single sample from each distribution already suffices to obtain a $(1-\eps)$-approximation for non-identical distributions with budgets $B_j \ge \Omega\big(\eps^{-6} \cdot\log^4(nm/\eps)\big)$.
Compared to the LP-rounding and random-order benchmarks, however, this leaves a polynomial gap in the dependence on $1/\eps$ and $\log m$, and also has an additional dependence on the time horizon $n$. 
Moreover, their techniques are tailored to the value maximization objective and do not directly yield analogous guarantees for minimization problems such as Load Balancing or Mixed Packing-Covering.

\paragraph{(ii) \ows{} model.} 
In this setting, the requests $\online_t$ are chosen adversarially, but a uniformly random $p$-fraction of the requests is revealed upfront to the online algorithm. Then the $n$ requests arrive online in the original adversarial order. 
This model originated in the study of online matching and secretary problems to overcome strong worst-case impossibility results \cite{KaplanNR20,KaplanNR22}. Argue, Frieze, Gupta, and Seiler \cite{ArgueF0S22} first studied this model for the restricted-assignment special case of online load balancing, obtaining a $(1+\epsilon)$-competitive algorithm when optimal makespan is  $\Omega\big(\eps^{-5} p^{-1} \cdot \log^3(m/\eps) \big)$.

More recently, Gupta and Molinaro \cite{GuptaM26} considered the general load balancing problem and its extension to multiple-choice mixed packing-covering in the \ows{} model. They obtain a  $(1+\epsilon)$-competitive algorithm when optimal makespan is  $\Omega\big(\eps^{-8} p^{-1} \cdot \log(n m K / \eps)^5\big)$, where $K$ is the number of processing options for each job. 
Compared to the LP-rounding and random-order benchmarks, this again leaves a polynomial gap in the dependence on $1/\eps$ and $\log m$, and also has an additional dependence on the time horizon $n$ and the number of options $K$. Moreover, their framework does not give the value-maximization  guarantees of online resource-allocation. 

The gaps present in these results motivate the following question addressed in this work:
\begin{quote}
    \emph{Can we obtain $(1\pm\eps)$-competitive online allocation algorithms in the \sspi{} and \ows{} models nearly-matching the large-budget requirements in the known-distribution LP-rounding and random-order settings?}
\end{quote}
\vspace{-3ex}

\subsection{Main Results}

In this paper, we resolve this question for the canonical maximization and minimization problems above.
Our first main result is for the single-sample model:

\begin{Theorem}[Main results for \sspi{} model]\label{thm:sspiInformal}
For every $\eps\in(0,1)$ and $\delta \in (0,\eps)$, the following guarantees hold in the \sspi{} model:
\begin{enumerate}
    \item For online resource allocation with $m$ resources, there is an algorithm whose expected welfare is at least a $(1-\calO(\eps))$-fraction of the expected offline optimum, provided that the budget of every resource $j$ satisfies $B_j = \Omega\big(\eps^{-2}\cdot \log({m}/{\eps}) \big)$.
    
    \item For online generalized load balancing on $m$ machines, there is an algorithm that, with probability at least $1-\delta$, outputs an allocation of makespan at most $(1+\calO(\eps))\opt$, where $\opt$ denotes the expected offline optimal makespan, provided that $\opt = \Omega\big(\eps^{-2}\cdot \log({m}/{\delta}) \big)$.
\end{enumerate}
\end{Theorem}

For online resource allocation, this matches the guarantee available under known distributions and in the random-order model, up to the logarithmic dependence on $\eps$. In fact, 
$\Omega(\eps^{-2}\log m)$ budget is information-theoretically necessary in general, even for known input distributions \cite{devanur2019near,GSW-STOC25}.
\Cref{thm:sspiInformal} improves the previous single-sample online resource-allocation guarantee of \cite{GSW-STOC25}, whose budget requirement was larger by a polynomial factor in $1/\eps$ and $\log m$, and also depended on time horizon $n$.
On the minimization side, the theorem gives analogous single-sample guarantees for generalized load balancing, a setting not covered by \cite{GSW-STOC25}.

Our second main result is for the \ows{} model.

\begin{Theorem}[Main Results for \ows{}]\label{thm:owsInformal}
For every $\eps\in(0,1)$, $\delta\in(0,\eps)$, and $p\in(0,1]$, the following guarantees hold in the $\ows{}$ model:

\begin{enumerate}
    \item For online resource allocation with $m$ resources, there is an algorithm whose expected welfare is at least a $(1-\calO(\eps))$-fraction of the expected offline optimum on the $n$ requests, provided that every resource budget satisfies $B_j = \Omega(p^{-1}\eps^{-2}\cdot\log({m}/{\eps}))$.

    \item For online generalized load balancing on $m$ machines, there is an online algorithm that, with probability at least $1-\delta$, outputs an allocation with makespan at most $(1+\calO(\eps))\opt$, where $\opt$ denotes the expected offline optimal makespan on the $n$ requests, provided that $\opt = \Omega(p^{-1}\eps^{-2} \cdot \log({m}/{\delta}))$. 
\end{enumerate}
\end{Theorem}

As a complementary set of results, we establish matching lower bounds for the main guarantees above.
In \Cref{sec:lower-bounds-new}, we also show that these bounds are information theoretically tight (up to logarithmic dependence on $\eps$) in $p, \eps, \log m$ for both online resource allocation and online load balancing. \Cref{thm:owsInformal} improves the prior \ows{} guarantees of \cite{GuptaM26}: for online generalized load balancing they required the optimal load to be polynomially larger in $\eps$ and $\log m$, and
we remove the dependence on both the time horizon $n$ and on the number of processing options $K$. On the maximization side, this gives the first near-optimal \ows{} guarantee for online resource allocation.

\medskip
\noindent \textbf{Pricing implications.} 
For many applications of online resource allocation, we desire \emph{item-pricing} algorithms: before each request arrives, the algorithm posts a price for each resource, and the request  takes the decision that maximizes the difference between its value and its price.
Additionally, in such  applications, the algorithm may not even have full samples from request distributions, and can only learn about a request through an item-pricing query as above.
In \Cref{subsec:PricingQuery}, we strengthen the online resource allocation results for the \sspi{} and \ows{} models stated above: when a rough estimator of $\opt$ is given, our algorithms can replace full observation of each request available in the offline learning phase with the outcome of a single pricing query to that request, while still achieving the same $(1-\eps)$-competitiveness guarantees as in our main results.

Taken together, our results in the \sspi{} and \ows{} models show that a small amount of offline information is sufficient to recover essentially the best guarantees known in the random-order model.

\subsection{Our Techniques}\label{sec:mainTechniques}

For simplicity, we first restrict our discussion to the \sspi{} model, where each online request $\online_t$ is drawn independently from $\calD_t$. 

\medskip
\noindent \textbf{\OnlHistIndText{} algorithms.} 
The starting point of our techniques is that ``\onlHistIndText{}'' allocation algorithms exist  with near-optimal performance: We call a policy \emph{\onlHistIndText{}} (\OnlHistInd{}) if, after the learning phase, its decision on request $t$ (before imposing any feasibility constraints) depends only on the current request $\online_t$ and on the \emph{\offHist{}}, but not on the \emph{online} history.
When request distributions are \emph{given}, $(1\pm \eps)$-competitive \OnlHistInd{} algorithms are known to exist: 
solve the LP relaxation offline and then use the fractional solution to independently round the action for request $\online_t$; feasibility is then enforced by the usual large-budget concentration/capping argument. This observation motivates the question: \emph{Given only one sample from each distribution, can we learn near-optimal \OnlHistInd{} allocation algorithms?}

To design such an \OnlHistInd{} allocation algorithm, we begin with a thought experiment. 
Suppose that, instead of coming in adversarial order, the requests came in a \emph{uniformly random order}.  As discussed earlier, online allocation in the random-order (denoted \ro{}) model is fairly well understood, and we know $(1\pm \eps)$-competitive algorithms under large-budget conditions \cite{KRTV-SICOMP18,GM-MOR16,agrawal2014fast}. 
A crucial feature of these \ro{} algorithms is that they do not require distributional knowledge: the multiset of requests may be adversarially chosen, and only its order is random.

\medskip
\noindent \textbf{Simulating \ro{} for an \OnlHistInd{} algorithm.}
To obtain an \OnlHistInd{} allocation algorithm from the single sample, our idea is to simulate a random order on the samples, and, when a request $\online_t$ arrives, simply make the allocation that the \ro{} algorithm would have made on the samples, if $\online_t$ replaced its sample $\sample_t$. More concretely, 
let $\{\sample_t\}_{t\in[n]}$ be the samples observed offline by the algorithm, and let $\pi$ be a uniformly random order of these samples. 
Now, if we run a \ro{} algorithm $\blackbox$ on the samples in order $\pi$, we get a sequence of sample requests $\sample_{\pi(i)}$ and corresponding decisions $\decisionsample[\pi(i)]$ made by $\blackbox$.
Let $\hist_{\pi(i)} = (\sample_{\pi(j)}, \decisionsample[\pi(j)])_{j<i}$ be the \offHist{} before processing sample $\sample_{\pi(i)}$ in this random order. Then we can express the decision of the \ro{} algorithm on sample $\pi(i)$ as:
\[
    \decisionsample = \decision[\pi(i)]{\sample_{\pi(i)}}
\]
Our proposal is to mimic this decision online with $\blackbox$ by using the \offHist{} in place of the online history. In particular, when request $\online_t$ arrives during the online phase, we ignore the online history and replace it with the \offHist{}  $\hist_t$ before the sample $\sample_t$ is processed in random order. Thus, we take the decision that $\blackbox$ \emph{would have made} on the sample, had $\online_t$ appeared offline instead of $\sample_t$:
\[
    \decisiononline = \decision[t]{\online_{t}}.
\]
Since $\hist_t$ is determined before $\sample_t$ is revealed in the simulated random order, it is independent of both $\sample_t$ and the fresh draw $\online_t$. 
Thus, conditioned on $\hist_t$, the sample decision $\decisionsample[t]$ and the online decision $\decisiononline$ have the same distribution.
Hence, the decisions on the sampled requests and the decisions on the actual requests form \textbf{tangent} sequences with respect to the learning histories, in the sense of \cite{KwapienW89,DelapenaG12}. This tangent relation is the bridge between the \ro{} execution on samples and the adversarial-order execution on the true requests.

Using this bridge, we can view the learning-phase process and the allocation-phase process as two martingales with matched conditional increments. 
Consequently, Freedman/Bernstein-type martingale inequalities imply that, for the consumption of a resource or the load placed on a machine, the online allocation phase is close to the learning phase.
This argument transfers bounded linear statistics, such as resource consumptions and machine loads, from the simulated \ro{} execution to the online execution.
For online load balancing, the idea is essentially sufficient to obtain our guarantees, since the load of each request on a resource is uniformly bounded. 

\medskip
\noindent \textbf{Handling correlations between value and budget exhaustion.}
For online resource allocation, however, the value associated with each request may be unbounded. Thus, the martingale concentration is not enough to transfer the \ro{} guarantees to this setting. 
The main challenge is to show that the event that our algorithm over-allocates is not (significantly) positively correlated with the value collected.
The tangent relation guarantees two useful facts: 
\begin{itemize}[noitemsep]
\item With high probability, the uncapped online execution does not consume too much of any resource.
\item The expected value of the uncapped online execution is exactly the same as the expected value of the \ro{} execution on the samples.
\end{itemize}
Thus, if resource consumption and value were independent, a simple argument would suffice: a small probability of running out of budget would immediately imply a proportional value loss.
However, in resource allocation these two quantities are correlated. The histories on which the algorithm is close to exhausting a resource may also be precisely the histories on which the algorithm obtains large value.

Our key additional idea is to control resource consumption on the bad histories rather than value directly. Using Cauchy--Schwarz and a second-moment bound on the deviation between the online allocation and the \ro{} allocation on the sample, we show that, for each resource $j$, the expected consumption on bad histories is at most
    $\calO\!\left(\sqrt{\pr[\text{bad}]}\, B_j\right)$.
Thus, dividing the bad-history decisions by this factor yields a feasible solution to the resource-allocation LP,
guaranteeing their value is bounded by the same factor times $\opt$.
By ensuring the concentration failure probability is $O(\eps^2)$, the loss then becomes $\calO(\eps\opt)$.  Thus, this second-moment argument is sufficient to preserve, in the online execution, a $(1-\eps)$-fraction of the \ro{}'s expected value on the sample.

\medskip
\noindent \textbf{Extensions to the \ows{} and \preview{} models.}
We establish our results in a setting we call the \preview{} model (\Cref{def:preview}), which generalizes the \sspi{} and \ows{} settings. Here, an adversary chooses $n$ request distributions $\D_1,\ldots,\D_n$, and requests are sampled independently $\online_t \sim \D_t$. Offline, a single sample $\sample_t \sim \D_t$ from a subset of the request distributions $t\in \samples$ are revealed, where $\samples$ is a uniformly-random $p$-sample of $[n].$

In the \preview{} model, the algorithm described above is no longer well-defined, since only the requests at times $t\in\samples$ have a sample pair; the remaining $(1-p)$-fraction have no sample information. However, there is a simple modification for requests $t\not\in\samples$: when request $\online_t$ arrives, we insert it into a uniformly random position of the sampled order and take the decision that $\blackbox$ would have made on $\online_t$ in this counterfactual execution.
Notice that we can express this decision as:
\[
    \decisiononline = \decision[T_t]{\online_t},
\]
where $T_t$ is the uniformly random position of $\online_t$ in this counterfactual sampled order.

To understand why this algorithm performs well in the allocation phase, it is helpful to compare it with a hypothetical \ro{} algorithm run on all $n$ requests. 
For a request processed in the online phase, the policy it would face in this hypothetical \ro{} execution is learned from a uniformly random prefix of the other requests.
In our \ows{} algorithm, the policy faced by the same request is learned from a random prefix of the sampled sequence. This prefix is shorter, because it comes only from the $p$-fraction of samples $\samples$ observed offline, but, up to this scaling,
 it is still a uniformly random subset of the relevant prefix of the instance. 
 Accordingly, we simulate the \ro{} algorithm on the sample with budgets and load thresholds scaled by $\alpha_p=\nicefrac{|\samples|}{n}=p$, the ratio of the number of samples to requests. The same tangent-sequence and martingale-transfer argument then compares the scaled sample execution to the online execution, yielding the same guarantees with the natural additional factor $\alpha_p$ in the required budget or makespan.

\subsection{Further Related Work}

\noindent \textbf{Sample-based online resource allocation.} 
A large body of work has studied sample-based online resource allocation problems. Early work of Azar, Kleinberg, and Weinberg \cite{AzarKW14} initiated the study of prophet inequalities with limited information, and subsequent work obtained sharp guarantees from samples in single-choice and related settings \cite{RubinsteinWW20,CorreaDFS19,correa2024sample,CristiZilliotto24}. 
Particularly relevant to our work is the line of algorithms based on greedy prices or greedy-ordered selection \cite{KorulaP09,CDFFLLP-SODA22,FOCS24-toapper}. These algorithms run an offline greedy procedure on samples, or on a sample/statistic mixture, and couple this greedy execution to the online process. Our reduction follows a similar high-level philosophy, but with a different offline process: rather than running a greedy algorithm, we run a random-order online algorithm on the samples and transfer its internal states to the adversarial-order online phase.

For the online resource allocation model studied in our paper, \cite{GSW-STOC25} first studied almost optimal online resource allocation through samples for heterogeneous agents. Recently, \cite{TW-26} studied the sample complexity of the dual pricing algorithm for the same online resource allocation model. Another line of research considers the online resource allocation problem with regret minimization objectives: \cite{BKMSW-ICML23} analyzed the single resource setting and provided a near-optimal algorithm which uses only a single sample. \cite{JLZ-MS25} investigated the generalization to multiple resources, which also focuses on how the Wasserstein distance between the learned empirical distribution and the true distribution impacts the regret.

\medskip
\noindent \textbf{Learning from samples in online algorithms.}
The \ows{} model and related sample-augmented online models were introduced to interpolate between worst-case online algorithms and random-order/stochastic models \cite{KaplanNR20,KaplanNR22,ArgueF0S22,gupta2024set}. In particular,  \cite{gupta2024set} also design a $p$-sample algorithm  using a random order algorithm as a subroutine; however, online set cover behaves very differently from online allocation, where our analysis crucially exploits properties of tangent sequences.
More broadly, our work is also related to data-driven algorithm design, where one uses historical instances to learn an algorithmic policy or tune an online algorithm \cite{Balcan-Chp20,balcan2021much}. A related template appears in work on online learning via offline algorithms \cite{DBLP:journals/corr/GolovinKS14,DBLP:conf/colt/RoughgardenW18,NiazadehGWSB21}, where structural properties of an offline algorithm, often greedy, are used to obtain online guarantees from samples. Our contribution is to show that random-order online algorithms themselves can serve as the offline learning primitive.

\medskip
\noindent \textbf{Online allocation and pricing.}
Online resource allocation has also been studied extensively in stochastic, adversarial, robust, and revenue-management models, often via primal-dual, re-solving, or pricing-based policies \cite{MSVV-JACM07,Mehta-Book13,agrawal2014fast,devanur2019near,balseiro2020dual,jasin2012resolving,jasin2014reoptimization,bumpensanti2020resolving,balseiro2023best,BGSZ-ITCS20,AGMS-SODA22}. In contrast to this literature, our main focus is on learning from extremely limited offline information---one sample per distribution in \sspi{} or a small random subset in \ows{}---while matching the large-budget guarantees of random-order algorithms.

\section{Model and Preliminaries}
\label{sec:model}

We are interested in Online Allocation problems both from the maximization and minimization perspectives. The following definition introduces the model we consider throughout the paper.

\begin{Definition}[Online Allocation problem]\label{def:onlineAlloc}
    An online allocation problem instance consists of a finite sequence of requests $(\online_t)_{t\in \requests}$ for $m$ resources. When request $\online_t \in \typespace[t]$ arrives, it must be satisfied immediately and irrevocably by a decision in $\decisionset$ (let $\decisionsetall$ and $\typespaceall$ denote the set of all possible decisions and requests, respectively). Selecting a decision $\decisiononline \in \decisionset$ generates an allocation $\alloconline \in [0,1]^m$ to potentially many of the $m$ resources. 
\end{Definition}
The Online Allocation problems we consider have additionally a problem-specific \emph{objective} to either maximize or minimize a function which depends on the allocation decisions. Additionally, in some settings, there may be \emph{budget constraints} on the set of feasible allocations. Throughout, the benchmark against which we compare an algorithm's performance is the \emph{hindsight optimal feasible allocation}. Our paper focuses on the following three canonical Online Allocation problems:

\begin{itemize}
    \item \textbf{Online Resource Allocation}: Here, requests are of the form $\online_t = (\val_t, \alloc_t, \decisionset)$. Selecting a decision $\decisiononline \in \decisionset$ generates value $\valonline \in \R_{\geq 0}$ and consumes $\alloconline \in [0,1]^m$ units of resources. We assume there is always a null action $\phi \in \decisionset$ which contributes no value and consumes no resources, i.e., $\val_t(\phi) = 0$ and $\alloc_t(\phi) = \0$. The $m$ resources have allocation budget constraints $\B = (B_j)_{j\in[m]}$, and the goal is to maximize the value collected while not exceeding the allocation budget $B_j$ on any resource $j\in[m]$.

    \item \textbf{Online Generalized Load Balancing}: Here, the online requests are of the form $\online_t = (\alloc_t, \decisionset)$. Selecting a decision $\decisiononline \in \decisionset$ generates an allocation load $\alloconline \in [0,1]^m$ on a subset of the $m$ resources. The goal is to minimize the \emph{makespan}, i.e., the maximum load on a resource after all requests have been satisfied.
    \item \textbf{Online Packing-Covering Multiple-Choice LPs}: This setting is (essentially) a generalization of the Online Generalized Load Balancing problem, where choosing a decision $\decisiononline \in \decisionset$ generates load $\alloconline = (\packingonline, \coveringonline)$, where $\packingonline, \coveringonline \in [0,1]^m$ correspond to the packing and covering loads, respectively. The objective is to (approximately) satisfy the packing and covering constraints $\Bpacking, \Bcovering\in \R_{\geq 0}^m$, i.e., $\sum_{t\in\requests} \packingonline \leq \Bpacking$ and $\sum_{t\in\requests} \coveringonline \geq \Bcovering$.
\end{itemize}

We study these online allocation problems in the following \preview{} setting. Here, each request is sampled independently from a distribution $\online_t \sim \D_t$. We consider the challenging setting where only a single sample from a $p$-fraction of the distributions is available offline:

\begin{Definition}[\preview{} model]\label{def:preview}
    In the \preview{} setting, parametrized by $p \in (0, 1]$, an adversary chooses a request distribution $\D_t$ for each $t \in [n]$. These requests are initially unknown to the learner. Offline, a subset $\samples \subseteq [n]$ of size $\nsamples = p\cdot n$ is sampled uniformly without replacement, and for each $t\in\samples$, an independent sample $\sample_t \sim \D_{t}$ is revealed. Then, online, independent requests $\online_t \sim \D_t$ for each $t\in[n]$ are revealed sequentially in order.
\end{Definition}

We note that the \previewat[1] captures the well-studied \emph{Single-Sample Prophet Inequality} setting, where a single sample from \emph{every} distribution is available offline.

\begin{Definition}[\sspi{} model]\label{def:sspi}
The single-sample prophet inequality (denoted \sspi{}) setting corresponds to the $\previewat[1]$ setting, where a single sample from every distribution $(\sample_1,\ldots,\sample_n) \sim \bD = \D_1 \times \ldots \times \D_n$ is revealed offline, then, online, requests $\online_t \sim \D_t$ are sampled independently then revealed in $t=1,\ldots,n$.
\end{Definition}

The \preview{} setting also captures the online \ows{} setting, where each request is chosen adversarially, and a $p$-fraction of the request sequence is revealed offline:

\begin{Definition}[\ows{} model]\label{def:ows}
The \ows{} setting, parametrized by $p \in (0,1]$, corresponds to the \preview{} setting, where each distribution $\D_t$ is a point-mass distribution corresponding to the adversarially chosen request $\online_t$. Offline, a subset of $\nsamples = p \cdot n$ uniformly random requests sampled without replacement is revealed. Then, online, all $[n]$ requests are revealed sequentially in their original order.\footnote{Our results can also be extended with minor modifications to the related model where only the unsampled indices $\samplesnot = [n]\setminus \samples$ arrive online, and $\opt$ depends only on the unsampled requests $\samplesnot = [n] \setminus \samples$.}
\end{Definition}

To obtain results in the \preview{} model, we will rely heavily on algorithms from the random-order model in which the request sequence is chosen adversarially, and the requests are presented in a uniformly random order:

\begin{Definition}[\ro{} model]\label{def:ro}
In the Random Order (denoted \ro{}) setting, an adversary chooses a request sequence $\online_1,\ldots,\online_n$. Then, online, the requests are presented in a uniformly random order.
\end{Definition}

Compared to the \preview{} setting, due to the random order of arrivals, no samples are needed to obtain $(1\pm\eps)$-competitive online algorithms for the \ro{} setting, since in expectation, the first $\eps$-fraction of requests contributes at
most an $\eps$-fraction of the $\opt$ solution. 
In particular, for Online Resource Allocation in the \ro{} model, \cite{KRTV-SICOMP18} gave an optimal LP-resolving based algorithm for the large budget setting of $B_j = \Omega(\eps^{-2}\cdot\log(m))$.
For Online Generalized Load Balancing in \ro{}, \cite{GM-MOR16} gave a $(1+\eps)$-competitive algorithm that succeeds with probability at least $1-\delta$ when the optimal makespan is $\Omega(\eps^{-2}\cdot\log(\nicefrac{m}{\delta}))$. Moreover, they also gave an algorithm for Online Packing-Covering Multiple-Choice LPs that satisfies the packing/covering constraints up to $(1\pm \eps)$ factors, assuming the constraints are $\Omega(\eps^{-2}\cdot\log(\nicefrac{m}{\delta}))$.

Our results for the \preview{} setting will use these \ro{} algorithms in a black-box manner. We use $\blackbox$ to denote a generic \ro{} algorithm which maps a request $\online$ with decision set $\decisionsetall$ and an interaction history $\hist$ of prior requests and decisions to a decision $\decisionvar$:
\begin{align*}
    \decisionvar = \decisionh[\hist]{\online} \in \decisionsetall.
\end{align*}
By definition, if the Online Allocation problem has constraints, the decision $\decisionvar$ must be feasible given the prior decisions in $\hist$.  We note that the decision $\decisionvar$ of $\blackbox$ may depend additionally on independent random coins in addition to the request $\online$ and $\histend$. For instance, the $\blackbox$ may first compute a fractional decision based on $\online$ and $\histend$, and then apply randomized rounding to convert it to an integral decision in $\decisionsetall$. For simplicity, we suppress the dependence on these independent coins in our notation.

\section{Online-History Independence via Random Order Algorithms}
\label{sec:generalFramework}

Recall from \Cref{sec:mainTechniques} that our goal is to design an \OnlHistInd{} algorithm, i.e., one whose decision, ignoring feasibility, depends only on the current request $\online_t$ and on the samples $(\sample_i)_{i\in\samples}$.
As discussed in \Cref{sec:model}, under the appropriate large-budget conditions, there exist $(1\pm\eps)$-competitive \ro{} algorithms for the online allocation problems studied in this paper.
Our plan is to use these \ro{} algorithms in a black-box manner to design \OnlHistInd{} algorithms with nearly matching competitive guarantees in the \preview{} setting. In particular, we aim to design online algorithms satisfying the following property:
\begin{Definition}[\OnlHistIndText{} Algorithm]\label{def:OnlHistInd}
An online allocation algorithm for the \preview{} setting is \emph{\onlHistIndText{}} ($\OnlHistInd{}$) if its tentative decision $\decisiononline$ (i.e., the decision before imposing feasibility) on request $\online_t$ is independent of all other requests $\online_{t'}$ and tentative decisions $\decisiononline[t']$, $t'\neq t$, conditioned on the samples $(\sample_i)_{i\in\samples}$ and any additional offline randomness.
\end{Definition}

Our main algorithmic idea to obtain $\OnlHistInd{}$ algorithms is to use a good \ro{} algorithm, denoted $\blackbox$, together with the sample to make decisions.
We first describe our approach in the \previewat[1] (i.e., \sspi{}) setting, since our algorithms are more intuitive there.
Here, each request $\online_t \sim \D_t$ has a \emph{paired, independent offline sample}, $\sample_t \sim \D_t$. 
We proceed as follows: (i)~First, sample a uniformly random ordering $\pi$ on the sample and run $\blackbox$ on the sample in this random order. Before processing sample $\pi(i)$, record the \offHist{} $\hist_{\pi(i)}$ generated by the previous samples and associated offline decisions. (ii)~Then, for online request $\online_t$, replace the online history with $\hist_t$, the \offHist{} immediately before processing $\sample_t$ in random order, and take the decision of $\blackbox$ on $\online_t$ given this history $\hist_t$.
In this way, we naturally obtain a \onlHistIndText{} algorithm.
Further, and crucially, the online tentative decision of \cref{alg:genericOnline} on request $\online_t$ is \emph{identically distributed} to the \ro{} algorithm's offline decision on sample $\sample_t$, conditioned on \offHist{} $\hist_t$.

More specifically, we convert a good \ro{} algorithm $\blackbox$ to an \OnlHistInd{} algorithm through the following two-phased procedure:

\paragraph{(1) Offline \ro{} simulation.}
First, before any online requests arrive, we select a uniformly random order $\pi$ of the samples $(\sample_i)_{i\in\samples}$, then run the $\blackbox$ algorithm on the sample in order $\pi$.
Before processing sample $\sample_{\pi(i)}$, we record the observation history $\hist_{\pi(i)}$ (i.e., the sample requests and decisions before processing $\sample_{\pi(i)}$). Since we will use $\blackbox$ algorithms that are $(1\pm\eps)$-competitive for \ro{}, the offline execution on the randomly ordered sample inherits the corresponding \ro{} guarantee. The goal is to use the decision policies $\decision[\pi(i)]{\cdot}$ obtained in this offline phase to obtain a good online algorithm. The details of this offline phase are described in \Cref{alg:genericOffline}.

\begin{algorithm}[H]
\caption{\textsc{Generating Offline History via \ro{} Simulation}}
\label{alg:genericOffline}
\KwIn{Black-box algorithm $\blackbox$, $\nsamples$ sampled indices $\samples \subseteq [n]$, samples $(\sample_t)_{t\in\samples}$}
Draw a permutation $\pi$ of $\samples$ uniformly at random, where $\pi(i)\in\samples$ denotes index of the $i$th sample processed in random order for each $i\in[\nsamples]$. \\
\For{$i = 1, \ldots, \nsamples$}
{
    Let $\hist_{\pi(i)} = (\sample_{\pi(j)}, \decisionsample[\pi(j)])_{j\in[i-1]}$ be the \offHist{} before processing $\sample_{\pi(i)}$.\\
Record $\decisionsample = \decision[\pi(i)]{\sample_{\pi(i)}}$, the $\blackbox$ decision on sample $\sample_{\pi(i)}$ in random order.\\
    Update the state of $\blackbox$ with $(\sample_{\pi(i)},\decisionsample[\pi(i)])$.\\
}
Let $\histend := (\pi(j), \sample_{\pi(j)}, \decisionsample[\pi(j)])_{j\in[\nsamples]}$ be the final \offHist{}.\\
\KwOut{$(\pi(i), \hist_{\pi(i)})_{i\in [\nsamples]}$}
\end{algorithm}

\paragraph{(2) Replacing online history with offline history for \OnlHistInd{} allocation.}
After the offline phase, we process online requests in adversarial order. 
To make an $\OnlHistInd{}$ allocation for request $\online_t$, we use $\blackbox$ to make tentative decisions as if the online history were replaced by the \offHist{} and the current request $\online_t$ appeared in the offline random-order. For requests $\online_t$ with a corresponding sample (i.e., $t\in\samples$), this corresponds to the tentative decisions if $\online_t$ appeared in the offline sample instead of $\sample_t$, i.e., $\decisiononline = \decision[t]{\online_t}$. For the remaining requests $t\not\in\samples$, this corresponds to the tentative decision if $\online_t$ had appeared at a uniformly random position in the sample ordering, i.e., $\decisiononline = \decision[T_t]{\online_t}$, where $T_t \sim \unif(\samples)$ are sampled independently. For feasibility-constrained problems, if the tentative decision is ever infeasible (e.g., in the resource allocation setting, if the online allocation ever exceeds budget $B_j$ for some resource $j\in[m]$), then the algorithm makes the null decision $\phi$ for every remaining time-step. The details of this procedure are given in \cref{alg:genericOnline}.

\begin{algorithm}[H]
\caption{\textsc{$\OnlHistInd{}$ Allocation via Offline History}}
\label{alg:genericOnline}
\KwIn{time horizon $n \geq 1$, $\nsamples$ sampled indices $\samples\subseteq [n]$, samples $(\sample_t)_{t\in\samples}$}
Run \Cref{alg:genericOffline} on samples $(\sample_t)_{t\in\samples}$ to obtain the sampling histories $(\pi(i), \hist_{\pi(i)})_{i\in[\nsamples]}$\\
\For{each online request $t\in \requests$}
{
Observe request $\online_t$ and record
        $\decisiononline = \begin{cases}
            \decision[t]{\online_t} & \text{if $t\in\samples$}\\
            \decision[T_t]{\online_t} & \text{if $t\not\in\samples$, where $T_t \stackrel{iid}{\sim} \unif(\samples)$}\\
        \end{cases}$\\
\uIf{$\alloconline$ is feasible given the prior allocations $(\alloconline[t'])_{t'<t}$}{
            Make decision $\decisiononline$.
    }\Else{
        Mark $\tterm = t$ and make the null decision $\phi$ for this and all future time-steps.
}
}
\end{algorithm}

To gain some intuition for why this approach might work, consider the online resource allocation problem in the \previewat[1] (i.e., \sspi{}) setting where $\samples = [n]$. Suppose we run \Cref{alg:genericOnline} with a black-box random-order algorithm $\blackbox$ with budget $(1-\calO(\eps))\B$. 

As noted above, conditioned on $\hist_t$, the offline decision
$\decisionsample[t]=\decision[t]{\sample_t}$ and the uncapped online proposal $\decisiononline[t]=\decision[t]{\online_t}$ are identically distributed, because $\sample_t,\online_t{\sim}\D_t$ and are independent of $\hist_t$. 
In particular, this implies that (i) the allocations $\allocsample[t]$ and $\alloconline$, as well as (ii) the corresponding values $\valsample[t]$ and $\valonline$, are identically distributed. Thus, if \Cref{alg:genericOnline} is never over-allocated, then the expected value collected by \Cref{alg:genericOnline}, $\E[\alg]$, would be the same as that collected by the $\blackbox$ algorithm, $\E[\algbbat[{[n]}]{(1-\calO(\eps))\B}]$, in \Cref{alg:genericOffline} (we establish the formal relation for the general \preview{} setting in Property 1 of \Cref{lem:algConnectionProperties} below). In particular, we would have that
\[
\E[\alg] = \E[\algbbat[{[n]}]{(1-\calO(\eps))\B}] \geq (1-\eps)\opt.
\]

In the \preview{} setting with $p\in(0,1)$, the samples and online requests are not paired: indeed, there are fewer samples than requests. Nevertheless, the revealed sample is a uniformly random subset of the full adversarial instance. This exchangeability, together with the scaling factor $p$, plays the role of the sample-online symmetry. This symmetry allows us to treat the \preview{} setting with similar ideas as the \sspi{} setting.

Note that the above bound holds with an extra assumption that \Cref{alg:genericOnline} is never over-allocated. To bypass this assumption, we face two main challenges:
\begin{enumerate}
    \item[C1.] \Cref{alg:genericOnline} may over-allocate online even though \Cref{alg:genericOffline} never does on the sample (by definition of $\blackbox$). 
    \item[C2.] In the Online Resource Allocation setting, the event that \Cref{alg:genericOnline} over-allocates could be positively correlated with the collected value, so that even if the probability of over-allocation were small, the expected value collected during over-allocation could be large.
\end{enumerate}

To resolve these two challenges, we bridge the performance of \cref{alg:genericOnline} on the online requests with the performance of the $\blackbox$ in \cref{alg:genericOffline} on the offline samples. The first step in connecting the performance of these two algorithms is observing that, by design, \cref{alg:genericOnline} is $\OnlHistInd{}$:

\begin{restatable}[\Cref{alg:genericOnline} is \onlHistIndText{}]{Observation}{restateObservationOblivious}\label{obs:oblivious}
In the \preview{} setting, \cref{alg:genericOnline} is $\OnlHistInd{}$. In particular, for every $t\in[n]$, the request $\online_t$ and tentative decision $\decisiononline$ are independent of $(\online_{t'}, \decisiononline[t'])_{t'\neq t}$, conditioned on the \offHist{}, $\histend$.
\end{restatable}
\begin{proof}
    Note that $\decisiononline$ depends only on the current request $\online_t$, $\hist_t$ (if $t\in\samples$) or $\hist_{T_t}$ where $T_t\sim \unif(\samples)$ (if $t\not\in\samples$), and independent randomness. Thus, since the $(\online_t)_{t\in[n]}$ are jointly independent and independent of $(\sample_t)_{t\in\samples}$, and hence also $\histend$ in the \preview{} setting, and $(T_t)_{t\in[n]}$ are jointly independent conditionally on $\histend$, $(\online_t, \decisiononline)$ is independent of $(\online_s, \decisiononline[s])_{s\neq t}$ given $\histend$.
\end{proof}
Notice that, as a consequence of \cref{obs:oblivious} and Bernstein's inequality, the tentative allocations $\sum_{t\in[n]}\alloconline$ of \cref{alg:genericOnline} concentrate about its expectation conditioned on the \offHist{} $\histend$.

To complete the connection between the online decisions of \cref{alg:genericOnline} and the offline decisions of \cref{alg:genericOffline}, we relate the expected performance of the tentative decisions of \cref{alg:genericOnline} with the sample decisions of \cref{alg:genericOffline}. 
In the \previewat[1] (i.e., \sspi{}) setting, this connection follows from the observation that the tentative online decisions of \cref{alg:genericOnline}, $\decisiononline$, are \emph{tangent} to the offline sample decisions of \cref{alg:genericOffline}, $\decisionsample[t]$. That is, conditioned on the \offHist{} $\hist_t$ before processing sample $\sample_t$, the decisions $\decisiononline$ and $\decisionsample[t]$ are \emph{identically distributed}. Moreover, in the general \preview{} setting, we can establish a similar relation between the decisions of \cref{alg:genericOffline,alg:genericOnline}. The crucial consequences of this tangent relation are summarized in \cref{lem:algConnectionProperties}:
equality of expectations, a second-moment bound used for the value-on-bad-histories argument, and a high-probability deviation bound used to transfer resource consumptions and machine loads. We prove this result in \Cref{sec:appendix:tangentSeqFacts}.

\begin{restatable}{Lemma}{restateAlgConnectionProperties}\label{lem:algConnectionProperties}
    Take $\decisionsample[\pi(i)]$ to be the decision of \Cref{alg:genericOffline} on sample $\sample_{\pi(i)}$ in random order $\pi(1),\ldots,\pi(\nsamples)$, and $\decisiononline[t]$ the decision of \Cref{alg:genericOnline} on request $\online_t$.
    Then, in the \preview{} setting, we have
    \begin{enumerate}
        \item \emph{Expectation:} For any measurable function $f : \decisionsetall \times \typespaceall \to \R$,
        \begin{align*}
            \textstyle p\, \E[\sum_{t\in \requests} f(\decisiononline, \online_t)] = \E[\sum_{i\in[\nsamples]} f(\decisionsample, \sample_{\pi(i)})]
        \end{align*}
        \item \emph{Second-moment:} For any bounded measurable function $f : \decisionsetall \times \typespaceall \to [0,1]$,
        \begin{align*}
            \textstyle\E[(p \sum_{t\in \requests} \E[f(\decisiononline, \online_t) \mid \histend] - \sum_{i\in[\nsamples]} f(\decisionsample, \sample_{\pi(i)}))^2]
            \leq 2\E[\sum_{i\in[\nsamples]} f(\decisionsample, \sample_{\pi(i)})]
        \end{align*}
    
        \item \emph{High-probability deviation:} For any bounded measurable function $f : \decisionsetall \times \typespaceall \to [0,1]$ and $\eta, \delta \in (0,1)$, with probability at least $1-2\delta$,
        \begin{align*}
            \textstyle \left |p \sum_{t\in \requests} \E[f(\decisiononline, \online_t) \mid \histend] - \sum_{i\in[\nsamples]} f(\decisionsample, \sample_{\pi(i)}) \right |
            \leq   \frac{5\log(1/\delta)}{\eta} + \eta\sum_{i\in[\nsamples]} f(\decisionsample, \sample_{\pi(i)})
        \end{align*}
\end{enumerate}
\end{restatable}

Let us now briefly discuss how to use \Cref{lem:algConnectionProperties} to handle the Challenges mentioned above.
For Challenge 1, we apply Property 3 of \Cref{lem:algConnectionProperties} to the allocation function $\alloconlinej \in [0,1]$. This shows that the conditional expected online consumption of each resource is close to its consumption in the offline sample run. Then, by the $\OnlHistInd{}$ concentration inequality from \cref{obs:oblivious}, the realized online consumption is close to this conditional expectation.
Since $\blackbox$ never over-allocates (by definition), we can conclude that \Cref{alg:genericOnline} will not either with high probability.

For Challenge 2, using an LP relaxation characterization of the optimal allocation, we show that the expected \emph{value} collected on bad histories is negligible as long as the \emph{resource consumption} on those histories is small in expectation.
To bound this quantity, we use Property 2 of \Cref{lem:algConnectionProperties}, which shows that the second moment of the difference between the online and offline allocations is bounded by the expected offline allocation.
This second-moment bound, combined with Cauchy--Schwarz, bounds the value carried by bad histories and resolves Challenge 2.

\section{Applications to Resource Allocation}
\label{sec:resourceAlloc}

We now apply the framework from \Cref{sec:generalFramework} to  Online Resource Allocation. Recall that, in this setting, the seller receives $n$ online requests from buyers arriving in arbitrary order, for $m$ kinds of limited resources with budgets $\B=(B_1,B_2,\cdots,B_m)\in \R^m_{\geq 0}$. The request from each buyer $t \in \requests$ is denoted by $\online_t = (\val_t, \alloc_t, \decisionset[t])$, 
where choosing a decision $\decisiononline \in \decisionset$ generates value $\valonline \in\R_{\geq 0}$ for buyer $t$, while consuming at most one unit of each resource $\alloconline=(a_{t,1}(\decisiononline),\cdots,a_{t,m}(\decisiononline))\in[0,1]^m$.
Without loss of generality assume $B_j\leq n$ for each resource $j$. We always assume that there is a null decision $\phi\in\decisionset$ for each buyer, such that $\val_t(\phi)=0$ and $\alloc_t(\phi)=\mathbf{0}$.
The objective of the seller is to maximize social welfare 
$\sum_{t\in\requests}\valonline[t]$, while not exceeding the budget of any resource.

Our benchmark is the expected \emph{hindsight fractional optimum} $\E[\optR{\B}]$, where $\optR{\B}$ is the optimal value of the configuration linear program defined below on the realized requests. For each $t\in\requests$, let $\zlp[t]$ denote the fractional allocation corresponding to the decision $\decisionvar \in \decisionset$ for request $\online_t$. Then we can write the configuration LP as follows:

\begin{equation}\tag{$\LPrarandom$}\label{program:ora-lp}
    \begin{aligned}
        \optR{\B} ~:=~ &\text{maximize}  &&\textstyle\sum_{t\in\requests} \sum_{\decisionvar\in\decisionset[t]} \valat[t] \zlp[t]\\
        &\text{s.t.}  &&\textstyle\sum_{t\in\requests} \sum_{\decisionvar\in\decisionset[t]} \allocat[t] \zlp[t] \leq \B\\
        &\forall t\in \requests,  &&\textstyle \sum_{\decisionvar\in\decisionset[t]} \zlp[t] \leq 1\\
        &\forall t\in \requests, \decisionvar\in\decisionset[t],  && \zlp[t] \geq 0
    \end{aligned}
\end{equation}

Our main result of this section is that, given any good black-box \ro{} algorithm $\blackbox$, the value collected by \Cref{alg:genericOnline} in the \preview{} setting is a $(1-\eps)$-approximation to 
the appropriately scaled value of $\blackbox$ in the large-budget setting.

\begin{Theorem}
\label{thm:resourceAlloc}
    Suppose we run any black-box \ro{} algorithm $\blackbox$ for Online Resource Allocation offline in \Cref{alg:genericOffline} with budget $p (1-2\eps) \B$ for some $p \in (0,1]$ and $\eps \in (0,\nicefrac{1}{2}]$ on samples $\sample_1,\ldots,\sample_{\nsamples}$, where the budgets satisfy 
$B_j = \Omega(\log(\nicefrac{m}{\eps})\cdot \eps^{-2} p^{-1})$
    for all $j\in[m]$. Let $\algbb$ denote the value collected by $\blackbox$ on the sample. Then, if we run \Cref{alg:genericOnline} in the \preview{} setting, and let $\alg$ denote the value collected by \Cref{alg:genericOnline} online,
    \[
        \E[\alg] 
        ~\geq~ (1-\eps)p^{-1}\cdot\E[\algbb] - 4\eps \E[\opt_{\requests}(\B)] .
    \]
\end{Theorem}

To obtain our online resource allocation results for the \preview{} setting, we combine this theorem with near-optimal \ro{} online algorithms in the large budget setting.

\begin{Lemma}[\cite{KRTV-SICOMP18}] \label{lem:RO_ORA}
     There exists an online resource allocation algorithm\footnotemark such that executing it on samples $(\sample_i)_{i\in\samples}$ in \ro{}  with budget $\alpha \B$, where $\alpha \in (0,1]$ and the budgets satisfy $\alpha B_j = \Omega(\log(\nicefrac{m}{\eps}) \cdot \eps^{-2})$ for all $j\in[m]$, returns an integral online allocation 
with expected value 
    \[
        \E[\algbbat[\samples]{\alpha \B}\mid (\sample_t)_{t\in\samples}]
        \geq
        (1-\eps)\optR[\samples]{\alpha\B},
    \]
    where the expectation is over the random arrival order and any possible randomness of the algorithm.
\end{Lemma}
\footnotetext{We note that \cite{agrawal2014fast,GM-MOR16,GSW-STOC25} also provide 
algorithms for online resource allocation in the \ro{} model under these budget assumptions. However, as written, they require either an estimate of the optimal allocation, or that each value is at most $\calO(\eps^2 \opt/\log(m))$. These conditions are not required in \cite{KRTV-SICOMP18}.}

We also need the well-known rescaling property of $\optR{\B}$; e.g., it appeared in \cite[Lemma 1]{KRTV-SICOMP18} under slightly different conditions. For completeness, we give its proof in \Cref{sec:appendix:oraDeferred}.

\begin{restatable}{Lemma}{restateLpRescaling}\label{lem:lpRescaling}
Let $\samples \subseteq [n]$ be a uniformly random subset of indices of size $s = p n$. Suppose that $B_j \geq p^{-1}\eps^{-2}(2 + 5 \log(\nicefrac{2m}{\eps}))$. Then,
  \[
        \E[\optR[\samples]{p(1-2\eps) \B}] \geq p (1-5\eps) \cdot \E[\optR{\B}].
  \]
\end{restatable}

Now combining \Cref{thm:resourceAlloc}  with the last two lemmas, 
we immediately obtain the desired guarantees for resource allocation in the \preview{} setting, and hence also the \sspi{} (\Cref{thm:sspiInformal}) and \ows{} (\Cref{thm:owsInformal}) settings.

\begin{Corollary}\label{obs:goodRO-ORA}
    If we run \Cref{alg:genericOnline} with $\blackbox$ being the \ro{} algorithm in \Cref{lem:RO_ORA}, instantiated with budget $p(1-2\eps)\B$ with each $B_j = \Omega(p^{-1}\eps^{-2}(1 + \log(\nicefrac{m}{\eps})))$, then
    \[
        \E[\alg] 
        ~\geq~ (1-\eps)^2 p^{-1}\cdot\E[\optR[\samples]{p(1-2\eps)\B}] - 4\eps \E[\optR{\B}] 
        ~\geq~ (1-\calO(\eps))\cdot\E[\optR{\B}].
    \]
\end{Corollary}

In the rest of the section, we focus on proving \Cref{thm:resourceAlloc} . The plan is to implement the strategy described in \Cref{sec:generalFramework}. 
Indeed, if there were no budget constraints, then Property 1 of \Cref{lem:algConnectionProperties} implies that the expected uncapped value collected online is the same as the expected value $\E[\algbb]$ collected by the \ro{} algorithm on the sample, up to the $p^{-1}$ scaling:
\begin{align*}
    \textstyle \E[\sum_{t\in \requests} \valonline]
    ~=~ p^{-1}\cdot\E[\sum_{i\in [\nsamples]} \valsample]
    ~=~ p^{-1}\cdot\E[\algbb].
\end{align*}

As a consequence, our high-level approach to prove the theorem is to (i) show that the probability of over-allocation is negligible (\Cref{lem:concentrationBudgetExhaustion}), and then (ii) show that the over-allocation event is not significantly positively correlated with the value collected online (\Cref{lem:ScaledAlgValue}). The heart of the proof lies in controlling this positive correlation.

\subsection{Proof of \cref{thm:resourceAlloc}}
As discussed above, if \Cref{alg:genericOnline} never violated the budget constraints, then the expected value collected by the algorithm would equal the value collected by the $\blackbox$ algorithm on the sample, scaled by $p^{-1}$. Thus, to reason about the value collected by the algorithm, we need to show the value collected under \emph{budget exhaustion} is not too large. We formally define budget exhaustion.

\begin{Definition}[Budget exhaustion]\label{def:budgetExhaustion}
    The \emph{budget exhaustion} event $\be$ is the event that one of the resources runs out of budget in the online execution of \Cref{alg:genericOnline}, i.e.,
    \begin{align*}
        \textstyle\be = \{\exists j\in[m] : \sum_{t\in\requests} \alloconlinej > B_j\}.
    \end{align*}
\end{Definition}

Notice that, by Property 1 of \Cref{lem:algConnectionProperties}, the value collected by \Cref{alg:genericOnline} satisfies

\begin{equation}\label{eq:valDecomp}
\begin{aligned}
    \E[\alg] 
    &\geq \textstyle\E[\sum_{t\in\requests}\valonline] - \E[\sum_{t\in \requests} \valonline \1[\be]]\\
    &= \textstyle p^{-1}\E[\algbb] - \E[\sum_{t\in \requests} \valonline \1[\be]].
\end{aligned}
\end{equation}
Thus, our objective is to upper bound the expected value on the budget exhaustion event $\be$.

Notice that $\be$ depends both on the offline samples and online requests. Instead of reasoning about $\be$ directly, we will instead reason about a related, simpler $\bad$ ``(pseudo) over-allocation'' event (and its complement, the $\good$ event) which depends only on the samples. 
\begin{Definition}[$\good$ (Pseudo) allocation event]\label{def:goodEvent}
    Define $\good$ to be the event such that the expected allocation of \Cref{alg:genericOnline} conditioned on the sampling history $\histend$ of \Cref{alg:genericOffline} is bounded, and $\bad$ to be the complement of $\good$, i.e.,
    \begin{align*}
        \textstyle\good = \left\{\sum_{t \in [n]} \E[\alloconline[t] \mid \histend] \leq (1 - \eps) \B\right\}
        \quad\text{and}\quad
        \bad = \good^c.
    \end{align*}
\end{Definition}

We use the $\good$ and $\bad$ events to partition the value collected under budget exhaustion $\be$ in two parts:
\begin{equation}\label{eq:valEarlyStopDecomp}
\begin{aligned}
    \textstyle\E[\sum_{t\in \requests} \valonline \1[\be]]
    &\leq\textstyle \E[\sum_{t\in \requests} \valonline \1[\bad]]
    + \E[\sum_{t\in \requests} \valonline \1[\good \cap \be]]\\
    &=\textstyle \E[\sum_{t\in \requests} \valonline \1[\bad]]
    + \E[\sum_{t\in \requests} \valonline \pr[\good \cap \be \mid \histend, \online_t, \decisiononline]],
\end{aligned}
\end{equation}
where the second equality follows since $\valonline$ is deterministic conditioned on the sampling history $\histend$ and the current request $\online_t$ and decision $\decisiononline$. We proceed by bounding the two terms in \eqref{eq:valEarlyStopDecomp} separately, where bounding the first term will be the main challenge of the proof.

\paragraph{Bounding Value collected on the $\bad$ event}
We focus first on bounding the first term of \eqref{eq:valEarlyStopDecomp}, which is the main challenge of the proof. We prove the following result in \cref{sec:scaledAvgValue}, and discuss the main steps below.
\begin{Lemma}
\label{lem:ScaledAlgValue}
    Let $\bad$ be the bad (pseudo) over-allocation event defined in \Cref{def:goodEvent}. Then,
    \[
    \textstyle\E [\1[\bad] \cdot \sum_{t\in \requests} \valonline] ~\leq~ 4 \eps \cdot \E[\opt_{\requests}(\B)].
    \]
\end{Lemma}

The main difficulty in establishing this claim comes from the fact that, although $\bad$ happens with small probability (an immediate consequence of Property 3 of \cref{lem:algConnectionProperties}), $\valonline$ could be arbitrarily large relative to $\E[\optR{\B}]$ and positively correlated with $\bad$. To handle this possibility, we begin by showing that the expected (tentative) \emph{allocation} on the $\bad$ event must be small. We show this by noting: (i) by construction, the $\blackbox$ running in \cref{alg:genericOffline} allocates at most $p(1-2\eps)\B$ units of resource on the sample, and (ii) the difference in allocations of the offline and online decisions in \cref{alg:genericOffline,alg:genericOnline} can be bounded using Cauchy-Schwarz together with the tangent connection between the online and offline decisions of \cref{alg:genericOffline,alg:genericOnline} (specifically, the second-moment bound in \cref{lem:algConnectionProperties}). We then translate this bound on the expected \emph{allocation} on the $\bad$ event to a bound on the expected \emph{value} on the bad event using the standard fluid characterization of $\E[\optR{\B}]$. The full proof is given in \cref{sec:scaledAvgValue}.

\paragraph{Bounding the Probability of Budget Exhaustion.}
To bound the second term in \eqref{eq:valEarlyStopDecomp}, it suffices to show that the probability of budget exhaustion $\be$ on the $\good$ event is small, even after conditioning on the sampling history $\histend$ and the current request $\online_t$.
\begin{Lemma}\label{lem:concentrationBudgetExhaustion}
    Recall the $\good$ (pseudo) allocation event from \Cref{def:goodEvent} and the budget exhaustion event $\be$ from \Cref{def:budgetExhaustion}. Then, for any run of \Cref{alg:genericOffline} generating a sampling history $\histend$, and any online request $\online_t$, $t\in\requests$,
    \begin{align*}
        \pr[\good \cap \be \mid \histend, \online_t, \decisiononline]
        \leq \eps.
    \end{align*}
\end{Lemma}
The main idea of the proof of \Cref{lem:concentrationBudgetExhaustion} is the following: By definition, whenever $\good$ and $\be$ are true, then the following two conditions are satisfied:
\begin{align*}
    \textstyle\sum_{t\in \requests} \E[\alloconline[t] \mid \histend] \leq (1-\eps) \B
    \quad\text{and}\quad
    \sum_{t\in \requests} \alloconlinej > B_j.
\end{align*}
Now, since \Cref{alg:genericOnline} is \OnlHistInd{} by \Cref{obs:oblivious}, the tentative decisions $\{\decisiononline\}_{t\in[n]}$ are independent conditioned on $\histend$. Thus, since $(\alloconline[t'])_{t\in\requests}$ are bounded in $[0,1]$, by Bernstein's inequality, $\sum_{t\in\requests}\alloconline[t]$ concentrates about $\sum_{t\in[n]}\E[\alloconline \mid \histend]$, conditioned on any \offHist{} $\histend$. Moreover, since the allocations are independent and bounded for all $t\in[n]$, conditioning on any single request $\online_t$ and decision $\decisiononline$ does not significantly change this concentration bound. As a consequence, we can conclude that $\good \cap \be$ must happen with small probability.
For the full proof of \Cref{lem:concentrationBudgetExhaustion}, refer to \Cref{sec:resourceAlloc-missing}.

\paragraph{Putting everything together.}

Using the above results, \Cref{thm:resourceAlloc} follows readily.
    Indeed, as we observed in \eqref{eq:valDecomp} and \eqref{eq:valEarlyStopDecomp}, the value collected by \Cref{alg:genericOnline} can be lower bounded as follows:
    \begin{align*}
        \textstyle\E[\alg] 
        &\textstyle\geq p^{-1}\E[\algbb] \\
        &\quad\textstyle- \E[\sum_{t\in \requests} \valonline \1[\bad]]
        - \E[\sum_{t\in \requests} \valonline \pr[\good \cap \be \mid \histend, \online_t, \decisiononline]],
    \end{align*}
    We use \cref{lem:ScaledAlgValue} to bound the second term above, and
    for the final term, we apply the conditional probability bound from \Cref{lem:concentrationBudgetExhaustion}, so that
    \begin{align*}
        \textstyle\E[\sum_{t\in \requests}\valonline[t]\pr[\good \cap \be \mid \histend, \online_t, \decisiononline]] 
        \leq \textstyle \eps \E[\sum_{t\in \requests} \valonline[t]]
        = \textstyle \eps p^{-1} \E[\algbb],
    \end{align*}
    where the last equality follows from Property 1 of \Cref{lem:algConnectionProperties}.
    Taken together, these bounds imply that
    \[
        \textstyle\E[\alg] 
        \geq (1-\eps)p^{-1}\E[\algbb] 
        - 4\eps\E[\optR{\B}],
    \]
    as claimed.

\subsection{Bounding Value on the $\bad$ event: Proof of \Cref{lem:ScaledAlgValue}}
\label{sec:scaledAvgValue}

We now turn our attention to establishing \cref{lem:ScaledAlgValue}, the main technical challenge in proving \cref{thm:resourceAlloc}. We begin by establishing that the probability of $\bad$ is negligible. Then, we show that the expected tentative \emph{allocations} on the $\bad$ event are at most $\calO(\eps) \B$. Using these results together with a fluid characterization of $\E[\optR{\B}]$ allows us to establish the result.

\subsubsection{Bounding the probability of $\bad$}
The first ingredient in proving \Cref{lem:ScaledAlgValue} is the following result, which shows that the probability of the $\bad$ event must be small. 
\begin{Lemma}
\label{lma:bad-tiny}
    Let $\bad$ be the event defined in \Cref{def:goodEvent}.
    Then, 
$\pr[\bad] \leq \epsilon^2$.
\end{Lemma}
\begin{proof}
    Since $\alloconlinej,\allocsamplej \in [0,1]$ for all $j\in[m]$, taking a union bound over the $m$ resources and applying \Cref{lem:algConnectionProperties} with $\delta\leftarrow \nicefrac{\delta}{m}$ implies that, with probability at least $1-2\delta$, for any $\eta\in(0,1)$,
    \[
        \textstyle p \sum_{t \in \requests} \E[\alloconline \mid \histend] 
        \leq \frac{5\log(\nicefrac{m}{\delta})}{\eta} + (1+\eta) \sum_{i \in [\nsamples]} \allocsample.
    \]
    Since $\blackbox$ allocates at most $p (1-2\eps)\B$ on the sample, we therefore have that
    \begin{align*}
        \textstyle \sum_{t \in \requests} \E[\alloconline[t] \mid \histend] 
        \leq \frac{5\log(\nicefrac{m}{\delta})}{p\eta} + (1-2\eps + \eta)\B.
    \end{align*}
    Thus, choosing $\eta = \nicefrac{\eps}{2}$ and $\delta = \nicefrac{\eps^2}{2}$, and using the fact that $B_j \geq \frac{40\log(\nicefrac{2 m}{\eps})}{p\eps^2}$, with probability at least $1-\eps^2$,
    \begin{align*}
        \textstyle \sum_{t \in \requests} \E[\alloconline \mid \histend] 
        \leq \frac{20\log(\nicefrac{2 m}{\eps})}{p\eps} + (1-2\eps+\nicefrac{\eps}{2})\B 
        \leq  \nicefrac{\eps}{2}\B + (1-2\eps+\nicefrac{\eps}{2})\B
        = (1-\eps)\B.
    \end{align*}
    Thus, by \Cref{def:goodEvent}, the above implies that $\pr[\bad] \leq \eps^2$, as claimed.
\end{proof}

\subsubsection{Bounding the expected allocation on the $\bad$ event}

Note that \Cref{lma:bad-tiny} alone is not sufficient to establish \Cref{lem:ScaledAlgValue}. This is because $\valonline$ can be arbitrarily large relative to $\E[\optR{\B}]$ 
and is correlated with the $\bad$ event.
Indeed, consider a problematic scenario when the value collected by \Cref{alg:genericOnline} is positively correlated with the $\bad$ (pseudo) over-allocation event. It could be the case that, even though $\pr[\bad] \leq \eps^2$, the value collected on this event could be very large, e.g., $\E[\1[\bad] \cdot \sum_{t\in \requests} \valonline] \gg \pr[\bad] \E[\sum_{t\in \requests} \valonline]$. 

Fortunately, we are able to show that this is not the case. The  following is the key result that shows that the expected \emph{allocation} on the $\bad$ event must be small. The main idea here is to (i) use the fact that, by definition, the $\blackbox$ running in \Cref{alg:genericOffline} allocates at most $p (1-2\eps)\B$ units of resource, and (ii) apply Cauchy-Schwarz together with the second-moment bound from Property 2 of \Cref{lem:algConnectionProperties} on the difference between the allocations made online and on the offline sample.
\begin{Lemma}\label{lem:allocOnEvent}
    Let $\event$ be any event depending only on the \offHist{} $\histend$. Then, the allocations made by \Cref{alg:genericOnline} satisfy the following guarantees:
    \begin{align*}
        \textstyle\E[\1[\event]\sum_{t\in\requests} \alloconline]
        \leq 2\sqrt{\pr[\event]}\B.
    \end{align*}
    In particular, since $\bad$ depends only on $\histend$, we have 
$\textstyle\E[\1[\bad]\sum_{t\in\requests} \alloconline]
        \leq 2\eps\B.$
\end{Lemma}
\begin{proof}
    For any resource $j\in[m]$, we can decompose the allocations made by \Cref{alg:genericOnline} for $j$ as follows: Using the fact that the request indices $\requests$ and $\event$ are deterministic given $\histend$,
    \begin{align*}
        \textstyle\E[\1[\event] \sum_{t\in\requests}\alloconlinej[t]]
        &= \textstyle\E[\1[\event] \sum_{t\in \requests}\E[\alloconlinej[t] \mid \histend]]\\
        &=\textstyle p^{-1}\E[\1[\event] \sum_{i\in[\nsamples]}\allocsamplej]\\
        &\quad+\textstyle p^{-1}\E[\1[\event] (p\sum_{t\in \requests}\E[\alloconlinej[t] \mid \histend] - \sum_{i\in[\nsamples]}\allocsamplej)]
    \end{align*}
    The first term in this decomposition depends on the allocations of $\blackbox$ on the sample, while the second depends on the difference between the allocations on the sample and those online.
    Now, since $\blackbox$ allocates at most $p(1-2\eps)\B$ units of resources deterministically on the sample by definition,
    \[
        \textstyle p^{-1}\E[\1[\event]\sum_{i\in[\nsamples]} \allocsamplej]
        ~\leq~ \pr[\event](1-2\eps)B_j .
    \]
    For the remaining term, we apply Cauchy-Schwarz, so that
    \begin{align*}
        &\textstyle\E[\1[\event] (p\sum_{t\in\requests}\E[\alloconlinej[t] \mid \histend] - \sum_{i\in[\nsamples]}\allocsamplej)]\\
        &\leq \textstyle\sqrt{\pr[\event] \E[(p\sum_{t\in \requests}\E[\alloconlinej[t] \mid \histend] - \sum_{i\in[\nsamples]}\allocsamplej)^2]}.
    \end{align*}
    Now, since $\allocsamplej, \alloconlinej \in [0,1]$, we may apply Property 2 in \Cref{lem:algConnectionProperties}, and again using the fact that $\blackbox$ allocates at most $p (1-2\eps)\B$ units of resource,
    \begin{align*}
        \textstyle\E[(p\sum_{t\in [n]}\E[\alloconlinej[t] \mid \histend] - \sum_{i\in[\nsamples]}\allocsamplej)^2]
        \leq 2\E[\sum_{i\in[\nsamples]}\allocsamplej]
        \leq 2 p(1-2\eps) B_j.
    \end{align*}
    Hence, combining these three inequalities, and 
    using the large-budget assumption, in particular $B_j\geq 2/p$ for all $j$, we conclude, denoting $q=\pr[\event]$:
    \begin{align*}
        \textstyle\E[\1[\event]\sum_{t\in\requests}\alloconlinej]
        \leq q B_j + p^{-1}\sqrt{q \cdot 2 p (1-2\eps)B_j}
        = \sqrt{q} B_j\left(\sqrt{q} + \sqrt{\frac{2(1-2\eps)}{p B_j}}\right)
        \leq 2\sqrt{q} B_j,
    \end{align*}
    as claimed.
\end{proof}

\subsubsection{Bounding the expected value on $\bad$ via a fluid LP characterization of OPT}

To conclude our proof of \Cref{lem:ScaledAlgValue}, we prove the following stronger claim: fix any $\event$ which depends only on the offline learning phase of \Cref{alg:genericOffline}, i.e., $\histend$. Then, we will show that
\begin{align}\label{eq:valueOnEvent}
\textstyle\E [\1[\event] \cdot \sum_{t\in \requests} \valonline] \leq 4 \sqrt{\pr[\event]} \cdot \E[\optR{\B}].
\end{align}
Notice that, since $\bad$ depends only on $\histend$ and $\pr[\bad]\leq \eps^2$ by \Cref{lma:bad-tiny}, the above inequality implies the claimed result. 

The main challenge in proving \eqref{eq:valueOnEvent} is that the values $\valonline$ can be arbitrarily large, and both $\event$ and $\valonline$ depend on the sampling history $\histend$. We thus need to ensure that, when the $\event$ occurs, the expected value collected is negligible, at most $\calO(\sqrt{\pr[\event]}\E[\optR{\B}])$. We establish this in two steps. 
First, we introduce a fluid upper bound $\optfl$ on $\E[\optR{\B}]$ which is tight up to $(1\pm\eps)$ factors in the large-budget regime. Then, we characterize the expected value collected by \Cref{alg:genericOnline} on $\event$ through a feasible solution to this LP \emph{with the constraints scaled by a factor $2\sqrt{\pr[\event]}$}. This characterization allows us to establish the claimed bound.

\paragraph{Fluid characterization of $\E[\optR{\B}]$}
In the \preview{} setting, the solution $\optR{\B}$ to \eqref{program:ora-lp} is a random variable where each request $\online_t = (\val_t, \alloc_t, \decisionset) \sim \D_t$ is random. 
Because of this source of randomness, it will be convenient to introduce the following fluid LP:

\begin{equation}\tag{$\LPrafluid$}\label{program:ora-ub-alt}
    \begin{aligned}
        \optfl ~:=~ &\text{maximize}  &&\textstyle\sum_{t\in[n]} \E_{\online_t}[\sum_{\decisionvar\in\decisionset[t]} \valat[t] \xlp[t]]\\
        &\text{s.t.}  &&\textstyle\sum_{t\in[n]} \E_{\online_t}[\sum_{\decisionvar\in\decisionset[t]} \allocat[t] \xlp[t]] \leq \B\\
        &\forall t\in [n], \online_t\in\typespace[t],  &&\textstyle \sum_{\decisionvar\in\decisionset[t]} \xlp[t] \leq 1\\
        &\forall t\in [n], \online_t\in\typespace[t], \decisionvar\in\decisionset[t],  && \xlp[t] \geq 0\\
    \end{aligned}
\end{equation}

It is folklore that $\optR{\B}$ and $\optfl$ are equivalent in the large budget regime, up to $(1\pm\eps)$ multiplicative factors. 

\begin{restatable}{Observation}{restateFluidOptUpperBound}\label{obs:fluidOptUpperBound}
    For any $\eps \in (0,1)$,
assuming $B_j = \Omega(\eps^{-2} \log(\nicefrac{m}{\eps}))$,
    we have that, assuming the requests are generated according to the \preview{} model:
    \begin{align*}
        \optfl
        \geq \E[\optR{\B}] 
        \geq (1-\eps) \optfl.
    \end{align*}
\end{restatable}
The upper bound on $\E[\optR{\B}]$ follows immediately by the standard averaging argument (e.g., \cite[Lemma 2.1]{devanur2019near}), noting that if $\zstarlp$ is an optimal solution \ref{program:ora-lp}, then $\xlpat{\online} = \E[\zstarlp \mid \online_t = \online]$ is feasible for \ref{program:ora-ub-alt} with value $\E[\optR{\B}]$. The lower bound follows from standard contention resolution schemes/randomized rounding with alterations arguments for packing LPs (e.g., \cite{bansal2012solving,chekuri2014contentionRes,chekuri2020ellone}). The main idea is to use an optimal solution $\xlpat{\online}$ to \ref{program:ora-ub-alt}$((1-\calO(\eps))\B)$ to define a feasible solution $\zlp$ to \ref{program:ora-lp}$(\B)$ on requests $(\online_t)_{t\in[n]}$ by taking $\zhlp = \xlp$, truncating the terms after the time $\tau$ that allocation of $\zhlp$ exhausts the budget of one of the resources, i.e., $\zlp = \zhlp \1[t \leq \tau]$, and applying Bernstein concentration bounds to show that truncation happens with negligible probability (under the large-budgets assumption).

\paragraph{Expressing the value collected on $\event$ via \eqref{program:ora-ub-alt}.}
The key ingredient in the proof of \Cref{lem:ScaledAlgValue} is demonstrating an LP solution to a rescaled version of \ref{program:ora-ub-alt}$(\B)$ where the constraints are scaled by $\sqrt{\pr[\event]}$, and showing that this scaled solution is feasible. In particular, we will establish the following:

\begin{Lemma}\label{lem:scaledLPSol}
Let $\event$ be any event depending only on the \offHist{} $\histend$. Then, define, for all $t\in [n]$, $\online = (\val, \alloc, \decisionset[t])\in\typespace[t]$, $\decisionvar \in \decisionset[t]$:
\begin{align*}
\ylpbb[t]{\online} := \pr[\event, \decisiononline = \decisionvar \mid \online]
\end{align*}
Denote $q := \pr[\event]$. 
Then, assuming $q>0$, the variables $\nicefrac{\ylpbb[t]{\online}}{2\sqrt{q}}$ for $t\in[n]$, $\online\in\typespace[t]$, and $\decisionvar\in\decisionset[t]$ form a feasible solution to \ref{program:ora-ub-alt}$(\B)$.
In particular,
\begin{enumerate}
    \item[(a)] $\sum_{t\in[n]} \E[\sum_{\decisionvar\in\decisionset[t]} \allocat[t]\ylpbb[t]{\online_t}] \leq 2\sqrt{q}\cdot \B$.
    \item[(b)] $\sum_{\decisionvar\in\decisionset} \ylpbb[t]{\online_t} \leq q$ for all $t\in[n]$ and $\online_t\in\typespace[t]$.
\end{enumerate}
\end{Lemma}

Using \Cref{lem:scaledLPSol}, \Cref{lem:ScaledAlgValue} follows readily by the tower rule of expectations. Indeed, we first rewrite the expected value collected by allocations $\decisiononline[t]$ on $\event$ as:
\begin{align*}
    \textstyle\E[\1[\event] \sum_{t\in \requests} \valonline]
    &=  \textstyle\sum_{t\in[n]} \E[\E[\sum_{\decisionvar \in \decisionset[t]}\valat[t] \1[\event,\decisiononline = \decisionvar] \mid \online_t]]\\
    &=  \textstyle\sum_{t\in[n]} \E[\sum_{\decisionvar \in \decisionset[t]}\valat[t] \pr[\event, \decisiononline = \decisionvar \mid \online_t]]\\
    &=  \textstyle\sum_{t\in[n]} \E[\sum_{\decisionvar \in \decisionset[t]}\valat[t] \ylpbb[t]{\online_t}].
\end{align*}
Now, if $q=0$, notice that $\ylpbb[t]{\online_t}=0$ (since $\online_t$ is independent of $\histend$ and thus also $\event$), so the claimed result is immediate. Otherwise, if $q>0,$
combining these observations and applying \Cref{lem:scaledLPSol}, we conclude:
\begin{align*}
    \textstyle\E[\1[\event] \sum_{t\in\requests} \valonline]
    =  2\sqrt{q}\textstyle\sum_{t\in[n]} \E[\sum_{\decisionvar \in \decisionset[t]}\valat[t] \nicefrac{\ylpbb{\online_t}}{2\sqrt{q}} ]
    \leq 2\sqrt{q}\optfl.
\end{align*}
Then, using \Cref{obs:fluidOptUpperBound} to upper-bound $\optfl \leq 2\E[\optR{\B}]$ yields the claimed result.
Thus, we conclude by establishing \Cref{lem:scaledLPSol}.

\paragraph{Bounding the expected allocation of the LP solution: Proof of \Cref{lem:scaledLPSol}(a)}
Now, we start to prove \Cref{lem:scaledLPSol}. Begin by noticing, by definition of $\ylpbb[t]{\online_t}$, and applying the tower rule of expectations, we have that:
\begin{align*}
    \textstyle\sum_{t\in[n]}\E[\sum_{\decisionvar\in\decisionset[t]}\allocatj[t]\ylpbb{\online_t}]
    &= \textstyle\sum_{t\in[n]}\E[\sum_{\decisionvar\in\decisionset[t]}\allocatj[t]\E[\1[\event, t\in\requests,\decisiononline[t] = \decisionvar] \mid \online_t]]\\
    &= \textstyle\sum_{t\in[n]}\E[\1[\event] \sum_{\decisionvar\in\decisionset[t]}\allocatj[t]\1[t\in\requests,\decisiononline[t] = \decisionvar]]\\
    &= \textstyle\E[\1[\event] \sum_{t\in\requests}\alloconlinej[t]].
\end{align*}
Now, by \Cref{lem:allocOnEvent}, the above expression is bounded as
\begin{align*}
    \textstyle\sum_{t\in[n]}\E[\sum_{\decisionvar\in\decisionset[t]}\allocatj[t]\ylpbb{\online_t}]
    \leq 2\sqrt{\pr[\event]} B_j
    = 2\sqrt{q} B_j
\end{align*}
This establishes the claim.

\paragraph{Bounding the allocation probability of the LP solution: Proof of \Cref{lem:scaledLPSol}(b)}
We begin by noting, for any fixed $t\in[n]$ and $\online_t = (\val_t, \alloc_t, \decisionset[t])\in\typespace[t]$, since $\blackbox$ selects a single decision,
\begin{align*}
    \textstyle\sum_{\decisionvar\in\decisionset[t]} \ylpbb[t]{\online_t}
    = \sum_{\decisionvar\in\decisionset[t]} \pr[\event,\decisiononline = \decisionvar \mid \online_t]
    = \pr[\event \mid \online_t] .
\end{align*}
Now, in the \preview{} model, notice that $\histend$, and hence also $\event$, is independent of $\online_t$. Therefore, the above simplifies to
\begin{align*}
    \textstyle\sum_{\decisionvar\in\decisionset[t]} \ylpbb[t]{\online_t}
    \leq \pr[\event]
    \leq q,
\end{align*}
as claimed.

\subsection{Bounding probability of $\be$: Proof of \cref{lem:concentrationBudgetExhaustion}}
\label{sec:resourceAlloc-missing}

We show that, on the $\good$ (pseudo) allocation event, the probability of budget exhaustion $\be$ is small, even when we condition on the sample history $\histend$ of \Cref{alg:genericOffline}.

\begin{proof}[Proof of \Cref{lem:concentrationBudgetExhaustion}]
Since $\be$ depends on $\online_t$, it will be convenient to introduce the following modification to the $\be$ event, $\bet$, which is the event of budget exhaustion when the allocation $\alloconline$ at time $t$ is ignored and $\B$ is replaced by $\B-\1$:
    \begin{align*}
        \textstyle\bet = \{\exists j\in[m] : \sum_{t'\in \requests\setminus\{t\}} \alloconlinej[t'] > B_j - 1\}.
    \end{align*}
    Notice that, since $\alloconlinej \in [0,1]$, if $\be$ occurs, then so must $\bet$, i.e., $\be \subseteq \bet.$ Therefore, it suffices to bound the following probability:
    \begin{align*}
        \pr[\good \cap \be \mid \histend, \online_t, \decisiononline]
        \leq \pr[\good \cap \bet \mid \histend, \online_t, \decisiononline].
    \end{align*}
To do this, observe that $\good \cap \bet$ is independent of the current request $\online_t$ and decision $\decisiononline$ conditioned on $\histend$. This follows since $\good$ is deterministic given $\histend$ by \cref{def:goodEvent}, and $\bet$ is defined after removing the allocation at time $t$, and the decisions of the algorithm are \OnlHistInd{} (\Cref{obs:oblivious}). Thus,
\begin{align*}
    \pr[\good \cap \bet \mid \histend, \online_t, \decisiononline]
    = \pr[\good \cap \bet \mid \histend] .
\end{align*}
We conclude by showing that $\pr[\good \cap \bet \mid \histend] \leq \eps$. Indeed, recall that, by \Cref{def:goodEvent,def:budgetExhaustion}, $\good\cap\bet$ implies that, for some resource $j\in[m]$,
\begin{align*}
    \textstyle\sum_{t\in \requests} \E[\alloconline[t] \mid \histend] \leq (1-\eps) \B
    \quad\text{and}\quad
    \sum_{t'\in \requests\setminus\{t\}} \alloconlinej[t'] > B_j - 1.
\end{align*}
Therefore, under $\good\cap \bet$, for some $j\in[m]$, using the above and the fact that $\alloconlinej \in [0,1]$,
\begin{align*}
     \textstyle\sum_{t'\in \requests} \alloconlinej[t'] - \E[\alloconlinej[t'] \mid \histend]
     &\geq\textstyle \sum_{t'\in \requests\setminus\{t\}} \alloconlinej[t'] - \sum_{t'\in \requests}\E[\alloconlinej[t'] \mid \histend]\\
     &\geq\textstyle B_j - 1 - \sum_{t'\in \requests}\E[\alloconlinej[t'] \mid \histend]\\
     &\geq\textstyle \eps B_j - 1.
\end{align*}
Now, since $B_j \geq 16\log(\nicefrac{2m}{\eps})\cdot\eps^{-2}$ and $B_j \geq \sum_{t'\in \requests} \E[\alloconlinej[t'] \mid \histend]$ (since $\good$ holds), we conclude that $\good\cap\bet$ implies
\begin{align*}
     \textstyle\sum_{t'\in \requests} \alloconlinej[t'] - \E[\alloconlinej[t'] \mid \histend]
     &\textstyle\geq \frac{\eps}{2} \frac{16\log(\nicefrac{2m}{\eps})}{\eps^2} - 1 + \frac{\eps}{2} \sum_{t'\in \requests} \E[\alloconlinej[t'] \mid \histend]\\
     &\textstyle> \frac{6\log(\nicefrac{2 m}{\eps})}{\eps} + \frac{\eps}{2} \sum_{t'\in \requests} \E[\alloconlinej[t'] \mid \histend].
\end{align*}
Finally, using the fact that $(\alloconlinej)_{t\in\requests}$ are bounded on $[0,1]$ and independent conditioned on the \offHist{} $\histend$ (by the $\OnlHistInd{}$ property in \cref{obs:oblivious}), we may apply Bernstein's inequality (specifically, \cref{cor:concentrationWithoutReplacement}) to conclude:
\begin{align*}
    &\pr[\good \cap \be \mid \histend]\\
    &\leq\pr[\good \cap \bet \mid \histend]\\
    &\leq \textstyle\pr\Big[\exists j\in[m] : \sum_{t\in\requests} \alloconlinej[t] > \frac{6\log(\nicefrac{2 m}{\eps})}{\eps} + (1+\frac{\eps}{2}) \sum_{t\in\requests} \E[\alloconlinej[t] \mid \histend] \mid \histend \Big]
    \leq \eps.  \qedhere
\end{align*}    
\end{proof}

\subsection{Online Resource Allocation with Pricing Queries}
\label{subsec:PricingQuery}

We show that, for online resource allocation, instead of observing the full information in each sample $\sample_t$, \Cref{alg:genericOnline} can use only the response to one \emph{pricing query} per $\sample_t$. The pricing query model, initiated for large-supply online resource allocation by \cite{TW-26}, is motivated by \emph{item-pricing algorithms}: before each request arrives, an item-pricing algorithm posts a price for each resource, and the buyer responds by choosing her favorite decision, namely, the decision that maximizes her utility---the difference between her value and her payment.
Item-pricing algorithms are desirable for their ease of implementation and because they directly yield truthful online mechanisms. We give a formal definition of item-pricing algorithms:

\begin{Definition}[Item-pricing algorithm]
\label{def:item-pricing}
    An online resource allocation algorithm is an \emph{item-pricing algorithm} if, before each request $\online_t$ arrives, it posts a per-unit price vector $\bprice_t\in[0,\infty]^m$. The buyer then chooses a utility maximizing decision $\decisiononline[t]\in\arg\max_{\decisionvar\in\decisionset[t]}
        \left\{\valat[t]-\langle\bprice_t,\allocat[t]\rangle\right\}$,
    receives the allocation $\alloconline[t]$, and pays $\langle\bprice_t,\alloconline[t]\rangle$.  When multiple decisions have the same utility,  we assume buyer $t$'s (possibly randomized) tie-breaking rule is independent to other buyers.
\end{Definition}

We study scenarios, where we may only have pricing-query access to each sample, rather than access to the full sample. More specifically, pricing-query access is defined as follows:

\begin{Definition}[Pricing query]
    A \emph{pricing query} to a sample $\sample_t=(\sampleval_t,\samplealloc_t,\decisionsetsample[t])$ specifies a price vector $\bprice\in[0,\infty]^m$, and the sample returns a decision
$\decisionsample[t]\in\arg\max_{\decisionvar\in\decisionsetsample[t]}
        \left\{\sampleval_t(\decisionvar)-\langle\bprice,\samplealloc_t(\decisionvar)\rangle\right\}$,
    together with its resource consumption $\allocsample[t]$, with the tie-breaking rule specified by \Cref{def:item-pricing}. The query does not reveal the sample's valuations, including the value $\valsample[t]$ of the chosen decision. 
\end{Definition}

Our main result in this subsection is to show that, given an estimate of $\opt$,  the sample access in \Cref{thm:resourceAlloc} can be further reduced to only one pricing query per sample. To achieve this, it suffices to provide a near-optimal \ro{} online algorithm that relies only on pricing queries. Fortunately, such an RO algorithm is known.

\begin{Lemma}[Theorem 4.3 of \cite{GSW-STOC25}]
    \label{lma:ora-pricing-query}
There exists an online resource allocation algorithm with the following guarantee. Given an estimate $\widehat \opt$ satisfying $\widehat \opt \in [\optR[\samples]{\alpha\B}/\beta, \optR[\samples]{\alpha\B}]$, when run on samples $(\sample_i)_{i\in\samples}$ in \ro{} with budget $\alpha\B$, where $\alpha\in(0,1]$ and $\alpha B_j = \Omega(\log(\nicefrac{\beta m}{\eps})\cdot\eps^{-2})$ for all $j\in[m]$, the algorithm returns an integral online allocation
with expected value
    \[
        \E[\algbbat[\samples]{\alpha \B}\mid (\sample_t)_{t\in\samples}]
        \geq
        (1-\eps)\optR[\samples]{\alpha\B},
    \]
    where the expectation is over the random arrival order and any internal randomness of the algorithm. Furthermore, the algorithm is an item-pricing algorithm and uses only one pricing query per sample $\sample_i$.
\end{Lemma}

Then, as an immediate corollary, we obtain the desired $(1-\epsilon)$-competitive algorithms in the \preview{} setting, as well as the \sspi{} (\Cref{thm:sspiInformal}) and \ows{} (\Cref{thm:owsInformal}) settings, via an item-pricing algorithm using only one pricing query per sample:

\begin{Corollary}\label{cor:ora-pricing-query}
    Suppose we are given an estimate $\hopt \in [\optR[\samples]{p(1-2\eps)\B}/\beta, \optR[\samples]{p(1-2\eps)\B}]$ which is independent of the samples $(\sample_t)_{t\in\samples}$. Run \Cref{alg:genericOnline} with $\blackbox$ being the \ro{} algorithm in \Cref{lma:ora-pricing-query}, instantiated with budget $p(1-2\eps)\B$, where each $B_j = \Omega(p^{-1}\eps^{-2}(1 + \log(\nicefrac{\beta m}{\eps})))$. Then, using only one pricing query per sample, we have
    \[
        \E[\alg]
        ~\geq~ (1-\eps)^2 p^{-1}\cdot\E[\optR[\samples]{p(1-2\eps)\B}] - 4\eps \E[\optR{\B}]
        ~\geq~ (1-\calO(\eps))\cdot\E[\optR{\B}].
    \]
\end{Corollary}

We remark that, to ensure that the algorithm in \Cref{cor:ora-pricing-query} is an item-pricing algorithm, we should modify \Cref{alg:genericOnline} as follows: instead of voiding an action that would cause budget exhaustion, we should terminate allocation, or equivalently set the price of every item to infinity, as soon as at least $B_j-1$ units of some resource $j$ have been allocated. The proof of \Cref{thm:resourceAlloc} still holds with minor adjustments. We also note that assuming access to an estimate of $\opt$ is necessary in the pricing query model, since learning an estimate given only polynomially many pricing queries is impossible (see \cite[Section 3]{TW-26}).

\section{Applications to Load Balancing and Mixed Packing-Covering}

We now show how to apply the framework from \Cref{sec:generalFramework} to Packing-Covering Multiple-Choice LPs. Then, we show how to apply these results to the Generalized Load Balancing problem.

\subsection{Packing-Covering Multiple-Choice LPs}
\label{sec:packingCovering}

We show that our framework from \Cref{sec:generalFramework} further applies to Online Packing-Covering Multiple-Choice LPs. Recall that, in this setting, we are given packing and covering constraints $\Bpacking, \Bcovering \in \R_{\geq 0}^m$, and $n$ requests arrive in an arbitrary order. Each request $\online_t = (\alloc_t, \decisionset)$ must be served by some decision $\decisionvar \in \decisionset$, which generates a load $\alloconline = (\packingonline, \coveringonline)$, where $\packingonline, \coveringonline \in [0,1]^m$ correspond to the packing and covering loads, respectively. Unlike in the Resource Allocation setting of \cref{sec:resourceAlloc}, we do not require that a null decision $\phi$ with $\packing_t(\phi) = \0 = \covering_t(\phi)$ be feasible. The goal is to approximately satisfy the packing and covering constraints $\Bpacking, \Bcovering$ after all requests have been served, i.e., to find an approximate online integral solution to the following program:
\begin{equation}\tag{$\IPpc$}\label{program:opc-lb}
    \begin{aligned}
        &\text{find}  &&\textstyle (\xip[t])_{t\in \requests, \decisionvar\in\decisionset[t]}\\
        &\text{s.t. } \textit{Packing}: &&
        \textstyle\sum_{t\in\requests}
        \sum_{\decisionvar\in\decisionset[t]}
        \packingat[t]\xip[t]
        \leq \Bpacking\\
        &\textit{Covering}: &&
        \textstyle\sum_{t\in\requests}
        \sum_{\decisionvar\in\decisionset[t]}
        \coveringat[t]\xip[t]
        \geq \Bcovering\\
        &\forall t\in \requests,  &&\textstyle \sum_{\decisionvar\in\decisionset[t]} \xip[t] = 1\\
        &\forall t\in \requests, \decisionvar\in\decisionset[t],  && \xip[t] \in [0,1]
    \end{aligned}
\end{equation}
In particular, we want to, with high probability, find an \emph{$\calO(\eps)$-feasible} integral solution: a solution to \ref{program:opc-lb} which satisfies the packing constraints $(1+\calO(\eps))\Bpacking$ and covering constraints $(1-\calO(\eps))\Bcovering$.

Recall that, in the \preview{} model, each sample $\sample_t$ and request $\online_t$ is sampled independently from distribution $\D_t$. Thus, the feasibility of \ref{program:opc-lb} depends on the realized requests. Because of this, it will be convenient to introduce the following stochastic LP relaxation:
\begin{equation}\tag{$\LPpc$}\label{program:opc-relax}
    \begin{aligned}
        &\text{find}  &&\textstyle (\xlp[t])_{t\in \requests, \online_t \in \typespace \decisionvar\in\decisionset[t]}\\
        &\text{s.t. } \textit{Packing}: &&
        \textstyle\sum_{t\in\requests}
        \E_{\online_t}[\sum_{\decisionvar\in\decisionset[t]}
        \packingat[t]\xlp[t]]
        \leq \Bpacking\\
        &\textit{Covering}: &&
        \textstyle\sum_{t\in\requests}
        \E_{\online_t}[\sum_{\decisionvar\in\decisionset[t]}
        \coveringat[t]\xlp[t]]
        \geq \Bcovering\\
        &\forall t\in \requests, \online_t \in \typespace,  &&\textstyle \sum_{\decisionvar\in\decisionset[t]} \xlp[t] = 1\\
        &\forall t\in \requests, \online_t \in \typespace, \decisionvar\in\decisionset[t],  && \xlp[t] \in [0,1]
    \end{aligned}
\end{equation}

We show that, given any black-box \ro{} algorithm $\blackbox$ which approximately satisfies its constraints with high probability, the allocation of \Cref{alg:genericOnline} in the \preview{} setting will also approximately satisfy the constraints with high probability. 

\begin{Theorem}
\label{thm:packingCovering}
    Suppose we run any black-box \ro{} algorithm $\blackbox$ for Online Packing-Covering Multiple-Choice LPs offline in \Cref{alg:genericOffline} on samples $(\sample_i)_{i\in\samples}$ with packing constraint $p(1+\eps)\Bpacking$ and covering constraint $p(1-\eps)\Bcovering$, where $p\in(0,1]$.
    If we run \Cref{alg:genericOnline} online in the \preview{} setting (with no feasibility cap imposed), then 
    for any $\eps,\delta\in(0,1)$, the resulting packing and covering allocations $\packingalg$ and $\coveringalg$ satisfy, coordinatewise,    
    with probability at least $1-\delta$,
    \begin{align*}
        \packingalg
        &\leq (1+\eps) p^{-1}\packingbb
        + 3\eps^{-1}(3 + 10 p^{-1})\log(\nicefrac{8 m}{\delta})\\
        \coveringalg
        &\geq (1-\eps) p^{-1}\coveringbb
        -3\eps^{-1}(3 + 10 p^{-1})\log(\nicefrac{8 m}{\delta}),
    \end{align*}
    where $\packingbb$ and $\coveringbb$ are the sample packing and covering allocations of the offline run of \cref{alg:genericOffline}.
\end{Theorem}

To obtain our results for the \preview{} setting, we combine this theorem with the following \ro{} algorithm for Online Multi-Choice Packing-Covering LPs from \cite{GM-MOR16}.

\begin{Lemma}[\cite{GM-MOR16}]\label{lem:packingCoveringRO}
    Fix any $\alpha, \beta > 0$, and
    suppose \ref{program:opc-lb} is feasible on samples $(\sample_i)_{i\in\samples}$ with packing and covering constraints $\alpha \Bpacking$, $\beta \Bcovering$. Then, there exists an online packing-covering multiple-choice LP algorithm such that, for an absolute constant $\eps_0\leq 1$ and any $\eps \in (0,\eps_0)$ and $\delta\in(0,\eps)$, with probability at least $1-\delta$, it returns an integral online allocation which is an approximate solution to \ref{program:opc-lb} assuming $\alpha \Bpackingj \wedge \beta\Bcoveringj = \Omega(\eps^{-2}\cdot\log(\nicefrac{m}{\delta}))$ for all $j\in [m]$, i.e.,
    \begin{align*}
        \packingbb[\alpha] \leq (1+\eps)\alpha\Bpacking 
        \quad \text{and} \quad
        \coveringbb[\beta] \geq (1-\eps)\beta\Bcovering,
    \end{align*}
    where the probability is over the random arrival order and any possible randomness of the algorithm.
\end{Lemma}
We note that, while the \ro{} result \cite[Theorem 4.1]{GM-MOR16} is stated for Packing-Covering Multiple-Choice LPs with ``full simplex'' constraints (i.e., the constraints $\sum_{\decisionvar\in\decisionset}\xip \leq 1$ instead of $\sum_{\decisionvar\in\decisionset}\xip = 1$ in \ref{program:opc-lb}), their result for LPs with the simplex constraint as in \ref{program:opc-lb} follows almost immediately from their arguments. We discuss the required modifications in \cref{sec:packingCoveringRO}.

We also need the following elementary result that guarantees, if \ref{program:opc-relax} is feasible, then with high probability,  \ref{program:opc-lb} on the samples $(\sample_t)_{t\in\samples}$ is feasible after budget rescaling. 
\begin{restatable}{Lemma}{restatePackingCoveringRescaling}\label{lem:packingCoveringRescaling}
    Suppose \ref{program:opc-relax} with packing and covering constraints $\Bpacking, \Bcovering$ is feasible.
    Let $(\sample_t)_{t\in\samples}$ be the $p$-sample in the \preview{} model.
    Then, for any $\eps\in(0,1), \delta\in(0,1)$, assuming $\Bpackingj \wedge \Bcoveringj \geq 192 p^{-1}\eps^{-2} \log(\nicefrac{8 m}{\delta})$ for all $j\in [m]$, with probability at least $1-\delta$, there is a feasible solution to \ref{program:opc-lb} on samples $(\sample_t)_{t\in\samples}$ with packing and covering constraints $p(1+\eps)\Bpacking$ and $p(1-\eps)\Bcovering$.
\end{restatable}

\begin{proof}
    Let $(\xstarlpat{\online})_{t\in[n], \online \in \typespace, \decisionvar \in \decisionset}$ be a feasible solution to \ref{program:opc-relax} with packing and covering constraints $\Bpacking, \Bcovering$. Then, let $\sample_t \sim \D_t$ for all $t\in[n]$, and $\samples \subseteq [n]$ be a uniformly random subset of size $\nsamples = p n$. We construct a feasible solution $\xip$ to \ref{program:opc-lb} on samples $(\sample_t)_{t\in\samples}$ and constraints $p(1+\eps)\Bpacking$ and $p(1-\eps)\Bcovering$ as follows: First, define $\xprimeip = \xstarlpat{\sample_t}$ for all $t\in [n]$ and $\decisionvar\in\decisionsetsample$, the decision set associated with $\sample_t$. Then, define $\xip = \xprimeip$ for all $t\in\samples$ and $\decisionvar \in \decisionsetsample$. Our argument proceeds in two steps: (i) first, we will show that $\xprimeip$ is feasible for \ref{program:opc-lb} on $(\sample_t)_{t\in[n], \decisionvar \in \decisionsetsample}$ with packing and covering constraints $(1+\calO(\eps))\Bpacking$ and $(1-\calO(\eps))\Bcovering$ with high probability. Then (ii) conditioned on the feasibility of $\xprimeip$, we will show that $\xip$ is feasible for \ref{program:opc-lb} on $(\sample_t)_{t\in\samples}$ with packing and covering constraints $p(1+\eps)\Bpacking$ and $p(1-\eps)\Bcovering$.

    First, by feasibility of $\xstarlpat{\online}$, we trivially have that $\xprimeip \in [0,1]$ and $\sum_{\decisionvar \in \decisionsetsample} \xprimeip = 1$. Thus, it remains to reason about the packing and covering constraints. Now, since the samples $(\sample_t)_{t\in[n]}$ are independent in the \preview{} model and identically distributed to the requests $(\online_t)_{t\in[n]}$, and $\packingsampleat[t], \coveringsampleat[t] \in [0,1]$, we may apply Bernstein's inequality (specifically, \cref{cor:concentrationWithoutReplacement}) to conclude that, for any $\eta, \delta\in(0,1)$, with probability at least $1-\nicefrac{\delta}{2}$, the following inequalities hold simultaneously for all $j\in[m]$
    \begin{align*}
        \sum_{t\in[n]} \sum_{\decisionvar\in\decisionsetsample} \packingsampleatj \xprimeip
        &\leq \frac{3\log(\nicefrac{8m}{\delta})}{\eta} + (1 + \eta)\sum_{t\in[n]} \E_{\online_t}\left[\sum_{\decisionvar\in\decisionset} \packingatj \xstarlpat{\online_t}\right]
        \leq \frac{3\log(\nicefrac{8m}{\delta})}{\eta} + (1 + \eta)\Bpackingj\\
        \sum_{t\in[n]} \sum_{\decisionvar\in\decisionsetsample} \coveringsampleatj \xprimeip
        &\geq - \frac{3\log(\nicefrac{8m}{\delta})}{\eta} + (1 - \eta)\sum_{t\in[n]} \E_{\online_t}\left[\sum_{\decisionvar\in\decisionset} \coveringatj \xstarlpat{\online_t}\right]
        \geq -\frac{3\log(\nicefrac{8m}{\delta})}{\eta} + (1 - \eta)\Bcoveringj.
    \end{align*}
    Now, conditioned on $(\sample_t)_{t\in[n]}$, $\samples$ is a uniformly random subset of size $\nsamples=pn$. Thus, we apply Bernstein's inequality for sampling without replacement (\cref{cor:concentrationWithoutReplacement}) so that, with probability at least $1-\nicefrac{\delta}{2}$, for all $j\in[m]$,
    \begin{align*}
        \sum_{t\in\samples} \sum_{\decisionvar\in\decisionsetsample} \packingsampleatj \xip
        &\leq \frac{3\log(\nicefrac{8m}{\delta})}{\eta} + (1 + \eta) p\sum_{t\in[n]} \sum_{\decisionvar\in\decisionsetsample} \packingsampleatj \xprimeip\\
        \sum_{t\in\samples} \sum_{\decisionvar\in\decisionsetsample} \coveringsampleatj \xip
        &\geq -\frac{3\log(\nicefrac{8m}{\delta})}{\eta} + (1 - \eta) p\sum_{t\in[n]} \sum_{\decisionvar\in\decisionsetsample} \coveringsampleatj \xprimeip.
    \end{align*}
    Combining these two inequalities, choosing $\eta = \nicefrac{\eps}{4}$, and using the fact that $\Bpackingj \wedge \Bcoveringj \geq 192\log(\nicefrac{8m}{\delta}) \eps^{-2} p^{-1}$, we have that, with probability at least $1-\delta$, for all $j\in[m]$,
    \begin{align*}
        \sum_{t\in\samples} \sum_{\decisionvar\in\decisionsetsample} \packingsampleatj \xip
        &\leq \frac{12\log(\nicefrac{8m}{\delta})}{\eta} + (1 + 3\eta) p \Bpackingj
        \leq (1 + 4\eta) p \Bpackingj = (1+\eps)p \Bpackingj\\
        \sum_{t\in\samples} \sum_{\decisionvar\in\decisionsetsample} \coveringsampleatj \xip
        &\geq -\frac{6\log(\nicefrac{8m}{\delta})}{\eta} + (1 - 2\eta) p\Bcoveringj
        \geq (1 - 3\eta) p\Bcoveringj \geq (1-\eps) p \Bcoveringj,
    \end{align*}
    as claimed.
\end{proof}

Now, combining \cref{thm:packingCovering} with \cref{lem:packingCoveringRO,lem:packingCoveringRescaling}, we immediately obtain the desired guarantees for the \preview{} setting (and hence also in the \sspi{} and \ows{} settings):

\begin{Corollary}\label{cor:packingCovering}
    Suppose \ref{program:opc-relax} with packing and covering constraints $\Bpacking, \Bcovering$ is feasible.
    In the setting of \cref{thm:packingCovering}, suppose we run \cref{alg:genericOnline} with $\blackbox$ being the \ro{} algorithm in \cref{lem:packingCoveringRO} instantiated with packing and covering constraints $p(1+\eps)\Bpacking$ and $p(1-\eps)\Bcovering$. Then, for an absolute constant $\eps_1\leq 1$ and any $\eps \in (0,\eps_1)$ and $\delta\in (0,\eps)$, with probability $1-\delta$, assuming $\Bpackingj \wedge \Bcoveringj = \Omega(p^{-1}\eps^{-2}\log(\nicefrac{m}{\delta})),$ \cref{alg:genericOnline} finds an integral $\calO(\eps)$-feasible solution to \ref{program:opc-lb} on requests $(\online_t)_{t\in[n]}$, i.e.,
    \[
        \packingalg \leq (1+\calO(\eps))\Bpacking 
        ~\text{and}~ 
        \coveringalg \geq (1-\calO(\eps))\Bcovering.
    \]
\end{Corollary}

We now present the proof of \Cref{thm:packingCovering}. 
The proof proceeds in two stages: first, use the fact that the online decisions of \Cref{alg:genericOnline} are \OnlHistInd{} by \Cref{obs:oblivious} together with Bernstein's inequality to conclude that the packing and covering loads concentrate about their expectation conditioned on the sample history $\histend$. Then, apply Property 3 from \Cref{lem:algConnectionProperties} to conclude that the online packing and covering loads concentrate about those of the $\blackbox$.

\begin{proof}[Proof of \Cref{thm:packingCovering}]
    By \Cref{lem:algConnectionProperties}, for any $\eta,\delta\in(0,1)$, with probability at least $1-\nicefrac{\delta}{2}$, for every $j\in[m]$, since $\packingatj[t], \coveringatj[t] \in [0,1]$, we have the following two inequalities:
    \begin{align*}
        \sum_{t\in \requests} \E[\packingonlinej \mid \histend] 
        &\leq \frac{5\log(\nicefrac{8 m}{\delta})}{p\eta}
        +\frac{(1+\eta)}{p} \sum_{i\in[\nsamples]} \packingsamplej[\pi(i)]\\
        \sum_{t\in \requests} \E[\coveringonlinej \mid \histend] 
        &\geq -\frac{5\log(\nicefrac{8 m}{\delta})}{p\eta}
        +\frac{(1-\eta)}{p} \sum_{i\in[\nsamples]} \coveringsamplej[\pi(i)].
    \end{align*}
    Further, since the allocations $\packingonlinej, \coveringonlinej \in [0,1]$ and are independent across time conditioned on the \offHist{} $\histend$ (by the $\OnlHistInd{}$ property, \cref{obs:oblivious}), we may apply Bernstein's inequality (specifically, \cref{cor:concentrationWithoutReplacement}) to conclude that, with probability at least $1-\nicefrac{\delta}{2}$, for all $j\in [m]$,
    \begin{align*}
    \sum_{t\in\requests}\packingonlinej[t]
        &\leq \frac{3\log(\nicefrac{8 m}{\delta})}{\eta}
        +(1+\eta)
        \sum_{t\in\requests}
        \E[\packingonlinej[t]\mid\histend]\\
        \sum_{t\in\requests}\coveringonlinej[t]
        &\geq -\frac{3\log(\nicefrac{8 m}{\delta})}{\eta}
        +(1-\eta)
        \sum_{t\in\requests}
        \E[\coveringonlinej[t]\mid\histend].
    \end{align*}

Hence all four inequalities hold
simultaneously with probability at least $1-\delta$. Taking
$\eta=\eps/3$ gives
\begin{align*}
    \packingalg
    &\leq (1+\eps)p^{-1}\packingbb + 3\eps^{-1}(3 + 10 p^{-1})\log(\nicefrac{8 m}{\delta})\\
    \coveringalg
    &\geq (1-\eps)p^{-1}\coveringbb - 3\eps^{-1}(3 + 10 p^{-1})\log(\nicefrac{8 m}{\delta})
\end{align*}
as claimed.
\end{proof} 
\subsection{Generalized Load Balancing}

Recall that, in this setting, there are $m$ machines, and $n$ job requests arrive online in arbitrary order. Each request $\online_t = (\alloc_t, \decisionset)$ must be served by some decision $\decisionvar \in \decisionset$, which generates an allocation load $\alloconline \in [0,1]^m$ on a subset of the $m$ resources. The goal is to minimize the \emph{makespan}, i.e., the maximum load $\sum_{t\in\requests} \alloconlinej$ among all machines $j\in[m]$. 

Our benchmark is the \emph{hindsight fractional optimum} $\E[\optRlb{}]$, where $\optRlb{}$ is the optimal value of the configuration linear program defined below on the requests. For each $t\in\requests$, let $\xip$ denote the fractional allocation corresponding to the decision $\decisionvar\in\decisionset$ for request $\online_t$. Then, we can write the configuration LP as:

\begin{equation}\tag{$\LPlb$}\label{program:olb-lb}
    \begin{aligned}
        \optRlb ~:=~ &\text{minimize}  &&\textstyle \left\|\sum_{t\in \requests}\sum_{\decisionvar\in\decisionset[t]}  \loadat[t]\xip[t]\right\|_{\max}\\
        &\text{s.t. }\forall t\in \requests,  &&\textstyle \sum_{\decisionvar\in\decisionset[t]} \xip[t] = 1\\
        &\forall t\in \requests, \decisionvar\in\decisionset[t],  && \xip[t] \in [0,1]\\
    \end{aligned}
\end{equation}
Here, we use the notation $\|v\|_{\max} = \max_{j\in[m]} v_j$.
In this section, we show how to obtain $(1+\calO(\eps))$-competitive algorithms for Online Generalized Load Balancing in the \preview{} model as a corollary of our results for Online Packing-Covering Multiple-Choice LPs.

\begin{Theorem}
\label{thm:loadBalancing}
There is an algorithm for Online Generalized Load Balancing in the \preview{} model such that, for an absolute constant $\eps_2 \leq 1$ and any $\eps \in (0, \eps_2)$ and $\delta\in(0,\eps)$, if $\E[\optRlb] = \Omega(p^{-1}\eps^{-2} \log(\nicefrac{m}{\delta}))$, then with probability at least $1-\delta$, it finds an integral allocation with makespan $\alg$ satisfying:
\[
    \alg \leq (1+\calO(\eps))\E[\optRlb].
\]
\end{Theorem}
\paragraph{Algorithm.} Our algorithm follows the standard template for converting an algorithm for Online Packing Multiple-Choice LPs to an Online Generalized Load Balancing algorithm (see, e.g., \cite{GuptaM26}). The main idea is to use part of the sample to estimate the expected optimal makespan $\E[\optRlb]$, and the remaining part of the sample to run the Packing-Covering version of \cref{alg:genericOnline} with this estimate as the packing constraint. We state the algorithm for a generic estimator $\hopt$ which will be chosen shortly.

\begin{algorithm}[H]
\caption{\textsc{Online Generalized Load-Balancing via Packing-Covering LPs}}
\label{alg:loadBalancingViaPackingCovering}
\KwIn{time horizon $n \geq 1$, $\nsamples$ sampled indices $\samples\subseteq [n]$, samples $(\sample_t)_{t\in\samples}$}
    Let $\samples_1$ and $\samples_2$ be the first and second $\nicefrac{p n}{2}$ sample indices of $\samples$, respectively.
    
    Compute estimator $\hopt$ of $\E[\optRlb]$ based on $\samples_1$.\\

    Run \cref{alg:genericOnline} for packing-covering LPs as specified in \cref{cor:packingCovering} on the following instance:
    \begin{itemize}
        \item Sample $U\sim\unif(0,1)$, and let the packing and covering constraints be $\Bpackingj = (1+ U \eps)\hopt$ and $\Bcoveringj = n$ $\forall j\in[m]$.
        \item For each load balancing request $\online_t = (\load_t, \decisionset)$, define the associated packing-covering request $\packingat = \loadat$ and $\coveringatj = 1$ for all $j\in[m], \decisionvar \in \decisionset$ (and similarly for the samples $(\sample_i)_{i\in\samples_2}$).
        \item Feed the samples $\samples_2$ offline and online requests to \cref{alg:genericOnline}, and take the decisions $\decisiononline$.
    \end{itemize}
\end{algorithm}
Notice that the covering constraints in the above LPs will be trivially satisfied (since all covering allocations are $\coveringonlinej = 1$). Since the packing constraint upper-bounds the load on each resource, finding feasible solutions to both LPs yields an allocation with makespan at most $(1+\eps)\hopt$. Thus, we need an estimator which (i) is close to $\E[\opt]$ so that the computed makespan is small, yet (ii) sufficiently large so that, with high probability, the constructed LPs will be feasible.

To obtain such an estimator of $\E[\optRlb]$, we use the well-known fact that $\optRlb[\samples]$ concentrates about an appropriately-scaled version of $\E[\optRlb]$; similar versions have appeared, e.g., in \cite[Lemma 5.2]{GuptaM26}. Since the proof is standard, we defer it to \cref{sec:loadBalancingConcentration}.

\begin{restatable}{Lemma}{restateOptMakespanConcentration}\label{lem:optMakespanConcentration}
    Let $\optRlb[\samples]$ denote the optimal fractional value to \ref{program:olb-lb} on samples $(\sample_i)_{i\in\samples}$, and similarly $\optRlb$ the optimal fractional value on requests $(\online_t)_{t\in[n]}$. If the samples and requests are generated according to \preview{}, then, for any $\eps, \delta\in (0,1)$, with probability at least $1-\delta$,
    \[
        |\optRlb[\samples] - p \E[\optRlb] |
        \leq p\eps\E[\optRlb] + 144 \eps^{-1} \log(\nicefrac{8 m}{\delta}).
    \]
\end{restatable}
\cref{lem:optMakespanConcentration} allows us to construct estimators $\hopt_\ell$ satisfying the two desired properties listed above. We are thus now prepared to prove \cref{thm:loadBalancing}.

\begin{proof}[Proof of \cref{thm:loadBalancing}]
We establish the result for \cref{alg:loadBalancingViaPackingCovering}.
Let $\samples_1$ and $\samples_2$ be the first and second $\nicefrac{p n}{2}$ samples from the \preview{} respectively.
Notice that $\samples_1$ and $\samples_2$ are each (marginally) a \previewatmath[$\nicefrac{p}{2}$] for requests in $[n]$.
In particular, we run \cref{alg:loadBalancingViaPackingCovering} with estimator:
\begin{align}\label{eq:optMakespanEst}
    \textstyle \hopt = \frac{1}{1-\nicefrac{\eps}{3}}((\nicefrac{p}{2})^{-1}\optRlb[\samples_1] + 144 (\nicefrac{p}{2})^{-1}(\nicefrac{\eps}{3})^{-1}\log(\nicefrac{16 m}{\delta})). 
\end{align}

Now, let $\goodsampling = \{\E[\optRlb] \leq \hopt \leq (1+2\eps)\E[\optRlb]\}$ be the good sampling event depending on samples in $\samples_1$, and notice $\pr[\goodsampling] \geq 1-\nicefrac{\delta}{2}$ by \cref{lem:optMakespanConcentration}, since $\E[\optRlb] = \Omega(p^{-1}\eps^{-2}\log(\nicefrac{m}{\delta}))$. 
For $b\geq 0$, let $\goodload[b]$ denote the event that \cref{alg:loadBalancingViaPackingCovering} run with packing constraints $\Bpackingj = b$ for all $j\in[m]$ finds an online allocation with makespan at most $(1+c \eps)b$ for an absolute constant $c>0$.

Notice that, whenever $b \geq \E[\optRlb]$, then the packing-covering LP relaxation \ref{program:opc-relax} defined in \cref{alg:loadBalancingViaPackingCovering} is feasible: Indeed, if we define $\xlpat{\online} = \E[\xstarip \mid \online_t = \online]$, where $(\xstarip)_{t\in[n], \decisionvar\in\decisionset}$ is an optimal solution to \ref{program:olb-lb} on requests $(\online_t)_{t\in[n]}$, then the covering constraints are trivially satisfied, and the packing constraints are satisfied since
\begin{align*}
    \textstyle\forall j \in [m]:
    \sum_{t\in[n]} \E_{\online_t}[\sum_{\decisionvar\in\decisionset}\packingatj \xlp]
    = \E[\sum_{t\in[n]}\sum_{\decisionvar\in\decisionset} \loadonlinej \xstarip]
    \leq \E[\optRlb] \leq b.
\end{align*}
Thus, since $\E[\optRlb] = \Omega(p^{-1}\eps^{-2}\log(\nicefrac{m}{\delta}))$ and $(\sample_{t})_{t\in\samples_2}$ is a \previewat[\nicefrac{p}{2}], when \cref{alg:loadBalancingViaPackingCovering} is run with any fixed packing constraint $b\geq \E[\optRlb]$, by \cref{cor:packingCovering}, $\pr[\goodload[b]] \geq 1-\nicefrac{\delta}{8}$.
Now, notice that the packing constraints from \cref{alg:loadBalancingViaPackingCovering} are $\Bpackingj \sim \unif[\hopt, (1+\eps)\hopt]$. Thus, we can decompose
\begin{align*}
    \textstyle\pr[\goodsampling \cap \goodloadcomp[(1+U\eps)\hopt]]
    = \E\left[\1[\goodsampling] \int_{\hopt}^{(1+\eps)\hopt} \frac{\1[\goodloadcomp[b]]}{\eps \hopt} ~ \mathrm{d}b\right]
    \leq \int_{\E[\optRlb]}^{(1+4\eps)\E[\optRlb]} \frac{\pr[\goodloadcomp[b]]}{\eps \E[\optRlb]}  ~ \mathrm{d}b
    \leq \frac{\delta}{2},
\end{align*}
where the first inequality uses the definition of $\goodsampling$, and the second inequality uses the fact that $\pr[\goodloadcomp[b]] \leq \nicefrac{\delta}{8}$ for all $b\geq \E[\optRlb]$. Combining the above inequality with our observation that $\pr[\goodsampling^c] \leq \nicefrac{\delta}{2}$, it follows that $\pr[\goodsampling^c \cup \goodloadcomp[(1+U\eps)\hopt]] = \pr[\goodsampling^c] + \pr[\goodsampling\cap \goodloadcomp[(1+U\eps)\hopt]] \leq \delta$.
By construction, the packing allocations correspond to the request loads, i.e., $\packingat = \loadat$. Therefore, we conclude that, whenever $\goodsampling \cap \goodload$ hold (i.e., with probability at least $1-\delta$), the makespan of \cref{alg:loadBalancingViaPackingCovering} satisfies:
\[
    \textstyle\alg 
    \leq (1+U \eps)\hopt
    = (1+\calO(\eps)) \E[\optRlb] .     \qedhere
\]
\end{proof}

\newcommand{\lbDim}{z}
\newcommand{\lbN}{N}
\newcommand{\lbBudget}{B}
\newcommand{\lbNumJobs}{n}
\newcommand{\lbNumGroups}{G}
\newcommand{\lbNumHiddenGroups}{M}
\newcommand{\lbA}[1]{A_{#1}}
\newcommand{\lbB}[1]{B_{#1}}
\newcommand{\lbACount}[1]{a_{#1}}
\newcommand{\lbBCount}[1]{b_{#1}}
\newcommand{\lbChoice}[1]{q_{#1}}
\newcommand{\lbOptChoice}[1]{q_{#1}^{*}}
\newcommand{\lbHighCount}[1]{X_{#1}}
\newcommand{\lbTouchedSet}[1]{\mathcal T_{#1}}
\newcommand{\lbTouchedImbalance}[1]{W_{#1}}
\newcommand{\lbMachineLoad}[1]{L_{#1}}
\newcommand{\lbCapacity}[1]{c_{#1}}
\newcommand{\lbOccupancy}[1]{X_{#1}}
\newcommand{\lbForcedCount}[1]{Z_{#1}}
\newcommand{\lbAlgRA}{\mathsf{Alg}_{\mathsf{ra}}}
\newcommand{\lbOptRA}{\mathsf{Opt}_{\mathsf{ra}}}
\newcommand{\lbAlgLB}{\mathsf{Alg}_{\mathsf{lb}}}
\newcommand{\lbOptLB}{\mathsf{Opt}_{\mathsf{lb}}}

\section{Lower Bounds}
\label{sec:lower-bounds-new}

In this section, we show that the dependencies on $m$, $\epsilon$, and $p$ in our positive results are necessary.  For online resource allocation, we focus on the \psample{} model.  For the \sspi{} model, the known-i.i.d. construction of \cite{devanur2019near} already gives an $\Omega(\log m\cdot\epsilon^{-2})$ budget lower bound, even when the common arrival distribution is known to the algorithm.  The same construction can be benchmarked against the expected hindsight fractional optimum by retaining the realization-wise benchmark used in its proof.  Our \psample{} construction uses the same hypercube geometry, which also appears in \cite{AWY-OR14,GSW-STOC25}.  We prove the following theorem.

\begin{Theorem}[$p$-sample resource allocation]
\label{thm:lb-psample-ra}
For every $p\in(0,1/2]$, $\epsilon\in(0,10^{-3}]$,  sufficiently large $m$, and every randomized fractional online resource-allocation algorithm, there exists a fixed deterministic instance in the \psample{} model such that its common resource budget satisfies $\lbBudget=\Theta(\log m/(p\epsilon^2))$, and
\begin{equation*}
    \mathbb E[\lbAlgRA]
    \le
    (1-10^{-4}\epsilon)\,\lbOptRA.
\end{equation*}
Here, $\lbOptRA$ is the hindsight integral optimum on all $\lbNumJobs$ online requests.  The expectation is over the uniformly random $p$-sample and the algorithm's random coins.
\end{Theorem}

For load balancing, we give lower bounds both for the known i.i.d.\,and the \psample{} models.

\begin{Theorem}[Known i.i.d. load balancing]
\label{thm:lb-one-sparse-iid}
For every $\epsilon\in(0,10^{-3}]$ and all sufficiently large $m$ and every randomized fractional online load-balancing algorithm, there exists an online load balancing instance with $m$ resources and $n$ requests with known i.i.d. distribution such that \begin{equation*}
    \mathbb E[\lbAlgLB]
    \ge
    (1+10^{-4}\epsilon)\,\lbOptLB,
\end{equation*}
where $\lbOptLB$ is the expected hindsight integral optimum on all $\lbNumJobs$ realized requests, and $\lbOptLB=\Theta(\log m/\epsilon^2)$. 
\end{Theorem}

\begin{Theorem}[$p$-sample load balancing]
\label{thm:lb-one-sparse-psample}
For every $p\in(0,1/2]$, $\epsilon\in(0,10^{-3}]$, and all sufficiently large $m$ and every randomized fractional online load-balancing algorithm, there exists an online load balancing instance in the \psample{} model such that 
\begin{equation*}
    \lbOptLB
    =
    \Theta(p^{-1}\eps^{-2}\log(m))
    \qquad\text{and}\qquad
    \mathbb E[\lbAlgLB]
    \ge
    (1+10^{-3}\epsilon)\,\lbOptLB.
\end{equation*}
Here, $\lbOptLB$ is the hindsight integral optimum on all $\lbNumJobs$ online requests, and the expectation is over the uniformly random $p$-sample and the algorithm's random coins.
\end{Theorem}

We remark that a crucial property of the hard instances for \Cref{thm:lb-one-sparse-iid} and \Cref{thm:lb-one-sparse-psample} is that the requests  are \emph{sparse}, i.e., each request is ``unit-demand'' where satisfying a request (in any possible way) only adds load to at most one resource. This is in contrast to the hard instances for the online resource allocation problem, where the requests are  \emph{dense}, i.e., each request puts load on $\Theta(m)$ resources simultaneously. Indeed, this property of dense requests is necessary for obtaining the $\log m$ factor in the supply lower bound for online resource allocation problems, as \cite{KRTV-SICOMP18} shows that when each request is interested in at most $s$ resources, the $\log m$ dependency can be further improved to $\log s$. Therefore, our lower bounds \Cref{thm:lb-one-sparse-iid} and \Cref{thm:lb-one-sparse-psample} rule out the possibilities of such improvements for generalized load balancing problem.

\subsection{Resource Allocation Lower Bound in the $\psample$ Model}
\label{sec:lb-psample-ra}

In this subsection, we give the proof of \Cref{thm:lb-psample-ra}. Throughout the proofs, we assume without loss of generality that $1/(2p)$ is an integer, and $m$ is a power of $2$.

\paragraph{The Hypercube Geometry.} We begin with introducing some basic facts we need for the hard instance. These facts mainly follow the hard instance constructions in \cite{devanur2019near}.

Let $\lbDim:=\log_2m$. Since we assume $m$ is a power of $2$, $\lbDim$ must be an integer. We identify $m = 2^{\lbDim}$ active resources with the bit strings in $\{0,1\}^{\lbDim}$.  For every $\ell\in[\lbDim]$, we define two request types:
\begin{itemize}
    \item a request of type $\lbA{\ell}$ requests all active resources whose $\ell$-th bit is $1$;
    \item a request of type $\lbB{\ell}$ requests all active resources whose $\ell$-th bit is $0$.
\end{itemize}
Thus, $\lbA{\ell}$ and $\lbB{\ell}$ are request types, not resources.  Every request is single-minded and receives its value only when allocated one copy of every requested resource.

\begin{Observation}[Hypercube capacity identity]
\label{obs:ra-hypercube}
Suppose that an allocation serves $\lbACount{\ell}$ units of type $\lbA{\ell}$ and $\lbBCount{\ell}$ units of type $\lbB{\ell}$, where the quantities may be fractional.  Its maximum consumption over the active resources is
    $\sum_{\ell=1}^{\lbDim}
    \max\{\lbACount{\ell},\lbBCount{\ell}\}$.
\end{Observation}

\begin{proof}
The load on resource $x\in\{0,1\}^{\lbDim}$ is $\sum_{\ell:x_\ell=1}\lbACount{\ell}+\sum_{\ell:x_\ell=0}\lbBCount{\ell}$.  Choosing each bit $x_\ell$ to select the larger of $\lbACount{\ell}$ and $\lbBCount{\ell}$ proves the observation.
\end{proof}

Now, we are ready to prove \Cref{thm:lb-psample-ra}.

\paragraph{The hard instance.} For simplicity of the notations, we define $\lbN:=(2\lceil(20\epsilon)^{-1}\rceil)^2$. Then, we have $\lbN=\Theta(\epsilon^{-2})$. The instance contains $m = 2^z$ resources, and the supply of each resource is $\lbBudget:=\lbDim\lbN/(2p)$, i.e., we have $\lbBudget=\Theta(\log m/(p\epsilon^2))$. For the arriving requests, we create the following three groups for every $\ell\in[\lbDim]$:
\begin{itemize}
    \item \emph{Group 1.} There are $\sqrt{\lbN}/(2p)$ requests of type $\lbA{\ell}$, each with value $2$.
    \item \emph{Group 2.} There are $2\lbN$ labeled macro-groups, each containing $1/(2p)$ identical requests of type $\lbA{\ell}$.  Independently for every macro-group, all its requests have value $1$ or all have value $3$, each with probability $1/2$.
    \item \emph{Group 3.} There are $\lbN/(2p)$ requests of type $\lbB{\ell}$, each with value $4$.
\end{itemize}

In the hard instance, the requests from Group 1 arrive first. Then, requests from Group 2 and Group 3 arrive in arbitrary order. Intuitively, if the $\ell$-th macro-group of Group 2 has value $1$, then the type $A_\ell$ requests from Group 1 should be accepted; otherwise these Group 1 requests need to be rejected. For the optimal offline algorithm, it can learn the values of all macro-groups from Group 2 in advance; however, in the \psample{} model, the online algorithm can only learn this information via $p$-fraction of the samples. We will later show that this information cannot be accurately learned for every group, and therefore every online algorithm must incur an $\Theta(\epsilon) \cdot \lbOptRA$ loss.

 We start from bounding the order of $\lbOptRA$.

\begin{Claim}
\label{claim:psample-ra-opt}
For every realization of the Group-2 values, the hindsight fractional optimum has an integral optimal solution that serves all $\lbN/(2p)$ requests of type $\lbB{\ell}$ and exactly $\lbN/(2p)$ requests of type $\lbA{\ell}$ for every $\ell$. Thus, the hindsight fractional and integral optima coincide. Moreover, we have $5\lbBudget\le\lbOptRA\le7\lbBudget$.
\end{Claim}

\begin{proof}
Write $\lbACount{\ell}$ and $\lbBCount{\ell}$ for the accepted masses of the two types.  By \Cref{obs:ra-hypercube}, feasibility is equivalent to $\sum_\ell\max\{\lbACount{\ell},\lbBCount{\ell}\}\le\lbBudget$.

First, no optimal allocation needs $\lbACount{\ell}>\lbN/(2p)$.  If this occurs, capacity can be freed by dropping type-$\lbA{\ell}$ mass of value at most $3$ and used to accept missing Group-3 mass of value $4$; such missing mass must exist because the total capacity is $\lbDim\lbN/(2p)$.  Once $\lbACount{\ell}\le\lbN/(2p)$ for all $\ell$, every missing Group-3 request can be added without violating feasibility: either it does not increase the corresponding maximum, or the sum of the maxima is strictly below $\lbBudget$.  Thus, all Group-3 requests may be accepted.  With $\lbBCount{\ell}=\lbN/(2p)$, increasing any $\lbACount{\ell}<\lbN/(2p)$ remains feasible and strictly increases value.  This proves the structural statement. For every $\ell$, after fixing the accepted type-$\lbA{\ell}$ mass to be $\lbN/(2p)$, accepting the $\lbN/(2p)$ highest-value requests of this type is integral and weakly increases the value. Therefore, an optimal fractional allocation can be chosen to be integral, and the two optima coincide.

The Group-3 requests contribute $4\lbBudget$.  The accepted type-$\lbA{\ell}$ requests contribute between $\lbBudget$ and $3\lbBudget$, which proves the two value bounds.
\end{proof}

Next, we upper-bound the performance of every online algorithm by giving the online algorithm more power through the following two-stage process.  After observing the sample, it first chooses how many Group-1 requests of each type to accept.  All Group-2 and Group-3 requests are then revealed simultaneously, and the algorithm may compute the best feasible completion.  Let $\lbChoice{\ell}$ be the accepted Group-1 mass of type $\lbA{\ell}$, and let $\lbHighCount{\ell}$ be the number of value-$3$ macro-groups for coordinate $\ell$.  For a real number $r$, we use $r^+:=\max\{r,0\}$.  By \Cref{claim:psample-ra-opt}, the hindsight-optimal first-stage choice is $\lbOptChoice{\ell}
    =
    \frac{1}{2p}
    \min\bigl\{\sqrt{\lbN},(\lbN-\lbHighCount{\ell})^+\bigr\}.$
For a fixed first-stage choice, the best completion loses exactly $|\lbChoice{\ell}-\lbOptChoice{\ell}|$ relative to hindsight in coordinate $\ell$: accepting too many value-$2$ requests displaces value-$3$ requests, while accepting too few replaces value-$2$ requests by value-$1$ requests.

\begin{Claim}
\label{claim:sampled-block-choice}
For every coordinate $\ell$ and sample-measurable $\lbChoice{\ell}\in[0,\sqrt{\lbN}/(2p)]$, 
$
    \mathbb E\bigl[|\lbChoice{\ell}-\lbOptChoice{\ell}|\bigr]
    \ge
    \frac{\sqrt{\lbN}}{3840p}.
$
\end{Claim}

\begin{proof}
Let $\lbTouchedSet{\ell}$ be the set of Group-2 macro-groups touched by the sample, and let $U_\ell$ be the number of sampled Group-2 requests for coordinate $\ell$. Since every touched macro-group contributes at least one sampled request, $|\lbTouchedSet{\ell}|\le U_\ell$. There are $\lbN/p$ such requests, and the sample is a uniformly random $p$-fraction of all requests, so $\mathbb E[U_\ell]=\lbN$. Applying \Cref{cor:concentrationWithoutReplacement} with $\eta=1/4$ and $\delta=1/2$, with probability at least $1/2$, we have
$ U_\ell
    \le
    \frac54\lbN+12\log 4
    \le
    \frac32\lbN,$
where the last inequality uses $\lbN\ge10^4$. Hence, $|\lbTouchedSet{\ell}|\le3\lbN/2$ with probability at least $1/2$.

Conditional on the touched set, let $\lbTouchedImbalance{\ell}$ be the number of touched value-$3$ groups minus the number of touched value-$1$ groups. It is a sum of $k:=|\lbTouchedSet{\ell}|$ independent Rademacher variables. We next show that, conditional on $k\le3\lbN/2$, the event $|\lbTouchedImbalance{\ell}|\le\sqrt{\lbN}/2$ has probability at least $1/4$. If $k\le\lbN/16$, write $\lbTouchedImbalance{\ell}=2V-k$, where $V$ is a sum of $k$ independent $\operatorname{Bernoulli}(1/2)$ random variables. Applying the variance-sensitive bound in \Cref{thm:bernsteins} gives
\begin{align*}
    \textstyle\pr\!\left[|\lbTouchedImbalance{\ell}|>\frac{\sqrt{\lbN}}2\right]
    \le
    2\exp\!\left(
        -\frac{\lbN/16}{k/2+\sqrt{\lbN}/6}
    \right)
    \le
    2e^{-3/2}
    <
    \frac34,
\end{align*}
where the second inequality uses $k\le\lbN/16$ and $\lbN\ge10^4$. Otherwise, \Cref{thm:rademacher-berry-esseen} implies
\begin{equation*}
    \textstyle \pr\!\left[|\lbTouchedImbalance{\ell}|\le\frac{\sqrt{\lbN}}2\right]
    \ge
    2\Phi\!\left(\frac1{\sqrt6}\right)-1-\frac{4.48}{\sqrt{\lbN}}
    >
    \frac14,
\end{equation*}
where $\Phi$ is the standard Gaussian distribution function and the last inequality uses $\lbN\ge10^4$. Thus, both events hold with probability at least $1/8$, and hence at least $1/12$.

Condition on the complete sample, the values of the touched macro-groups, the algorithm's random coins, and these two events. Between $\lbN/2$ and $2\lbN$ macro-groups remain untouched, and their values are still independent fair coins. If their signed value imbalance is denoted by $S$, then $2\lbHighCount{\ell}-2\lbN=\lbTouchedImbalance{\ell}+S$. Since $|\lbTouchedImbalance{\ell}|\le\sqrt{\lbN}/2$, another application of \Cref{thm:rademacher-berry-esseen} gives
\begin{equation*}
    \textstyle \pr\!\left[\lbHighCount{\ell}\ge\lbN+\frac{\sqrt{\lbN}}4\right]
    \ge
    1-\Phi(\sqrt2)-0.56\sqrt{\frac2{\lbN}}
    >
    \frac1{20}
\end{equation*}
and, by symmetry, we also have
\[
    \textstyle\pr\!\left[\lbHighCount{\ell}\le\lbN-\frac{\sqrt{\lbN}}4\right]
    >
    \frac1{20}.
\]
On the first event, $\lbOptChoice{\ell}=0$; on the second, $\lbOptChoice{\ell}\ge\sqrt{\lbN}/(8p)$.  If $\lbChoice{\ell}\ge\sqrt{\lbN}/(16p)$, the first event causes loss at least $\sqrt{\lbN}/(16p)$; otherwise, the second event causes the same loss.  The conditional expected loss is therefore at least $\sqrt{\lbN}/(320p)$.  Multiplying by $1/12$ proves the claim.
\end{proof}

Now, we are ready to finish the proof of \Cref{thm:lb-psample-ra}. Summing \Cref{claim:sampled-block-choice} over the coordinates gives
\begin{equation}
    \textstyle \mathbb E[\lbOptRA-\lbAlgRA]
    \ge
    \frac{\lbDim\sqrt{\lbN}}{3840p}.
    \label{eq:psample-ra-gap}
\end{equation}
This inequality averages over the random Group-2 values in addition to the sample and the algorithm's coins.  Averaging over those values fixes a deterministic instance for which \eqref{eq:psample-ra-gap} continues to hold over the random sample and the algorithm's coins.  Since $\sqrt{\lbN}\le1/(5\epsilon)$ for $\epsilon\le10^{-3}$, the additive gap in \eqref{eq:psample-ra-gap} is at least
$
    \frac{\lbBudget}{1920\sqrt{\lbN}}
    \ge
    \frac{\epsilon\lbBudget}{400}.
$
Since $\lbOptRA\le7\lbBudget$ by \Cref{claim:psample-ra-opt}, this gap is at least $\epsilon\lbOptRA/2800$, which proves the theorem.

\subsection{Load Balancing Lower Bounds}
\label{sec:lb-sparse-lb}

\newcommand{\lbFlexibleType}{F}
\newcommand{\lbExclusiveType}[1]{E_{#1}}

In this subsection, we prove \Cref{thm:lb-one-sparse-iid} and \Cref{thm:lb-one-sparse-psample}. Both load-balancing hard instances only contain the following two types of requests. Here, $\mathbf e_j$ denotes the $j$-th standard basis vector in $\mathbb R^m$.
\begin{itemize}
    \item A request of type $\lbFlexibleType$ has action set $[m]$, where selecting action $j\in[m]$ consumes one unit of resource $j$.
    \item For every $j\in[m]$, a request of type $\lbExclusiveType{j}$ has only single feasible action, which consumes one unit of resource $j$.
\end{itemize}
Thus, every decision in the two hard instances generates unit load on a single resource. With these two types of requests, we use the following elementary characterization of the hindsight integral optimum.

\begin{Observation}
\label{obs:one-sparse-optimum}
Consider $\lbNumJobs$ unit requests, each of which has type $\lbFlexibleType$ or $\lbExclusiveType{j}$ for some $j\in[m]$. If $\lbForcedCount{j}$ is the number of requests of type $\lbExclusiveType{j}$, then the hindsight integral optimum is
$
    \textstyle\max\left\{\left\lceil\frac{\lbNumJobs}{m}\right\rceil,\max_{j\in[m]}\lbForcedCount{j}\right\}.
$
\end{Observation}

\begin{proof}
The right-hand side is a lower bound by averaging and because every request of type $\lbExclusiveType{j}$ must generate load on resource $j$. After allocating all requests of type $\lbExclusiveType{j}$ to their unique decisions, the total residual capacity under the makespan on the right-hand side is at least the number of requests of type $\lbFlexibleType$. This residual capacity consists of integral unit slots, and every request of type $\lbFlexibleType$ can generate load on any single resource. Therefore, these requests can be integrally allocated within the residual capacity.
\end{proof}

\subsubsection{Known i.i.d. arrivals}

We first introduce the hard instance we use to prove \Cref{thm:lb-one-sparse-iid}. Let $\lbNumJobs:=100m\eps^{-2}\log m$, and consider $\lbNumJobs$ i.i.d. requests from the following known distribution. Every request has type $\lbFlexibleType$ with probability $\epsilon$. Otherwise, a resource $j$ is selected uniformly from $[m]$, and the request has type $\lbExclusiveType{j}$.

\begin{proof}[Proof of \Cref{thm:lb-one-sparse-iid}]
We show that the above instance satisfies the desired properties in \Cref{thm:lb-one-sparse-iid}. For every resource $j$, the random variable $\lbForcedCount{j}$ is a sum of independent Bernoulli random variables with mean
$
    \mu_j
    :=
    \mathbb E[\lbForcedCount{j}]
    =
    \frac{(1-\epsilon)\lbNumJobs}{m}.
$
Applying \Cref{thm:bernsteins} with $t:=\lbNumJobs/m-\mu_j=100\eps^{-1}\log m$, we obtain
\begin{align*}
    \textstyle\pr\!\left[\lbForcedCount{j}>\frac{\lbNumJobs}{m}\right]
    \le
    2\exp\!\left(-\frac{t^2}{2\mu_j+t}\right)
    \le
    2\exp(-50\log m)
    \le
    m^{-20}.
\end{align*}
A union bound gives $\pr[\max_j\lbForcedCount{j}>\lbNumJobs/m]\le m^{-19}$. By \Cref{obs:one-sparse-optimum}, on the complementary event, the hindsight integral optimum is $\lceil\lbNumJobs/m\rceil$. Since every makespan is at most $\lbNumJobs$, for all sufficiently large $m$, the expected hindsight integral optimum satisfies
\begin{equation}
    \textstyle\lbOptLB
    \le
    \frac{\lbNumJobs}{m}
    +
    \frac{\log m}{1000\epsilon}.
    \label{eq:one-sparse-iid-opt}
\end{equation}

Consider the moment before the final $3m\log(m)/(50\epsilon^2)$ requests arrive. Let $\lbMachineLoad{j}$ be the load already generated on resource $j$, and set $\lbCapacity{j}:=\lbNumJobs/m-\lbMachineLoad{j}$. The residual capacities satisfy $\sum_j\lbCapacity{j}=3m\log(m)/(50\epsilon^2)$. We further strengthen the algorithm by revealing the complete suffix at once and allowing it to select the decisions for all suffix requests of type $\lbFlexibleType$ after observing all suffix requests of type $\lbExclusiveType{j}$. If $\lbOccupancy{j}$ is the number of suffix requests of type $\lbExclusiveType{j}$, then we have
\begin{equation}
    \textstyle\lbAlgLB
    \ge
    \frac{\lbNumJobs}{m}
    +
    \left(\max_{j\in[m]}(\lbOccupancy{j}-\lbCapacity{j})\right)^+.
    \label{eq:one-sparse-iid-alg}
\end{equation}
The suffix is independent of the entire prior history, including any independent training data given to the algorithm. We next prove that, conditionally on every prior history, we have
\begin{equation*}
    \textstyle\pr\!\left[
        \lbAlgLB
        \ge
        \frac{\lbNumJobs}{m}
        +
        \frac{\log m}{25\epsilon}
        \,\middle|\,
        \text{prior history}
    \right]
    \ge
    \frac12.
\end{equation*}
Fix the prior history, and set $x:=3\log(m)/(50\epsilon^2)$ and $d:=\log m/(25\epsilon)$. If $\lbCapacity{j}\le-d$ for some resource $j$, then the desired inequality follows deterministically from \eqref{eq:one-sparse-iid-alg}. Thus, suppose that $\lbCapacity{j}>-d$ for every resource $j$. Since $\sum_j\lbCapacity{j}=mx$, the number of resources satisfying $\lbCapacity{j}\le x+d/2$ is at least $md/(2x+3d)\ge\epsilon m/4$. Indeed, otherwise summing the lower bound $-d$ over these resources and the lower bound $x+d/2$ over all remaining resources would give $\sum_j\lbCapacity{j}>mx$.

Fix any such resource $j$. The random variable $\lbOccupancy{j}$ has distribution $\operatorname{Bin}(mx,(1-\epsilon)/m)$ and mean $\mu:=(1-\epsilon)x$. The event $\lbOccupancy{j}\ge x+3d/2$ implies $\lbOccupancy{j}-\lbCapacity{j}\ge d$. For sufficiently large $m$, the parameters $q=(1-\epsilon)/m$, $\mu$, and $t:=\epsilon x+3d/2$ satisfy the assumptions of \Cref{lem:binomial-upper-tail-lower-bound}. Therefore,
\begin{equation*}
    \textstyle \pr\!\left[\lbOccupancy{j}\ge x+\frac{3d}{2}\right]
    \ge
    \frac{1}{20\sqrt{(1-\epsilon)x}}
    \exp\!\left(-\frac{3(\epsilon x+3d/2)^2}{(1-\epsilon)x}\right)
    \ge
    m^{-3/4}
\end{equation*}
for all sufficiently large $m$. If $Y$ is the number of candidate resources satisfying $\lbOccupancy{j}\ge x+3d/2$, then \Cref{lem:multinomial-exceedance} gives
$\pr[Y=0]
    \le
    \frac{1}{\mathbb E[Y]}.$
Since $\mathbb E[Y]\ge\epsilon m^{1/4}/4$, this probability is at most $1/2$ for all sufficiently large $m$. This proves the displayed conditional bound.

Consequently, $\mathbb E[\lbAlgLB]\ge\lbNumJobs/m+\log(m)/(50\epsilon)$. Together with \eqref{eq:one-sparse-iid-opt} and $\lbNumJobs/m=100\eps^{-2}\log m$, we have
\begin{equation*}
    \textstyle\mathbb E[\lbAlgLB]-\lbOptLB
    \ge
    \frac{19\log m}{1000\epsilon}
    \ge
    10^{-4}\epsilon\cdot \lbOptLB,
\end{equation*}
which proves the theorem.
\end{proof}

\subsubsection{The $p$-sample model}

We first introduce the hard instance we use to prove \Cref{thm:lb-one-sparse-psample}. Let $\lbNumGroups:=100m\eps^{-2}\log m$. We construct $\lbNumGroups$ labeled macro-groups, each containing $1/(2p)$ identical unit requests - we assume without loss of generality that $(2p)^{-1}$ is an integer. Independently for every macro-group, all its requests have type $\lbFlexibleType$ with probability $\epsilon$. Otherwise, a resource $j$ is selected uniformly from $[m]$, and all requests in the macro-group have type $\lbExclusiveType{j}$. In the online order, all macro-groups of type $\lbFlexibleType$ precede all other macro-groups. The macro-group partition and labels are known to the algorithm. The total number of requests is $\lbNumJobs:=\lbNumGroups/(2p)$, so $\lbNumJobs/m=50\log m/(p\epsilon^2)$.

\begin{proof}[Proof of \Cref{thm:lb-one-sparse-psample}]
We show that the above random construction contains a fixed deterministic instance with the desired properties in \Cref{thm:lb-one-sparse-psample}. Call a realization of the construction balanced if, for every resource $j$, the number of macro-groups of type $\lbExclusiveType{j}$ is at most
$
    (1-\nicefrac{\epsilon}{2})
    100\eps^{-2}\log m.
$
For a fixed resource $j$, let $\lbForcedCount{j}$ be the number of macro-groups of type $\lbExclusiveType{j}$. This is a sum of independent Bernoulli random variables with mean
$
    \mu_j
    :=
    \mathbb E[\lbForcedCount{j}]
    =
    100\eps^{-2}(1-\epsilon)\log m.
$
The balanced threshold exceeds $\mu_j$ by $t:=50\eps^{-1}\log m$. Hence, by \Cref{thm:bernsteins},
\begin{align*}
    \textstyle\pr\!\left[
        \lbForcedCount{j}
        >
        \left(1-\frac{\epsilon}{2}\right)
        \frac{100\log m}{\epsilon^2}
    \right]
    \le
    2\exp\!\left(-\frac{t^2}{2\mu_j+t}\right)
    \le
    2\exp\!\left(-\frac{50}{4-3\epsilon}\log m\right)
    \le
    m^{-6}.
\end{align*}
A union bound shows that the realization is balanced with probability at least $1-m^{-5}$. For every balanced realization, the load generated by the requests of type $\lbExclusiveType{j}$ on every resource $j$ is smaller than $\lbNumJobs/m$. Hence, by \Cref{obs:one-sparse-optimum}, we have
\begin{equation}
    \lbOptLB
    =
    \left\lceil\frac{\lbNumJobs}{m}\right\rceil.
    \label{eq:one-sparse-psample-opt}
\end{equation}

The sample contains $\lbNumGroups/2$ requests and therefore touches at most $\lbNumGroups/2$ macro-groups. Thus, at least $\lbNumGroups/2$ macro-groups remain untouched. Conditional on the set of untouched macro-groups, let $\lbNumHiddenGroups$ be the number of untouched macro-groups whose requests have type $\lbExclusiveType{j}$ for some $j\in[m]$. If the untouched set contains $u$ macro-groups, then $\lbNumHiddenGroups\sim\operatorname{Bin}(u,1-\epsilon)$ and has mean $\mu:=(1-\epsilon)u\ge(1-\epsilon)\lbNumGroups/2$. Applying \Cref{thm:bernsteins} with $t:=\mu-\lbNumGroups/4$ gives
\begin{align*}
    \textstyle\pr\!\left[\lbNumHiddenGroups<\frac{\lbNumGroups}{4}\right]
    \le
    2\exp\!\left(-\frac{t^2}{2\mu+t}\right)
    \le
    2\exp\!\left(-\frac{\lbNumGroups}{24}\right)
    \le
    \exp\!\left(-\frac{\lbNumGroups}{40}\right),
\end{align*}
where the second inequality uses that $(\mu-\lbNumGroups/4)^2/(3\mu-\lbNumGroups/4)$ is increasing in $\mu>\lbNumGroups/4$ and $\epsilon\le10^{-3}$. Therefore, we have
\begin{equation}
    \textstyle\pr\!\left[\lbNumHiddenGroups\ge\frac{\lbNumGroups}{4}\right]
    \ge
    1-e^{-\lbNumGroups/40}.
    \label{eq:one-sparse-many-hidden-groups}
\end{equation}

Consider the moment immediately after all requests of type $\lbFlexibleType$ have been allocated. Every touched macro-group whose requests have type $\lbExclusiveType{j}$ has already revealed the corresponding resource $j$ through the sample. Let $\lbMachineLoad{j}$ be the load allocated to resource $j$ by the requests of type $\lbFlexibleType$, plus the inevitable future load from the touched macro-groups of type $\lbExclusiveType{j}$, and set $\lbCapacity{j}:=\lbNumJobs/m-\lbMachineLoad{j}$. Conditional on the complete history and on $\lbNumHiddenGroups$, the resources associated with the untouched non-flexible macro-groups remain independent and uniform over $[m]$. If $\lbOccupancy{j}$ is the number of these macro-groups whose requests have type $\lbExclusiveType{j}$, then $\sum_j\lbCapacity{j}=\lbNumHiddenGroups/(2p)$ and
\begin{equation}
    \lbAlgLB
    =
    \frac{\lbNumJobs}{m}
    +
    \max_{j\in[m]}
    \left(\frac{\lbOccupancy{j}}{2p}-\lbCapacity{j}\right).
    \label{eq:one-sparse-psample-alg}
\end{equation}
On the event in \eqref{eq:one-sparse-many-hidden-groups}, we have $25m\log m/\epsilon^2\le\lbNumHiddenGroups\le100m\log m/\epsilon^2$. We next prove that, conditionally on the complete history and $\lbNumHiddenGroups$, we have
\begin{equation*}
    \textstyle\pr\!\left[
        \lbAlgLB
        \ge
        \frac{\lbNumJobs}{m}
        +
        \frac{\log m}{2p\epsilon}
        \,\middle|\,
        \text{complete history},\lbNumHiddenGroups
    \right]
    \ge
    \frac12.
\end{equation*}
Set $x:=\lbNumHiddenGroups/m$ and $d:=\log m/\epsilon$. Since $\sum_j2p\lbCapacity{j}=\lbNumHiddenGroups=mx$, if $2p\lbCapacity{j}\le-d$ for some resource $j$, then the desired inequality follows deterministically from \eqref{eq:one-sparse-psample-alg}. Thus, suppose that $2p\lbCapacity{j}>-d$ for every resource $j$. The number of resources satisfying $2p\lbCapacity{j}\le x+d/2$ is at least $md/(2x+3d)\ge\epsilon m/201$. This follows from the same averaging argument as above and the upper bound $x\le100\log m/\epsilon^2$.

Fix any such resource $j$. The random variable $\lbOccupancy{j}$ has distribution $\operatorname{Bin}(\lbNumHiddenGroups,1/m)$ and mean $x$. The event $\lbOccupancy{j}\ge x+3d/2$ implies $\lbOccupancy{j}-2p\lbCapacity{j}\ge d$. For sufficiently large $m$, the parameters $q=1/m$, $\mu=x$, and $t:=3d/2$ satisfy the assumptions of \Cref{lem:binomial-upper-tail-lower-bound}. Therefore,
\begin{equation*}
    \textstyle\pr\!\left[\lbOccupancy{j}\ge x+\frac{3d}{2}\right]
    \ge
    \frac{1}{20\sqrt{x}}
    \exp\!\left(-\frac{27d^2}{4x}\right)
    \ge
    m^{-1/3}
\end{equation*}
for all sufficiently large $m$, where the last inequality uses $x\ge25\log m/\epsilon^2$. If $Y$ is the number of candidate resources satisfying $\lbOccupancy{j}\ge x+3d/2$, then \Cref{lem:multinomial-exceedance} gives
$
    \pr[Y=0]
    \le
    \mathbb E[Y]^{-1}.
$
Since $\mathbb E[Y]\ge\epsilon m^{2/3}/201$, this probability is at most $1/2$ for all sufficiently large $m$. This proves the displayed conditional bound.

For the random choice of macro-group types and their associated resources, the last event and the balanced event occur jointly with probability at least $2/5$ for all sufficiently large $m$. Averaging fixes a balanced deterministic instance for which the displayed makespan gap occurs with probability at least $2/5$ over the sample and the algorithm's coins. On this instance, by \eqref{eq:one-sparse-psample-opt}, for all sufficiently large $m$, we have
\begin{equation*}
    \textstyle\mathbb E[\lbAlgLB]
    \ge
    \lbOptLB
    +
    \frac{\log m}{5p\epsilon}-1
    \ge
    \lbOptLB
    +
    \frac{\log m}{6p\epsilon}
    \ge
    (1+10^{-3}\epsilon) \cdot \lbOptLB,
\end{equation*}
where the last inequality uses $\lbOptLB\le 1+50 p^{-1}\eps^{-2}\log m$. This proves the theorem.
\end{proof}

\paragraph{Acknowledgments.}
The authors used GPT-5.5 Pro for assistance with writing and to  explore proof strategies. The authors have checked the arguments and take full responsibility for all content of the paper.

\appendix

\section{Missing Proofs}

\subsection{Concentration Inequalities}

\begin{Theorem}[Bernstein's inequality, e.g., Lemma 2.2.9, \cite{weakConvergence}]\label{thm:bernsteins}
    Suppose that $Y = \{Y_i\}_{i\in[n]}$ is a collection of independent random variables taking values in $[0,1]$. Denote $Y_{[n]} = \sum_{i\in[n]} Y_i$, $\mu = \frac{1}{n}\sum_{i\in[n]}\E[Y_i]$, and $\sigma^2 = \frac{1}{n}\sum_{i\in[n]}\operatorname{Var}(Y_i)$. Then, for any $t > 0$,
    \begin{align*}
        \pr[|Y_{[n]} - n\mu | > t] 
        \leq 2 \exp\left(- \frac{t^2}{2 n\sigma^2 + \nicefrac{2 t}{3}}\right)
        \leq 2 \exp\left(- \frac{t^2}{2 n\mu + t}\right)
    \end{align*}
\end{Theorem}

\begin{Theorem}[Theorem 2.14.19 in \cite{weakConvergence}]\label{thm:concentrationWithoutReplacement}
Let $Y = \{Y_i\}_{i\in[n]}$ be a set of real numbers in $[0,1]$, 
and let $S$ be a uniformly random subset of the index set $[n]$ of size $s$, drawn without replacement. Then, denoting $Y_S = \sum_{i\in S} Y_i$, $\mu = \frac{1}{n} \sum_{i\in[n]} Y_i$, and $\sigma^2 = \frac{1}{n}\sum_{i\in[n]} (Y_i - \mu)^2$, for any $t> 0$,
\begin{align*}
    \pr[|Y_S - s\mu| > t] 
    \leq 2\exp\left(-\frac{t^2}{2 s \sigma^2 + t}\right)
    \leq 2\exp\left(-\frac{t^2}{2 s \mu + t}\right)
\end{align*}
where the second inequality uses the fact that, since $Y_i \in [0,1]$, $\sigma^2 \leq \mu$.
\end{Theorem}
\begin{Corollary}[of \cref{thm:bernsteins} and \ref{thm:concentrationWithoutReplacement}]\label{cor:concentrationWithoutReplacement}
    In the same setting as \cref{thm:bernsteins} or \ref{thm:concentrationWithoutReplacement}, for any $\eta, \delta\in(0,1)$, with probability at least $1-\delta$,
    \begin{align*}
        |Y_S - s\mu| \leq \eta s \mu + 3\eta^{-1}\log(\nicefrac{2}{\delta}),
    \end{align*}
    where $S = [n]$ and $s=n$ in the setting of \cref{thm:bernsteins}, and $S$ is a uniformly random subset of $[n]$ of size $s$ in the setting of \cref{thm:concentrationWithoutReplacement}.
\end{Corollary}
\begin{proof}
    In the context of {\cref{thm:bernsteins} or \ref{thm:concentrationWithoutReplacement}}, choose $t=\max\{\eta s \mu, 3\eta^{-1}\log(\nicefrac{2}{\delta})\}$, and consider the two cases that (i) $s \mu \leq \frac{3\log(\nicefrac{2}{\delta})}{\eta^2}$, and (ii) $s \mu > \frac{3\log(\nicefrac{2}{\delta})}{\eta^2}$. In case $(i)$, the corresponding bound from \Cref{thm:bernsteins,thm:concentrationWithoutReplacement} implies that
\begin{align*}
    \pr[|Y_S - s\mu| > 3\eta^{-1}\log(\nicefrac{2}{\delta})] 
    \leq 2\exp\left(-\frac{3}{2+\eta}\log(\nicefrac{2}{\delta})\right)
    \leq \delta.
\end{align*}
In case (ii), we similarly have
\begin{align*}
    \pr[|Y_S - s\mu| > \eta s \mu] 
    \leq 2\exp\left(-\frac{\eta^2 s \mu}{2 +\eta}\right)
    \leq 2\exp\left(-\frac{3}{2+\eta}\log(\nicefrac{2}{\delta})\right)
    \leq\delta.
\end{align*}
Combining these two cases establishes the claim.
\end{proof}

\begin{Theorem}[Berry--Esseen inequality for Rademacher sums; \cite{esseen1942liapunov,shevtsova2010improvement}]
\label{thm:rademacher-berry-esseen}
Let $k\ge1$, let $R_1,\ldots,R_k$ be independent Rademacher random variables, let $S_k:=\sum_{i=1}^k R_i$, and let $\Phi$ be the standard Gaussian distribution function. Then, for every $x\in\mathbb R$,
\begin{equation*}
    \textstyle\left|
        \pr\!\left[\frac{S_k}{\sqrt{k}}\le x\right]-\Phi(x)
    \right|
    \le
    \frac{0.56}{\sqrt{k}}.
\end{equation*}
Consequently, for every $a\ge0$,
\begin{align*}
    \textstyle\pr[|S_k|\le a]
    \ge
    2\Phi\!\left(\frac{a}{\sqrt{k}}\right)-1-
    \frac{1.12}{\sqrt{k}} \quad \text{and} \quad 
    \pr[S_k\ge a]
    \ge
    1-\Phi\!\left(\frac{a}{\sqrt{k}}\right)-
    \frac{0.56}{\sqrt{k}}.
\end{align*}
\end{Theorem}

\begin{Lemma}[A binomial upper-tail lower bound]
\label{lem:binomial-upper-tail-lower-bound}
Let $X\sim\operatorname{Bin}(r,q)$ and let $\mu:=rq$. If $q\le1/4$, $\mu\ge2$, and $2\le t\le\mu/2$, then
\begin{equation*}
    \textstyle\pr[X\ge\mu+t]
    \ge
    \frac{1}{20\sqrt{\mu}}
    \exp\!\left(-\frac{3t^2}{\mu}\right).
\end{equation*}
\end{Lemma}
\begin{proof}
Let $k:=\lceil\mu+t\rceil$ and $a:=k/r$. The factorial bounds
\begin{equation*}
    \textstyle\sqrt{2\pi}\,h^{h+1/2}e^{-h}
    \le
    h!
    \le
    e\,h^{h+1/2}e^{-h}
    \qquad(h\ge1)
\end{equation*}
imply
\begin{equation*}
    \textstyle\pr[X=k]
    \ge
    \frac{1}{4\sqrt{ra(1-a)}}
    \exp\!\left(-rD(a\|q)\right),
\end{equation*}
where $D(a\|q):=a\log(a/q)+(1-a)\log((1-a)/(1-q))$ is the binary relative entropy. Since $t\le\mu/2$, $t\ge2$, and $q\le1/4$, we have $k\le2\mu$, $a\le1/2$, and
\begin{align*}
    \textstyle r D(a\|q)
    ~\le~
    \frac{(k-\mu)^2}{\mu(1-q)}
    ~\le~
    \frac{3t^2}{\mu},
\end{align*}
where the first inequality uses $D(a\|q)\le(a-q)^2/(q(1-q))$, and the second uses $k-\mu\le t+1\le3t/2$. Moreover, $ra(1-a)\le k\le2\mu$, so $1/(4\sqrt{ra(1-a)})\ge1/(20\sqrt{\mu})$. Combining these bounds with $\pr[X\ge\mu+t]\ge\pr[X=k]$ proves the claim.
\end{proof}

\begin{Lemma}[Multinomial exceedance indicators]
\label{lem:multinomial-exceedance}
Let $(X_0,X_1,\ldots,X_m)$ be the occupancy counts from $r$ independent trials with cell probabilities $(q_0,q_1,\ldots,q_m)$. For any $J\subseteq[m]$ and thresholds $(a_j)_{j\in J}$, define
\begin{equation*}
    \textstyle Y:=\sum_{j\in J}\mathbf 1[X_j\ge a_j].
\end{equation*}
Then, we have $\operatorname{Var}(Y)\le\E[Y]$. Consequently, whenever $\E[Y]>0$, we have $\pr[Y=0]\le\frac{1}{\E[Y]}.$
\end{Lemma}
\begin{proof}
Fix distinct coordinates $i,j\in J$. If $q_i=1$, all other counts are deterministic and the claim is immediate. Otherwise, conditional on $X_i=h$, $X_j$ has distribution
$
    \operatorname{Bin}\!\left(r-h,\nicefrac{q_j}{(1-q_i)}\right).
$
Hence, $h\mapsto\pr[X_j\ge a_j\mid X_i=h]$ is nonincreasing, while $h\mapsto\mathbf 1[h\ge a_i]$ is nondecreasing. For any random variable $H$, any nondecreasing function $f$, and any nonincreasing function $g$, an independent copy $H'$ gives
\begin{equation*}
    2\operatorname{Cov}(f(H),g(H))
    =
    \E[(f(H)-f(H'))(g(H)-g(H'))]
    \le0.
\end{equation*}
It follows that the indicators $\mathbf 1[X_i\ge a_i]$ and $\mathbf 1[X_j\ge a_j]$ have nonpositive covariance. Therefore,
\begin{equation*}
    \textstyle\operatorname{Var}(Y)
    \le
    \sum_{j\in J}\operatorname{Var}(\mathbf 1[X_j\ge a_j])
    \le
    \sum_{j\in J}\E[\mathbf 1[X_j\ge a_j]]
    =
    \E[Y].
\end{equation*}
Finally,
\begin{equation*}
    \operatorname{Var}(Y)
    =
    \E[(Y-\E[Y])^2]
    \ge
    \pr[Y=0]\,\E[Y]^2,
\end{equation*}
which proves the probability bound.
\end{proof}

\subsection{Properties of Tangent Sequence: Proof of \Cref{lem:algConnectionProperties}}
\label{sec:appendix:tangentSeqFacts}

Let us begin by recalling the setup of the \preview{} model. The samples and requests are generated as follows: let $(\sample_1,\ldots,\sample_n) \sim \bD = \D_1 \times \ldots \times \D_n$ and $(\online_1,\ldots,\online_n) \sim \bD$ be sampled independently, and let $\samples \subseteq [n]$ be a uniformly random subset of $\nsamples = p\cdot n$ indices. Offline, $\nsamples$ samples $(\sample_t)_{t\in\samples}$ are revealed to the algorithm. Then, online, the requests $(\online_t)_{t\in[n]}$ are revealed sequentially in the original order. Let $\samplesnot = [n] \setminus \samples$ denote the set of unsampled request indices, and notice that $\samplesnot$ is a uniformly random subset of size $(1-p)n$. By construction, both $\samples$ and $\samplesnot$ are measurable with respect to the sampling history $\histend$.

Instead of viewing the sample indices $\samples$, the corresponding samples $(\sample_t)_{t\in\samples}$, and the random ordering $\pi$ as being chosen before running $\blackbox$ on the sample as in \cref{alg:genericOffline}, we can equivalently view these quantities as being sampled iteratively, as in \cref{alg:equivGenericOffline}:

\begin{algorithm}
\caption{\textsc{Equivalent Offline Random-Order Simulation}}
\label{alg:equivGenericOffline}
\KwIn{Black-box algorithm $\blackbox$, number of samples $\nsamples$, total number of requests $n$}
Let $K$ be a random variable in $\{\lfloor (1-p)\nsamples \rfloor, \lceil (1-p)\nsamples \rceil\}$ such that $\E[K]=(1-p)\nsamples$.\\
Given $K$, let $I$ be a uniformly random subset of $[\nsamples]$ of size $K$.\\
Initialize $U_1 = [n]$\\
\For{$i = 1,\ldots,\nsamples$}
{
    Sample $\pi(i) \sim \unif(U_i)$, take $\hist_{\pi(i)} = (\sample_{\pi(j)}, \decisionsample[\pi(j)])_{j<i}$, and sample $\sample_{\pi(i)}, \online_{\pi(i)} \stackrel{iid}{\sim}\D_{\pi(i)}$\\
    Record $\decisionsample = \decision[\pi(i)]{\sample_{\pi(i)}}$ and $\decisiononline[\pi(i)] = \decision[\pi(i)]{\online_{\pi(i)}}$\\
    \uIf{$i\in I$}{
        Sample $\hpi(i) \sim \unif(U_i \setminus \{\pi(i)\})$ and sample $\online_{\hpi(i)} \sim \D_{\hpi(i)}$\\
        Record $\decisiongeneric[i] = \decision[\pi(i)]{\online_{\hpi(i)}}$\\
        Take $U_{i+1}=U_i \setminus \{\pi(i), \hpi(i)\}$\\
    }\Else{
        Take $U_{i+1}=U_i \setminus \{\pi(i)\}$\\
    }
}
Take $\samples = \{\pi(i)\}_{i\in[\nsamples]}$, $\samplesnotp = \{\hpi(i)\}_{i\in I}$, $\samplesnot = [n]\setminus\samples$, and $\histend=\sigma\{(\pi(i), \sample_{\pi(i)}, \decisionsample)\}$\\
\KwOut{$(\pi(i), \hist_{\pi(i)})_{i\in [\nsamples]}$}
\end{algorithm}

Let us record a few useful observations:
\begin{Observation}\label{obs:equivOffline}
    The samples generated by \cref{alg:equivGenericOffline} satisfy the \preview{} assumptions, and the permutation $\pi$, samples $(\sample_{\pi(i)})_{i\in[\nsamples]}$, sample decisions $\decisionsample$, and sample histories are identically distributed to those in \cref{alg:genericOffline}. The online samples $\online_{\pi(i)}$ and decisions $\decisiononline[\pi(i)]$ for each $i\in[\nsamples]$ are also identically distributed to the corresponding samples and decisions at iterations satisfying $t=\pi(i)$ in \cref{alg:genericOnline}.

    Moreover, for any $i\in [\nsamples]$, $\pr[i\in I] = 1-p$.
\end{Observation}
\begin{proof}
    Notice that sampling $\pi(i)$ and $\hpi(i)$ as defined in \cref{alg:equivGenericOffline} is always well-defined. Indeed, since $\nsamples = p n$ is an integer, and thus also $(1-p)n$, and $p \leq 1$,
    \begin{align*}
        \nsamples + |I|
        = \nsamples + K
        \leq \nsamples + \lceil (1-p)\nsamples \rceil
        = \nsamples + \lceil p(1-p)n \rceil
        \leq \nsamples + (1-p)n 
        = \nsamples + n - \nsamples 
        = n.
    \end{align*}
    As a consequence, $\pi(1),\ldots,\pi(\nsamples)$ must be $\nsamples$ distinct samples from $[n]$, and $(\hpi(i))_{i\in I}$ similarly must be a set of $K\leq n-\nsamples$ distinct samples from $[n] \setminus \samples$.

    Then, the fact that $\samples = \{\pi(1),\ldots,\pi(\nsamples)\}$ are sampled uniformly without replacement from $[n]$ follows by symmetry: relabeling the elements of $[n]$ with any permutation $\sigma$ does not change the distribution of $K$, $I$, or $\pi(1)\sim \unif([n])$. Thus, since the distribution of $\pi(i+1)$ is completely determined by $[n]$, $I$, and $\pi(1),\ldots,\pi(i)$, it follows that the distribution of $\pi(1),\ldots,\pi(\nsamples)$ and of $\sigma(\pi(1)),\ldots,\sigma(\pi(\nsamples))$ is the same, and thus must be a uniformly random size $\nsamples$ sample without replacement from $[n]$.

    As a consequence, since $\sample_{\pi(i)}, \online_{\pi(i)} \stackrel{iid}{\sim}\D_{\pi(i)}$ are sampled independently, the resulting $\decisionsample$, $\decisiononline[\pi(i)]$ and all sample histories are identically distributed in \cref{alg:genericOffline,alg:equivGenericOffline}.

    The second statement follows since, for any $i\in[\nsamples]$,
    \begin{align*}
        \textstyle(1-p)\nsamples 
        = \E[K]
        = \E[|I|]
        = \sum_{i'\in[\nsamples]} \pr[i'\in I]
        = \nsamples \pr[i\in I],
    \end{align*}
    where the last equality follows since $I$ is sampled uniformly, and thus by symmetry, $\pr[i\in I]=\pr[i'\in I]$ for all $i, i'\in[\nsamples]$.
\end{proof}

\paragraph{Connecting the offline \ro{} and online decisions via a tangent sequence.}

In order to establish the claims in \cref{lem:algConnectionProperties}, we need to connect the decisions made online by \cref{alg:genericOnline} to the decisions made offline in random order in \cref{alg:genericOffline}. Using the equivalence from \cref{obs:equivOffline}, we can reason instead about the execution of \cref{alg:equivGenericOffline}. This allows us to construct the following sequence of random variables, which is the key to connecting the performance of \cref{alg:genericOffline,alg:genericOnline}.
\begin{Definition}\label{def:tangent}
Let $f : \decisionsetall \times \typespaceall \to \R$ be any measurable function.
For any iteration $i\in[\nsamples]$ of \cref{alg:equivGenericOffline}, let $\sample_{\pi(i)}, \online_{\pi(i)} \sim \D_{\pi(i)}$ be the sample and observation generated by \cref{alg:equivGenericOffline}, and let $\decisionsample, \decisiononline[\pi(i)]$ the associated decisions. If $i\in I$, let $\online_{\hpi(i)}\sim \D_{\hpi(i)}$ and $\decisiongeneric[i]$ be as in \cref{alg:equivGenericOffline}.
Then, let
\begin{align*}
    X_i = \begin{cases}
        f(\decisiononline[\pi(i)], \online_{\pi(i)}) & \text{if $i \not\in I$}\\
        f(\decisiongeneric[i], \online_{\hpi(i)}) & \text{if $i \in I$}
    \end{cases}
    \quad\text{and}\quad
    Y_i = f(\decisionsample, \sample_{\pi(i)})
\end{align*}
Let $\F_i$ be the filtration corresponding to all observations of \cref{alg:equivGenericOffline} at the end of step $i$, where $\F_0 = \sigma\{K, I\}$.
\end{Definition}
The key observation is that the random variables $(X_i, Y_i)$ from \cref{def:tangent} satisfy the following tangent relation:
\begin{Lemma}\label{lem:tangent}
    In the setting of \cref{def:tangent}, the random variables $(X_i, Y_i)$ are adapted to $\F_i$, and the distributions of $X_i$ and $Y_i$ conditioned on $\F_{i-1}$ are equal. In particular,
\[
        \E[X_i \mid \F_{i-1}] = \E[Y_i \mid \F_{i-1}].
  \]
    Moreover, taking $\histend$ to be the sample history as defined in \cref{alg:genericOffline,alg:equivGenericOffline},
    \begin{align*}
        \textstyle\sum_{i\in[\nsamples]}\E[X_i \mid \histend]
        = p \sum_{t\in [n]} \E[f(\decisiononline, \online_t) \mid \histend].
    \end{align*}
\end{Lemma}

We note that tangent sequences satisfy the following concentration inequality, which will be useful in establishing \cref{lem:algConnectionProperties}.

\begin{Theorem}\label{thm:tangenConcentration}
Suppose that $X_t, Y_t\in[0,1]$ for $t\in[n]$ are two sequences of random variables adapted to $\F_t$ such that $\E[Y_t \mid \F_{t-1}] = \E[X_t \mid \F_{t-1}]$. Then, for any $\sigma$-algebra $\calG$ such that $Y_t$ is measurable w.r.t. $\calG$ for all $t$, for any $\eta, \delta\in(0,1)$, with probability at least $1-\delta$,
    \begin{align*}
        \textstyle|\sum_{t\in[n]} (\E[X_t \mid \calG] - Y_t)| \leq  \frac{5\log(\nicefrac{2}{\delta})}{\eta} + \eta \sum_{t\in[n]} Y_t
    \end{align*}
\end{Theorem}

First, we will show how to use \cref{lem:tangent} to prove \cref{lem:algConnectionProperties}. Then, we give the proof of \cref{lem:tangent,thm:tangenConcentration}.

\begin{proof}[Proof of \cref{lem:algConnectionProperties}]
We note that Property 1 follows immediately from \cref{lem:tangent}, the tower rule of expectations, and linearity of expectations:
\begin{align*}
    \textstyle\sum_{i\in[\nsamples]}\E[f(\decisionsample, \sample_{\pi(i)})]
    = \sum_{i\in[\nsamples]}\E[Y_i]
    = \sum_{i\in[\nsamples]}\E[X_i]
    = \sum_{i\in[\nsamples]}\E[\E[X_i \mid \histend]]
    &= p\sum_{t\in[n]}\E[f(\decisiononline,\online_t)].
\end{align*}
For Property 2, we begin by noting that, applying \cref{lem:tangent}, the fact that $Y_i = f(\decisionsample, \sample_{\pi(i)})$ is $\histend$-measurable, and Jensen's inequality:
\begin{align*}
    \textstyle\E[(p\sum_{t\in[n]} \E[f(\decisiononline, \online_t) \mid \histend] - \sum_{i\in[\nsamples]}f(\decisionsample, \sample_{\pi(i)}))^2]
    &=\textstyle\E[(\sum_{i\in[\nsamples]} (\E[X_i \mid \histend] - Y_i))^2]\\
    &=\textstyle\E[(\E[\sum_{i\in[\nsamples]} (X_i - Y_i)\mid \histend] )^2]\\
    &\leq\textstyle\E[(\sum_{i\in[\nsamples]} (X_i - Y_i) )^2].
\end{align*}
Now, since we assume $f$ is bounded on $[0,1]$, $X_i, Y_i \in [0,1]$. Thus, since $X_i - Y_i$ are $\F_{i}$-measurable and $\E[X_i - Y_i \mid \F_{i-1}]=0$ by \cref{lem:tangent}, $X_i-Y_i \in [\pm 1]$ is a martingale difference sequence. Thus, we conclude that
\begin{align*}
    \textstyle\E[(\sum_{i\in[\nsamples]} (X_i - Y_i) )^2]
    = \E[\sum_{i\in[\nsamples]} (X_i - Y_i)^2]
    \leq \E[\sum_{i\in[\nsamples]} X_i + Y_i]
    &=\textstyle 2\sum_{i\in[\nsamples]} \E[Y_i]\\
    &=\textstyle 2\sum_{i\in[\nsamples]} \E[f(\decisionsample, \sample_{\pi(i)})],
\end{align*}
where we used the fact that $\E[X_i]=\E[Y_i]$ by \cref{lem:tangent}.

For Property 3, assuming $f$ is bounded on $[0,1]$, by \cref{lem:tangent} and \cref{thm:tangenConcentration}, with probability at least $1-2\delta$,
\begin{align*}
    \textstyle|p \sum_{t\in[n]} \E[f(\decisiononline, \online_t) \mid \histend] - \sum_{i\in[\nsamples]} f(\decisionsample, \sample_{\pi(i)})|
    &= \textstyle|\sum_{i\in[\nsamples]} (\E[X_i \mid \histend] - Y_i)|\\
    &\leq\textstyle \frac{5\log(\nicefrac{1}{\delta})}{\eta} + \eta \sum_{i\in[\nsamples]} Y_i\\
    &=\textstyle \frac{5\log(\nicefrac{1}{\delta})}{\eta} + \eta \sum_{i\in[\nsamples]} f(\decisionsample, \sample_{\pi(i)}). \qedhere
\end{align*}
\end{proof}

\paragraph{Establishing the tangent property and concequences: Proof of \cref{lem:tangent,thm:tangenConcentration}.}

\begin{proof}[Proof of \cref{lem:tangent}]
Notice that $X_i = f(\decisiononline[\pi(i)], \online_{\pi(i)})\1[i\not\in I] + f(\decisiongeneric[i], \online_{\hpi(i)})\1[i\in I]$ and $Y_i = f(\decisionsample, \sample_{\pi(i)})$ are $\F_i$-measurable, since all relevant random variables are sampled by the end of iteration $i$ of \cref{alg:equivGenericOffline}.
Now, conditioned on $i\in I$ and $\F_{i-1}$, $\hpi(i)$ and $\pi(i)$ are both (marginally) uniform on $U_i$, so, since the decisions are \OnlHistInd{} by \cref{obs:oblivious}, $f(\decisiongeneric[i], \online_{\hpi(i)})$ and $f(\decisiononline[\pi(i)], \online_{\pi(i)})$ are (conditionally) identically distributed. Thus, by definition of $X_i$, $X_i$ and $f(\decisiononline[\pi(i)], \online_{\pi(i)})$ are identically distributed, conditioned on $\F_{i-1}$. Therefore, since $\online_{\pi(i)}, \sample_{\pi(i)} \stackrel{iid}{\sim} \D_{\pi(i)}$ conditioned on $\F_{i-1}$, and again applying the \OnlHistInd{} property, $f(\decisionsample, \sample_{\pi(i)})$ and $f(\decisiononline[\pi(i)], \online_{\pi(i)})$ are identically distributed conditionally on $\F_{i-1}$. It thus follows that $X_i$ and $Y_i = f(\decisionsample, \sample_{\pi(i)})$ are identically distributed conditionally on $\F_{i-1}$, as claimed.

As a consequence,
\begin{align*}
    \E[X_i \mid \F_{i-1}]
    &= \E[f(\decisiononline[\pi(i)], \online_{\pi(i)}) \1[i\not\in I] + f(\decisiongeneric[i],\online_{\hpi(i)}) \1[i\in I] \mid \F_{i-1}]\\
    &= \E[f(\decisiononline[\pi(i)], \online_{\pi(i)}) \mid \F_{i-1}]\\
    &= \E[f(\decisionsample[\pi(i)], \sample_{\pi(i)}) \mid \F_{i-1}]
    = \E[Y_i \mid \F_{i-1}].
\end{align*}

For the final claim, begin by decomposing, using the tower rule of expectations:
\begin{align*}
    \textstyle\sum_{i\in[\nsamples]}\E[X_i \mid \histend]
    &=\textstyle \sum_{i\in[\nsamples]}\E[f(\decisiononline[\pi(i)], \online_{\pi(i)})\1[i\not\in I] \mid \histend]
    +\E[f(\decisiongeneric[i], \online_{\hpi(i)})\1[i\in I] \mid \histend]\\
    &=\textstyle \sum_{i\in[\nsamples]}\E[f(\decisiononline[\pi(i)], \online_{\pi(i)})\pr[i\not\in I \mid \online_{\pi(i)}, \decisiononline[\pi(i)], \histend] \mid \histend]\\
    &\quad\textstyle+\sum_{i\in[\nsamples]}\E[\E[f(\decisiongeneric[i], \online_{\hpi(i)}) \mid I, \histend]\1[i\in I] \mid \histend]
\end{align*}
For the first term, observe that
\begin{align*}
    \pr[i\not\in I \mid \online_{\pi(i)}, \decisiononline[\pi(i)], \histend]
    = \pr[i\not\in I]
    = p,
\end{align*}
since, by construction of \cref{alg:equivGenericOffline}, conditioned on any realization of $I$, the sample indices $\pi(1),\ldots,\pi(\nsamples)$ are sampled uniformly without replacement from $[n]$, so $(\pi(i))_{i\in[\nsamples]}$ is independent of $I$, and thus also the sample history $\histend$, online request $\online_{\pi(i)}$, and decision $\decisiononline[\pi(i)] = \decision[\pi(i)]{\online_{\pi(i)}}$ are independent of $I.$

For the second term, observe that, for any non-empty $I$, $\hpi(i)$ is uniform on $\samplesnot$ conditioned on $\histend$ and $I$, so on the event that $i\in I$:
\begin{align*}
    \E[f(\decisiongeneric[i], \online_{\hpi(i)}) \mid I, \histend]
    &=\textstyle \frac{1}{|\samplesnot|} \sum_{t\in\samplesnot} \E[f(\decision[\pi(i)]{\online_t}, \online_{t}) \mid I, \histend]\\
    &=\textstyle \frac{1}{|\samplesnot|} \sum_{t\in\samplesnot} \E[f(\decision[\pi(i)]{\online_t}, \online_{t}) \mid \histend],
\end{align*}
where the second equality follows since $I$ is independent of $\histend$, $(\online_t)_{t\in\samplesnot}$, and any additional randomness of $\blackbox$. Thus, since $\pr[i\in I \mid \histend] = \pr[i\in I] = 1-p$ by \cref{obs:equivOffline},
\begin{align*}
    \textstyle\sum_{i\in[\nsamples]}\E[X_i \mid \histend]
    &=\textstyle p\sum_{i\in[\nsamples]}\E[f(\decisiononline[\pi(i)], \online_{\pi(i)})\mid \histend]\\
    &\quad\textstyle+\sum_{i\in[\nsamples]}\pr[i\in I \mid \histend] \frac{1}{|\samplesnot|}\sum_{t\in\samplesnot}\E[f(\decision[\pi(i)]{\online_t}, \online_{t}) \mid \histend]\\
    &=\textstyle p\sum_{i\in[\nsamples]}\E[f(\decisiononline[\pi(i)], \online_{\pi(i)})\mid \histend]\\
    &\quad\textstyle+(1-p)\frac{\nsamples}{|\samplesnot|}\sum_{t\in\samplesnot}\E[\frac{1}{\nsamples}\sum_{i\in[\nsamples]}f(\decision[\pi(i)]{\online_t}, \online_{t}) \mid \histend]\\
    &=\textstyle p\sum_{i\in[\nsamples]}\E[f(\decisiononline[\pi(i)], \online_{\pi(i)})\mid \histend]\\
    &\quad\textstyle+p\sum_{t\in\samplesnot}\E[f(\decision[T_t]{\online_t}, \online_{t}) \mid \histend],
\end{align*}
where the last line follows since $|\samples|=\nsamples = p n$, $|\samplesnot| = (1-p)n$, and $T_t \sim \unif(\samples)$ (note that we treat the terms summing over $\samplesnot$ as $0$ in the case of $\samplesnot = \emptyset$, i.e., when $p=1$). Thus, using the definition of $\decisiononline$ for $t\in [n]\setminus\samples$ from \cref{alg:genericOnline}, we conclude that
\begin{align*}
    \textstyle\sum_{i\in[\nsamples]}\E[X_i \mid \histend]
    = p \sum_{t\in[n]} \E[f(\decisiononline, \online_t) \mid \histend],
\end{align*}
as claimed.
\end{proof}

\begin{proof}[Proof of \cref{thm:tangenConcentration}]

    This result follows by a similar proof as, e.g., \cite[Theorem 1]{beygelzimer2011contextual}. We begin by establishing that, for any $\eta \in (-1,1)$ and $c=\nicefrac{1}{5}$,
    \begin{align}\label{eq:concentrationMainStepAlt}
        \textstyle\E[\exp(c\eta \sum_{t\in[n]} \left( X_t - (1+\eta)Y_t \right))] \leq 1.
    \end{align}
    Given \eqref{eq:concentrationMainStepAlt}, the claim follows easily. Indeed, note that, by Jensen's inequality, and since $Y_t$ is measurable w.r.t. $\calG$,
    \begin{align*}
        \textstyle\E[\exp(c\eta \sum_{t\in[n]} (\E[X_t \mid \calG ] - (1+\eta)Y_t))]
        &=\textstyle\E[\exp(\E[c\eta \sum_{t\in[n]} \left( X_t - (1+\eta)Y_t \right) \mid \calG ] )]\\
        &\leq\textstyle\E[\E[\exp(c\eta \sum_{t\in[n]} \left( X_t - (1+\eta)Y_t \right) ) \mid \calG ]] ~\leq~ 1.
    \end{align*}
    Thus, we conclude that
    \begin{align*}
        \pr\Big[\sum_{t\in[n]} \E[X_t \mid \calG] - Y_t > \frac{\log(\nicefrac{1}{\delta})}{c\eta} + \eta \sum_{t\in[n]} Y_t\Big] \leq \delta.
    \end{align*}
    Repeating the same argument with $\eta$ replaced with $-\eta$, we similarly conclude
    \begin{align*}
        \pr\Big[\sum_{t\in[n]} \E[X_t \mid \calG] - Y_t < -\frac{\log(\nicefrac{1}{\delta})}{c\eta} - \eta \sum_{t\in[n]} Y_t \Big] \leq \delta.
    \end{align*}
    Thus, we focus on establishing \eqref{eq:concentrationMainStepAlt}. Using the fact that, by Jensen's inequality, for $z\in[0,1]$, $e^{\alpha z} \leq z e^\alpha + (1-z)$, we can bound:
    \begin{align*}
        \E[\exp(c\eta (X_t - (1+\eta)Y_t)) \mid \F_{t-1}]
        &\leq \E[(1 + X_t(\exp(c\eta) - 1))(1 + Y_t(\exp(-(1+\eta)c \eta) - 1)) \mid \F_{t-1}]\\
        &\leq \E[1 + X_t(\exp(c\eta) - 1) + Y_t(\exp(-(1+\eta)c\eta) - 1) \mid \F_{t-1}]\\
        &= 1 + \mu_t[(\exp(c\eta) - 1) + (\exp(-(1+\eta)c\eta) - 1)],
    \end{align*}
    where the second inequality follows since $X_t Y_t \in [0,1]$ and $(\exp(c\eta) - 1)(\exp(-(1+\eta)c\eta) - 1) \leq 0$ for $\eta \in (-1,1)$, and the equality follows since $\E[X_t \mid \F_{t-1}] = \mu_t = \E[Y_t \mid \F_{t-1}]$. Therefore, since $e^x \leq 1 + x + x^2$ for $x\leq 1$ and $\mu_t \geq 0$, the above is bounded as
    \begin{align*}
        \E[\exp(c\eta (X_t - (1+\eta)Y_t)) \mid \F_{t-1}]
        &\leq 1 + \mu_t[c\eta + c^2 \eta^2 - (1+\eta)c\eta + (1+\eta)^2 c^2 \eta^2]\\
        &= 1 - c \eta^2 \mu_t[1 - (2 + 2\eta + \eta^2)c]\\
        &\leq 1,
    \end{align*}
    where the last inequality used the fact that $\eta\leq 1$ and $c=\nicefrac{1}{5}$. Therefore, since $X_t, Y_t$ are $\F_t$-measurable, we can apply the above inequality recursively to conclude:
    \begin{align*}
       & \E[\exp(c\eta \sum_{t\in[n]}(X_t - (1+\eta)Y_t))]\\
        & \quad =\E[\exp(c\eta \sum_{t\in[n-1]}(X_t - (1+\eta)Y_t))\E[\exp(c\eta (X_n - (1+\eta)Y_n)) \mid \F_{n-1}]]
        \leq 1,
    \end{align*}
    as claimed.
\end{proof}

\subsection{Proof of \cref{lem:lpRescaling}}
\label{sec:appendix:oraDeferred}

The result essentially follows from Lemma 1 of \cite{KRTV-SICOMP18}. However, the statement of their result requires that $p \geq 2\sqrt{\frac{1+\log(m)}{B_{\min}}}$. While minor modifications to their argument can remove this requirement, we choose to prove this result here for completeness.

Throughout, we will assume without loss of generality that $s\geq 1$, since otherwise $s = p n = 0$, so either $p$ or $n$ is $0$ and the claim is trivially true. Additionally, in the arguments below, we condition on a fixed realization of $n$ requests $(\online_t)_{t\in[n]}$, and take expectations \emph{only} with respect to the uniformly random sample $\samples \subseteq [n]$ of size $\nsamples = p n$. Our final claim will then hold by taking expectations over these $n$ realized requests.

First, notice that:
\begin{align}\label{eq:scaledSampleLPVal}
    \optR[\samples]{p(1-2\eps)\B} \geq (1-2\eps)\optR[\samples]{p \B},
\end{align}
since for any optimal solution $(\zlp)_{t\in\samples, \decisionvar \in \decisionset}$ to \ref{program:ora-lp}$(p\B; \samples)$, $((1-2\eps)\zlp)_{t\in\samples,\decisionvar\in\decisionset}$ is feasible for \ref{program:ora-lp}$(p(1-2\eps)\B; \samples)$ and has value $(1-2\eps)\optR[\samples]{p\B}$. Thus, we will focus on lower bounding $\E[\optR[\samples]{p \B}]$.

Let $(\zstarlp)_{t\in\requests, \decisionvar\in\decisionset[t]}$ be an optimal solution to \ref{program:ora-lp}$((1-2\eps)\B; [n])$. We use this solution to define a feasible solution $(\zlp[i])_{i\in\samples, \decisionvar\in\decisionset[i]}$ to \ref{program:ora-lp}$(p \B; \samples)$ 
\begin{align*}
    \forall i \in \samples : 
    \zhlp[i] = \zstarlp[i]
    \quad\text{and}\quad
    \zlp[i] = \zstarlp[i] \1[i \leq \tau],
\end{align*}
where $\tau$ is the first time $\zhlp[i]$ saturates the budget of one of the resources, i.e.,
\begin{align*}
    \textstyle\tau = n \wedge \min\{i > 0 : \sum_{\substack{\ell\in\samples\\ \ell\leq i}}\sum_{\decisionvar\in\decisionset[\ell]} \allocatj[\ell] \zhlp[\ell] > p B_j - 1 \text{ for some } j\in[m]\}.
\end{align*}
First, notice that $\zlp[i]$ is feasible for \ref{program:ora-lp}$(p \B; \samples)$, since, by feasibility of $\zstarlp[i]$,
\begin{align*}
 \textstyle   \forall i\in\samples : 
    \sum_{\decisionvar\in\decisionset[i]} \zlp[i]
    ~\leq~ \sum_{\decisionvar\in\decisionset[i]} \zstarlp[i]
    ~\leq~ 1
\end{align*}
and, by construction of $\tau$ together with the fact that $\allocatj \in [0,1]$,
\begin{align*}
 \textstyle   \sum_{i\in \samples}\sum_{\decisionvar \in \decisionset[i]} \allocat[i] \zlp[i]
    ~=~ \sum_{i\in \samples: i\leq\tau}\sum_{\decisionvar \in \decisionset[i]} \allocat[i] \zstarlp[i]
    ~\leq~ p \B,
\end{align*}
so $\zlp$ is indeed feasible for \ref{program:ora-lp}$(p \B; \samples)$. Moreover, notice that the expected value this LP solution is:
\begin{equation}\label{eq:scaledLpValue}
\begin{aligned}
    \textstyle\E[\sum_{i\in\samples}\sum_{\decisionvar\in\decisionset[i]} \valat[i] \zlp[i]]
    &= \textstyle \sum_{t\in\requests}\sum_{\decisionvar\in\decisionset} \valat \zstarlp \pr[t\in\samples, \tau \geq t]\\
    &= \textstyle p\sum_{t\in\requests}\sum_{\decisionvar\in\decisionset} \valat \zstarlp \pr[\tau \geq t \mid t\in\samples],
\end{aligned}
\end{equation}
where, in the second line, we used the fact that $\pr[t\in\samples] = \nicefrac{\nsamples}{n} = p$ for all $t\in\requests$.
Thus, we conclude by upper-bounding $\pr[\tau < t \mid t\in\samples]$. Notice that we can upper bound this quantity as follows: For every $t\in \requests$, let $A_{t,j}^* = \sum_{\decisionvar\in\decisionset} \allocatj \zstarlp$. Then, we have
\begin{align*}
    \textstyle\pr[\tau < t \mid t\in \samples]
    \leq \pr[\tau < n \mid t\in \samples]
    &\leq\textstyle \sum_{j\in[m]}\pr[\sum_{i\in\samples} A_{i,j}^* > p B_j - 1 \mid t\in \samples]\\
    &\leq\textstyle \sum_{j\in[m]}\pr[\sum_{i\in\samples\setminus\{t\}} A_{i,j}^* > p B_j - 2 \mid t\in \samples],
\end{align*}
Where the last inequality follows from the fact that $A_{i,j}^* \in [0,1]$.
Now, conditioned on $t\in\samples$, the remaining $\nsamples - 1$ elements of $\samples$ are uniform over $[n]\setminus\{t\}$. Thus, by \Cref{cor:concentrationWithoutReplacement}, together with the fact that $\sum_{t\in\requests}A_{t,j}^*\leq (1-2\eps)B_j$ (by feasibility of $\zstarlp$), we have that, with probability at least $1-\nicefrac{\eps}{m}$, conditioned on $t\in\samples$, for any $j\in[m]$
\begin{align*}
    \sum_{i\in\samples\setminus\{t\}} A_{i,j}^* 
    \leq (1+\eps) \frac{\nsamples - 1}{n-1} \sum_{i\in\requests\setminus\{t\}} A_{i,j}^* + \frac{5\log(\nicefrac{2m}{\eps})}{\eps}
    &\leq (1+\eps)(1-2\eps) \frac{\nsamples-1}{n-1} B_j + \frac{5\log(\nicefrac{2m}{\eps})}{\eps}\\
    &\leq (1-\eps) p B_j + \frac{5\log(\nicefrac{2m}{\eps})}{\eps}\\
    &\leq p B_j - 2
\end{align*}
where, in the second line, we used the fact that $1\leq \nsamples \leq n$, and thus $\nicefrac{(\nsamples-1)}{(n-1)} \leq \nicefrac{\nsamples}{n}=p$, and the final line used the fact that $\eps p B_j \geq 2 + \frac{5\log(\nicefrac{2m}{\eps})}{\eps}$. Combining this observation with the previous, we conclude that
\[
    \textstyle\pr[\tau < t \mid t\in \samples]
    \leq \eps.
\]
Therefore, using \eqref{eq:scaledLpValue}, we conclude that
\begin{align*}
    \textstyle\E[\optR[\samples]{p\B} ]
    \geq\textstyle \E[\sum_{i\in\samples}\sum_{\decisionvar\in\decisionset[i]} \valat[i] \zlp[i]]
    &\geq\textstyle (1-\eps)p \sum_{t\in\requests}\sum_{\decisionvar\in\decisionset} \valat \zstarlp\\
    &=\textstyle (1-\eps)p \optR{(1-2\eps)\B}\\
    &\geq\textstyle (1-\eps)(1-2\eps)p \optR{\B},
\end{align*}
where the final inequality follows from the observation that, if $\zlp'$ is an optimal solution for \ref{program:ora-lp}$(\B; [n])$, then $(1-2\eps)\zlp'$ is feasible for \ref{program:ora-lp}$((1-2\eps)\B; [n])$.
Combining this with \eqref{eq:scaledSampleLPVal}, we conclude that
\begin{align*}
    \textstyle\E[\optR[\samples]{p(1-2\eps)\B} ]
    \geq\textstyle (1-\eps)(1-2\eps)^2 p \optR{\B}
    &\geq\textstyle (1-5\eps) p \optR{\B},
\end{align*}
as claimed.

\subsection{Proof of \cref{lem:packingCoveringRO}}
\label{sec:packingCoveringRO}

Here, we describe the modifications to the proof of Theorem 4.1 from \cite{GM-MOR16} to obtain the guarantee stated in \cref{lem:packingCoveringRO}. While their algorithm handles more general Packing-Covering Multiple-Choice LPs with an objective, we will focus on the simpler feasibility version of the problem, as this is the version needed for our result. We outline below how to modify their arguments, which hold for LPs with ``full simplex'' constraints (i.e., the constraints $\sum_{\decisionvar\in\decisionset} \xip \leq 1$ instead of $\sum_{\decisionvar\in\decisionset} \xip = 1$ in \ref{program:opc-lb}) to instead handle the simplex constraints in \ref{program:opc-lb}. For simplicity of notation, we will denote the scaled packing and covering constraints specified in \cref{lem:packingCoveringRO} as $\Bpackingjt = \alpha \Bpackingj$ and $\Bcoveringjt=\beta\Bcoveringj$. There are two main modifications:

\textbf{Modification 1: Load Balancing request construction.}
Their \ro{} algorithm (Algorithm 2 in their paper) reduces the packing-covering LP problem to a generalized load-balancing instance with possibly negative loads (for which they give a \ro{} algorithm, see their Theorem 3.1) of the following form: Each request at time $t\in[n]$ for $2m$ resources can take a decision $\decisionvar\in\decisionset \cup \{\phi\}$, where $\phi$ denotes the null decision, and has (in our notation) load of the form:
the first $m$ resources $j\in[m]$ have loads $\loadatj = \nicefrac{\packingatj}{\Bpackingjt}$ and $\loadatthisj{\phi}=0$ for all $\decisionvar\in\decisionset$, and the last $m$ resources have loads $\loadatthisj[t,j+m]{\decisionvar} = \nicefrac{2}{n} - \nicefrac{\coveringatj}{\Bcoveringjt}$ and $\loadatthisj[t,j+m]{\phi}=\nicefrac{2}{n}$ for all $\decisionvar\in\decisionset$.

In our setting, since we do not assume the option of null requests, we modify the instance to have requests for $2m+1$ resources, where each request is of the form:
\begin{align}\label{eq:loadRequestsPCReduction}
    \textstyle\loadatj = \frac{\packingatj}{\Bpackingjt},
    \quad
    \textstyle\loadatthisj[t,j+m]{\decisionvar} = \frac{2}{n} - \frac{\coveringatj}{\Bcoveringjt},
    \quad
    \textstyle\loadatthisj[t,2m+1]{\decisionvar} = \frac{1}{n} 
    \quad
    \forall j\in[m], \decisionvar \in \decisionset,
\end{align}
Under this construction, notice that the following analogue of their Claim 4.2 (the first step in the proof of their Theorem 4.1) holds:
\begin{Claim}[Analogue of Claim 1, \cite{GM-MOR16}]\label{claim:loadBalancingReductionMakespan}
    Consider any Packing-Covering requests for which \ref{program:opc-lb} is feasible. Let \eqref{eq:loadRequestsPCReduction} be the associated (possibly negative) load requests, and $\lambda^*$ be the optimal makespan of \ref{program:olb-lb}.
    Then, $\lambda^* = 1$.
\end{Claim}
\begin{proof}
Let $(\xip)_{t\in[n], \decisionvar\in\decisionset}$ be any feasible solution to \ref{program:opc-lb}. Then, by feasibility of the solution and the definition of load balancing requests in \eqref{eq:loadRequestsPCReduction}, $\sum_{t\in[n]}\sum_{\decisionvar\in\decisionset} \loadatj \xip \leq 1$ for all $j\in [2m]$.
Moreover, every feasible solution to \ref{program:olb-lb} $\xprimeip$ (which, in particular, includes the solution $\xip$ to \ref{program:opc-lb}) further satisfies:
$\sum_{t\in[n]}\sum_{\decisionvar\in\decisionset} \loadatthisj[t,2m+1]{\decisionvar} \xprimeip
= \sum_{t\in[n]}\sum_{\decisionvar\in\decisionset} \frac{1}{n} \xprimeip
= 1.$
Thus, noting that the makespan is defined as the largest load (i.e., not the largest \emph{magnitude} load), we conclude that the optimal makespan of this instance must be $\lambda^* = 1$.
\end{proof}
It is also easy to verify that their Claim 2 (the second step in the proof of their Theorem 4.1) continues to hold under the modified loads from \eqref{eq:loadRequestsPCReduction}.
\begin{Claim}[Claim 2, \cite{GM-MOR16}]\label{claim:loadBalancingReductionWellBounded}
    Given any feasible Packing-Covering instance \ref{program:opc-lb}, the corresponding Load Balancing instance with requests \eqref{eq:loadRequestsPCReduction} is $(M, 2)$-well-bounded for $M^{-1}=\nicefrac{1}{2}\min_{j\in[m]} \Bpackingjt \wedge \Bcoveringjt = \Omega(\eps^{-2}\log(\nicefrac{m}{\delta}))$ (as in Definition 2 from \cite{GM-MOR16}), i.e., every load balancing requests $\loadatj \in [-M, M]$ and, for each $j\in[2m+1]$ either $\loadatj \geq -\nicefrac{2\lambda^*}{n}$ for all $t\in[n], \decisionvar\in\decisionset$ or $\loadatj \leq \nicefrac{2\lambda^*}{n}$ for all $t\in[n], \decisionvar\in\decisionset$, where $\lambda^*$ is the optimal value of \ref{program:olb-lb} with requests from \eqref{eq:loadRequestsPCReduction}.
\end{Claim}
\begin{proof}
    By definition of the load balancing requests \eqref{eq:loadRequestsPCReduction}, and by assumption that $\Bpackingjt = \Omega(\eps^{-2}\log(\nicefrac{m}{\delta}))$ and $\packingatj\in[0,1]$, it follows that, for each $j\in[m]$, $\loadatj = \nicefrac{\packingatj}{\Bpackingjt} \in [0, M]$. Further, since $\Bcoveringjt = \Omega(\eps^{-2}\log(\nicefrac{m}{\delta}))$ and $\coveringatj\in[0,1]$, it follows that for each $j\in[m]$, $\loadatthisj[t,j+m]{\decisionvar} = \nicefrac{2}{n} - \nicefrac{\coveringatj}{\Bcoveringjt} \in [-M, \nicefrac{2}{n}]$. Further, $\loadatthisj[t,2m+1]{\decisionvar} = \nicefrac{1}{n}$. Note also that, without loss, $n\geq \Bcoveringjt$ for all $j\in[m]$, since otherwise \ref{program:opc-lb} cannot be feasible, and thus $\nicefrac{2}{n} \leq \nicefrac{2}{\Bcoveringjt} \leq M$. Thus, since $\lambda^*=1$ by \cref{claim:loadBalancingReductionMakespan}, the claim follows.
\end{proof}

\textbf{Modification 2: No rescaling of the multiplicative weights update solution.}
The final step of their \ro{} Packing-Covering algorithm (Phase 3 in \cite[Algorithm 2]{GM-MOR16}) rescales the solution $\xprimeip$ returned by their expertLB procedure by a multiplicative factor $(1-\calO(\eps))$ in order to guarantee that the packing constraints are satisfied \emph{exactly}. Without this rescaling, however, the solution still satisfies the packing constraints $(1+\eps)\Bpackingjt$. Thus, the appropriate modification to their Algorithm 2 is the following:

Run their Online Generalized Load-Balancing \ro{} algorithm (Algorithm 1, expertLB, with the multiplicative weights update algorithm from their Theorem 2.1) on the generalized load balancing requests from \eqref{eq:loadRequestsPCReduction}, with scale parameter $M = 2\max_{j\in[m]} (\Bpackingjt \wedge \Bcoveringjt)^{-1} = \calO(\nicefrac{\eps^{2}}{\log(\nicefrac{m}{\delta})})$ and error parameter $\eps'=c_1 \eps$ (for a sufficiently large universal constant $c_1$) to obtain a solution $\xip$ to \ref{program:olb-lb} online, and take decision $\decisiononline=\decisionvar$ with probability $\xip$ for the Packing-Covering requests.

\textbf{Concluding the proof.}
With this modification to their algorithm, and using the modified load balancing instance from \eqref{eq:loadRequestsPCReduction} and \cref{claim:loadBalancingReductionMakespan,claim:loadBalancingReductionWellBounded}, the proof follows easily. Indeed, by their Theorem 3.1, and since the instance is $(M,2)$-well-bounded by \cref{claim:loadBalancingReductionWellBounded} and $\lambda^* = 1 \geq \nicefrac{3 M \log(\nicefrac{2m}{\delta})}{(\eps')^2}$ by \cref{claim:loadBalancingReductionMakespan} and choice of $M,\eps'$, the decisions returned by their Algorithm 1 satisfies, with probability at least $1-\nicefrac{\delta}{2}$,
\begin{align*}
    \textstyle\forall j\in[2m+1]:
    \sum_{t\in[n]}\sum_{\decisionvar\in\decisionset} \loadatj \xip - \sum_{t\in[n]}\eps' |\sum_{\decisionvar\in\decisionset}\loadatj \xip|
    \leq 1+2 c_2 \eps',
\end{align*}
where $\lambda^*$ is the optimal value of \ref{program:olb-lb} for requests from \eqref{eq:loadRequestsPCReduction} and $c_2$ is a positive universal constant. Plugging in the request definitions from \eqref{eq:loadRequestsPCReduction} and using the fact that $\sum_{\decisionvar\in\decisionset}\xip = 1$, it follows immediately that $\sum_{t\in[n]}\sum_{\decisionvar\in\decisionset} \packingatj \xip \leq (1+\calO(\eps))\Bpackingjt$ and $\sum_{t\in[n]}\sum_{\decisionvar\in\decisionset} \coveringatj\xip \geq (1-\calO(\eps)) \Bcoveringjt$ for all $j\in[m]$.
Since the algorithm selects decision $\decisiononline = \decisionvar$ with probability $\xip$, and since the fractional decisions $\xip$ do not depend on this rounding step, Bernstein's inequality (\cref{cor:concentrationWithoutReplacement}) then implies that, with probability at least $1-\nicefrac{\delta}{2}$, $\sum_{t\in[n]} \packingonlinej \leq (1+\calO(\eps))\Bpackingjt$ and $\sum_{t\in[n]} \coveringonlinej \geq (1-\calO(\eps))\Bcoveringjt$ for all $j\in[m]$ (since $\Bpackingjt \wedge \Bcoveringjt = \Omega(\eps^{-2}\log(\nicefrac{m}{\delta}))$). Thus, by rescaling $\eps$ appropriately, the claimed result follows. 
\subsection{Proof of \cref{lem:optMakespanConcentration}}
\label{sec:loadBalancingConcentration}
In the \preview{} setting, each request is sampled independently from a distribution $\D_t$, so the solution to \ref{program:olb-lb} is random. Thus, it will be convenient to instead reason about the following fluid LP relaxation of \ref{program:olb-lb}:

\begin{equation}\tag{$\LPlbrelax$}\label{program:olb-lb-relax}
    \begin{aligned}
        \optRlbfl ~:=~ &\text{minimize}  &&\textstyle \| \sum_{t\in[n]} \E_{\online_t}[\sum_{\decisionvar\in\decisionset} \loadat \xlp] \|_{\max}\\
        &\text{s.t. }\forall t\in [n], \online_t\in\typespace,  && \textstyle\sum_{\decisionvar \in \decisionset}\xlp[t] = 1\\
        &\forall t\in[n], \online_t\in\typespace, \decisionvar \in \decisionset, && \xlp \geq 0
    \end{aligned}
\end{equation}
By construction of the fluid LP \ref{program:olb-lb-relax}, notice that 
\begin{align*}
    \optRlbfl \leq \E[\optRlb],
\end{align*}
since, if $(\xstarip)_{t\in[n], \decisionvar\in\decisionset}$ is an optimal solution to \ref{program:olb-lb}, then $\xlpat{\online} = \E[\xstarip \mid \online_t = \online]$ is feasible for \ref{program:olb-lb-relax} with value $\|\sum_{t\in[n]} \E_{\online_t}[\sum_{\decisionvar \in \decisionset}\loadat \xstarip] \|_{\max} \leq \E[\|\sum_{t\in[n]} \sum_{\decisionvar \in \decisionset}\loadat \xstarip \|_{\max}] = \E[\optRlb]$.

Further, taking $(\xstarlpat{\online})_{t\in[n],\online\in\typespace,\decisionvar\in\decisionset}$ to be an optimal solution to \ref{program:olb-lb-relax} and $\tilde{L}_{t,j} = \sum_{\decisionvar \in \decisionsetsample} \loadsampleatj \xstarlpat{\sample_t}$ to be the corresponding load on resource $j$ corresponding to sample $\sample_t$, we have that
$\optRlb[\samples] \leq \max_{j\in[m]}\sum_{t\in\samples} \tilde{L}_{t,j}$. Now, since $\tilde{L}_{t,j} \in [0,1]$ by feasibility, and since $\samples$ is a uniformly random size-$pn$ subset of $[n]$ by the \preview{} model, we may apply Bernstein's inequality for sampling without replacement (\cref{cor:concentrationWithoutReplacement}) to obtain that, for any $\eta,\delta\in(0,1)$, with probability at least $1-\nicefrac{\delta}{4}$,
\begin{align*}
    \textstyle \optRlb[\samples] 
    \leq \max_{j\in[m]}\sum_{t\in\samples} \tilde{L}_{t,j} \leq 3\eta^{-1}\log(\nicefrac{8 m}{\delta}) 
    + p (1+\eta) \max_{j\in[m]}\sum_{t\in[n]} \tilde{L}_{t,j}.
\end{align*}
Finally, since $\sample_t,\online_t\sim\D_t$, and hence also the loads $\tilde{L}_{t,j}$, are sampled independently for every $t\in[n]$ in the \preview{} model, and $\tilde{L}_{t,j} \in [0,1]$, Bernstein's inequality (specifically, we use the version in \cref{cor:concentrationWithoutReplacement}) implies that, with probability $1-\nicefrac{\delta}{4}$,
\begin{align*}
    \max_{j\in[m]} \sum_{t\in[n]} \tilde{L}_{t,j} 
    \leq 3\eta^{-1}\log(\nicefrac{8 m}{\delta}) + (1+\eta)\max_{j\in[m]} \sum_{t\in[n]} \E[\tilde{L}_{t,j}]
    = 3\eta^{-1}\log(\nicefrac{8 m}{\delta}) + (1+\eta)\optRlbfl.
\end{align*}
Combining these three inequalities and setting $\eta = \nicefrac{\eps}{3}$, we conclude, with probability at least $1-\nicefrac{\delta}{2}$, 
\begin{align}\label{eq:lbOptConcUpper}
    \optRlb[\samples]
    \leq 21\eps^{-1} \log(\nicefrac{8m}{\delta}) 
    + p(1+\eps) \optRlbfl
    \leq 21\eps^{-1} \log(\nicefrac{8m}{\delta}) 
    + p(1+\eps) \E[\optRlb],
\end{align}
which establishes the upper bound.

For the lower bound, by integrating \eqref{eq:lbOptConcUpper} (with $p=1$, i.e., $\samples = [n]$), we have that, denoting $D = 21\eta^{-1}\log(8m) + (1+\eta)\optRlbfl$ for $\eta\in(0,1)$,
\begin{align*}
    \textstyle\E[\optRlb] 
    \leq D + \int_0^\infty \pr[\optRlb > D + t] dt
    &\textstyle\leq D + \int_0^\infty \exp(-\nicefrac{\eta t}{21}) dt\\
    &\textstyle= 21\eta^{-1}(\log(8m) + 1) + (1+\eta)\optRlbfl
\end{align*}
Now, by the minimax theorem, the dual characterization of \ref{program:olb-lb} and \ref{program:olb-lb-relax} is
\begin{align*}
    \textstyle\optRlb[\samples] = \max_{y\in\Delta_m} \sum_{t\in\samples} g_t(y, \sample_t)
    \quad\text{and}\quad
    \optRlbfl = \max_{y\in\Delta_m} \sum_{t\in[n]} \E_{\sample_t}[g_t(y, \sample_t)],
\end{align*}
where $\triangle_m = \{y \in \R^m_{\geq 0} : \sum_{j\in[m]} y_j = 1\}$ and $g_t(y, \sample_t) = \min_{ \xipdot \in \triangle_{|\decisionsetsample|}} \langle y, \sum_{\decisionvar \in \decisionsetsample} \allocsampleat \xip \rangle \in [0,1]$.
Thus, taking $y^*$ to be an optimal dual solution to \ref{program:olb-lb-relax}, and since $\samples$ is a uniformly random sample of $[n]$ of size $\nsamples = p n$, by Bernstein's inequality for sampling without replacement (\cref{cor:concentrationWithoutReplacement}), for any $\eta,\delta\in(0,1)$, with probability at least $1-\nicefrac{\delta}{4}$,
\begin{align*}
    \textstyle\optRlb[\samples] 
    \geq \sum_{t\in\samples} g_t(y^*, \sample_t)
    \geq -3\eta^{-1}\log(\nicefrac{8}{\delta}) + (1-\eta) p \sum_{t\in[n]} g_t(y^*, \sample_t),
\end{align*}
and since $(\sample_t)_{t\in[n]}$ are sampled independently, by Bernstein's inequality (\cref{cor:concentrationWithoutReplacement}), with probability at least $1-\nicefrac{\delta}{4}$,
\begin{align*}
    \textstyle\sum_{t\in[n]} g_t(y^*, \sample_t)
    \geq -3\eta^{-1}\log(\nicefrac{8}{\delta}) + (1-\eta) \sum_{t\in[n]} \E_{\sample_t}[g_t(y^*,\sample_t)]
    = -3\eta^{-1}\log(\nicefrac{8}{\delta}) + (1-\eta) \optRlbfl.
\end{align*}
Combining these inequalities and choosing $\eta=\nicefrac{\eps}{3}$, we conclude that, with probability at least $1-\nicefrac{\delta}{2}$,
\begin{align*}
    \optRlb[\samples]
    \geq - 18\eps^{-1} \log(\nicefrac{8}{\delta}) + p (1-\nicefrac{\eps}{3})^2 \optRlbfl
    \geq -144 \eps^{-1}\log(\nicefrac{8 m}{\delta}) + p (1-\eps)\E[\optRlb].
\end{align*}

{\small
\bibliographystyle{alpha}
\bibliography{ref.bib}
}

\appendix
\end{document}